\pdfoutput=1
\RequirePackage{ifpdf}
\ifpdf 
\documentclass[pdftex]{sigma}
\else
\documentclass{sigma}
\fi

\usepackage{braket}
\usepackage{cancel}

\numberwithin{equation}{section}

\newtheorem{Theorem}{Theorem}[section]
\newtheorem*{Theorem*}{Theorem}
\newtheorem{Corollary}[Theorem]{Corollary}

\newtheorem{Proposition}[Theorem]{Proposition}
 { \theoremstyle{definition}
\newtheorem{Definition}[Theorem]{Definition}

 }

\begin{document}
\allowdisplaybreaks

\newcommand{\arXivNumber}{2401.00500}

\renewcommand{\PaperNumber}{061}

\FirstPageHeading

\ShortArticleName{Deformation Quantization with Separation of Variables of $G_{2,4}(\mathbb{C})$}

\ArticleName{Deformation Quantization\\ with Separation of Variables of $\boldsymbol{G_{2,4}(\mathbb{C})}$}

\Author{Taika OKUDA~$^{\rm a}$ and Akifumi SAKO~$^{\rm b}$}

\AuthorNameForHeading{T.~Okuda and A.~Sako}

\Address{$^{\rm a)}$~Graduate School of Science, Department of Mathematics and Science Education,\\
\hphantom{$^{\rm a)}$}~Tokyo University of Science, 1-3 Kagurazaka, Shinjuku-ku, Tokyo 162-8601, Japan}
\EmailD{\href{mailto:warriors.aboot@gmail.com}{warriors.aboot@gmail.com}}
\URLaddressD{\url{https://sites.google.com/view/taiokutus/}}

\Address{$^{\rm b)}$~Faculty of Science Division II, Department of Mathematics, Tokyo University of Science,\\
\hphantom{$^{\rm b)}$}~1-3 Kagurazaka, Shinjuku-ku, Tokyo 162-8601, Japan}
\EmailD{\href{mailto:sako@rs.tus.ac.jp}{sako@rs.tus.ac.jp}}
\URLaddressD{\url{https://www.rs.tus.ac.jp/sako/}}

\ArticleDates{Received December 10, 2024, in final form July 14, 2025; Published online July 23, 2025}

\Abstract{We construct a deformation quantization with separation of variables of the Grassmannian $G_{2,4}(\mathbb{C})$. A star product on $G_{2,4}(\mathbb{C})$ can be explicitly determined as the solution of the recurrence relations for $G_{2,4}(\mathbb{C})$ given by Hara and one of the authors (A.~Sako). To provide the solution to the recurrence relations, it is necessary to solve a system of linear equations in each order. However, to give a concrete expression of the general term is not simple because the variables increase with the order of the differentiation of the star product. For this reason, there has been no formula to express the general term of the recurrence relations. In this paper, we overcome this problem by transforming the recurrence relations into simpler ones. We solve the recurrence relations using creation and annihilation operators on a Fock space. From this solution, we obtain an explicit formula of a star product with separation of variables on $G_{2,4}(\mathbb{C})$.}


\Classification{14M15; 32Q15; 46L87; 53D55}

\section{Introduction}
\label{G24_introduction}

In this paper, we construct a noncommutative complex Grassmannian $G_{2,4}(\mathbb{C})$ explicitly.
In gauge theories with background magnetic fields or in high-energy physics where gravity is quantized, 
noncommutative manifolds or quantized spacetime appear naturally. $G_{2,4}(\mathbb{C})$ appears in various situations in physics. As an example of the motivation for this paper, let us take the twistor theory here.

The twistor theory proposed by Penrose was introduced in an attempt to quantize gravity theories \cite{Pen1,Pen2,Pen3}. Although it does not directly correspond to the original purpose, 
Penrose's twistor theory is useful for constructing solutions of (anti-)self-dual Yang--Mills equations, Einstein's equations, etc. For example, it is known that solutions of (anti-)self-dual Yang--Mills equations correspond to holomorphic vector bundles over the complex projective space $\mathbb{C}P^{3}$ \cite{Pen1,Pen2,Pen3,Ward1,WaWe}.
For more detailed discussions, see \cite{MW,WaWe}. This correspondence is understood as the twistor (Klein) correspondence \smash{$G_{2,4}(\mathbb{C})\xleftarrow{\pi_{2}}F_{1,2,4}(\mathbb{C})\xrightarrow{\pi_{1}}\mathbb{C}P^{3}$}. Noncommutative deformation of the (anti-)self-dual Yang--Mills theory is known as a successful example of noncommutative deformations of integrable systems.
See, for example, \cite{Hama1,HaNa1,NS} and references therein.
It is also expected that it makes sense to construct a noncommutative deformation of the twistor correspondence \smash{$G_{2,4}(\mathbb{C})\xleftarrow{\pi_{2}}F_{1,2,4}(\mathbb{C})\xrightarrow{\pi_{1}}\mathbb{C}P^{3}$},
where $F_{1,2,4}(\mathbb{C})$ is the complex flag manifold defined by
$F_{1,2,4}(\mathbb{C})
:=\{(V_{1},V_{2})\mid
V_{1}
\subset
V_{2}
\subset
\mathbb{C}^{4},\,
\dim_{\mathbb{C}}
V_{k}
=
k,\,
k=1,2
\} $.
To achieve this, it is essential to construct a deformation quantization of $G_{2,4}(\mathbb{C})$. Noncommutative $\mathbb{C}P^{3}$ was already constructed. However, their explicit ones have not yet been determined for general complex Grassmannians. The quantizations of $\mathbb{C}P^{3}$ and $G_{2,4}(\mathbb{C})$ may contribute not only to the development of noncommutative twistor theory but also to the successful quantization of gravity.
For this reason, it is worth constructing noncommutative $G_{2,4}(\mathbb{C})$ explicitly in physics.

For a symplectic manifold, the quantization method called ``deformation quantization" is one of the construction methods of noncommutative differentiable manifolds. It is known as the quantization method based on a deformation for a Poisson algebra. Two types of deformation quantizations are known: ``formal deformation quantization" proposed by Bayen et al.~\cite{BFFLS} and ``strict deformation quantization", based on $C^{\star}$-algebra proposed by Rieffel \cite{Rief1,Rief2,Rief3}. In this paper, we use ``deformation quantization" in the sense of ``formal deformation quantization".
\begin{Definition}\label{dq_def}
 Let
 $
 (
 M,
 \{
 \cdot,
 \cdot
 \}
 )
 $
 be a Poisson manifold and
 $
 C^{\infty}
 (M)
 [\!
 [
 \hbar
 ]
 \!]
 $
 be the ring of formal power series over
 $
 C^{\infty}(M).
 $ Let $\ast$ be a star product denoted by
 $
 f
 \ast
 g
 =
 \sum_{k}
 C_{k} (f,g )
 \hbar^{k}
 $
 satisfying the following conditions:
 \begin{enumerate}\itemsep=0pt
 \item[(1)]
 For any
 $ f,g,h \in C^{\infty} (M)[\![ \hbar] \!]$, $ f \ast( g \ast h) = ( f \ast g) \ast h$.
 \item[(2)]
 For any
 $
 f
 \in
 C^{\infty}
 (M)
 [\!
 [
 \hbar
 ]
 \!]
 $,
 $
 f
 \ast
 1
 =
 1
 \ast
 f
 =
 f.
 $
 \item[(3)]
 Each $C_{k}(\cdot,\cdot)$ is a bidifferential operator.

 \item[(4)]
 $
 C_{0}
 (
 f,
 g
 )
 =
 fg$,
 $
 C_{1}
 (
 f,
 g
 )
 -
 C_{1}
 (
 g,
 f
 )
 =
 \{f,g\} $.
 \end{enumerate}
 A pair
 $
 (
 C^{\infty}
 (M)
 [\!
 [
 \hbar
 ]
 \!],
 \ast
 )
 $
 is called ``a deformation quantization" for the Poisson manifold $M$.
\end{Definition}

A well-known example of a star product is the Moyal product. It was constructed independently by Groenewold \cite{Groe} and Moyal \cite{Moya}, and is often called Groenewold-Moyal product. Since it is a star product on $\mathbb{R}^{2N}$, it gives the noncommutative $\mathbb{R}^{2N}$ as a deformation quantization. A star product on $\mathbb{C}^{N}$, called Wick--Voros product, is also known as a star product that is equivalent to the Moyal product~\cite{Voro}. For more detailed reviews of deformation quantization, see for example \cite{GUT,Mosha}. For any symplectic manifolds, a construction methods of deformation quantization is given by de Wilde--Lecomte \cite{DL}, Omori--Maeda--Yoshioka \cite{OMY1} and Fedosov~\cite{Fedo}. More generally, a deformation quantization was constructed by Kontsevich for Poisson manifolds~\cite{Kont1}. Recently, for any contact manifolds, its deformation quantization was constructed by Elfimov--Sharapov~\cite{ES}.

On the other hand, construction method of a deformation quantization for any K\"ahler manifold was studied by Karabegov \cite{Kara2,Kara1}. Let $M$ be an $N$-dimensional K\"ahler manifold and
$(U,\phi)$ be a chart of $M$. For
$
p
\in
M
$, we choose
$
\phi(p)
=
 \bigl(z^{1}_{p},\dots,z^{N}_{p} \bigr)
$
as local coordinates at
$
p
\in
M
$.
To use the K\"ahler potential of $M$, we assume that $U$ is contractible in the following discussion. In the following discussion, we omit
$
p
\in
M
$
when denoting local coordinates of $U$. For the K\"ahler manifold $M$, the K\"ahler 2-form $\omega$ and the K\"ahler metric $g$ can be locally expressed as
\begin{align*}
 \omega
 =
 i
 g_{k\bar{l}}
 {\rm d}z^{k}
 \wedge
 {\rm d}\bar{z}^{l},
 \qquad
 g_{k\bar{l}}
 =
 \partial_{k}
 \partial_{\bar{l}}
 \Phi
\end{align*}
by using a K\"ahler potential $\Phi$, where
$
\partial_{k}
:=
\partial/
\partial
z^{k}
$
and
$
\partial_{\bar{l}}
:=
\partial/
\partial
z^{\bar{l}}
=
\partial/
\partial
\overline{z}^{l}
$.
Note that we use the Einstein summation convention for index $k$ and $\bar{l}$ on the above.
We also denote the inverse matrix of $ (g_{\overline{k}l} )$ by $ \bigl(g^{l\overline{k}} \bigr)$.
Karabegov proposed a construction method of deformation quantization for K\"ahler manifolds such that the following conditions, called ``separation of variables", are satisfied~\cite{Kara2,Kara1}.
\begin{Definition} \label{sep_of_bar_starprd}
 Let $M$ be an $N$-dimensional K\"ahler manifold. A star product $\ast$ on $M$ is ``separation of variables" if the following two conditions are satisfied for any open set $
 U
 \subset
 M
 $
 and
 $
 f
 \in
 C^{\infty}
 (
 U
 )
 $:
 \begin{enumerate}\itemsep=0pt
 \item[(1)]
 For any holomorphic function
 $ a $
 on $U$,
 $ a
 \ast
 f
 =
 af $.
 \item[(2)]
 For any anti-holomorphic function
 $ b $
 on $U$,
 $
 f
 \ast
 b
 =
 fb
 $.
 \end{enumerate}
\end{Definition}

Here we introduce two local differential operators. The differential operators
$
 D^{k}
$
and
$
 D^{\overline{k}}
$
are defined by
$
 D^{k}
 :=
 g^{k\overline{l}}
 \partial_{\overline{l}}$,
$
 D^{\overline{k}}
 :=
 g^{\overline{k}l}
 \partial_{l}$,
where we use the Einstein summation convention for index~$l$ and $\bar{l}$.
Note that \smash{$ g^{k\overline{l}}= g^{\overline{l}k}$}.
We define the set of differential operators
\begin{align*}
 {\mathcal S}
 :=
 \bigg\{
 A \,
 \bigg|\,
 A
 =
 \sum_{\vec{\beta}^{\ast}}
 a_{\vec{\beta}^{\ast}}
 D^{\vec{\beta}^{\ast}},\,
 a_{\vec{\beta}^{\ast}}
 \in
 C^{\infty}
 (U)
 \bigg\}.
\end{align*}
Here,
$
D
:=
 \bigl(
D^{\bar{1}},
\dots,
D^{\bar{N}}
 \bigr)
$
and
$
 \vec{\beta}
 =
 (
 \beta_{1},
 \dots,
 \beta_{N}
 )
$ is a multi-index, i.e., \smash{$D^{\vec{\beta}^{\ast}}$} is expressed as
\smash{$
 D^{\vec{\beta}^{\ast}}
 =
 \bigl(
 D^{\bar{1}}
 \bigr)^{\beta_{1}}
 \cdots
 \bigl(
 D^{\bar{N}}
 \bigr)^{\beta_{N}}
$}.
We also define the differential operators
\smash{$
 D^{\vec{\alpha}_{n}}
$}
and
\smash{$
 D^{\vec{\beta}_{n}^{\ast}}
$}
as
\begin{gather*}
 D^{\vec{\alpha}_{n}}
 :=
 D^{\alpha_{1}^{n}}\cdots D^{\alpha_{N}^{n}},
 \qquad
 D^{\alpha_{k}^{n}}
 :=
 \bigl(
 D^{k}
 \bigr)^{\alpha_{k}^{n}},
 \\
 D^{\vec{\beta}_{n}^{\ast}}
 :=
 \overline{
 D^{\vec{\beta}_{n}}
 }
 =
 \overline{D^{\beta_{1}^{n}}}
 \cdots
 \overline{D^{\beta_{N}^{n}}},
 \qquad
 \overline{D^{\beta_{k}^{n}}}
 :=
 \bigl(
 D^{\overline{k}}
 \bigr)^{\beta_{k}^{n}}
\end{gather*}
for
\begin{align*}
 \vec{\alpha}_{n},
 \vec{\beta}_{n}
 \in
 \left\{
 \bigl(\gamma_{1}^{n},\dots,\gamma_{N}^{n} \bigr)
 \in
 \mathbb{Z}^{N}
 \,\bigg|\,
 \sum_{k=1}^{N}
 \gamma_{k}^{n}
 =
 n
 \right\}.
\end{align*}
The sum
\smash{$\vec{\alpha}_{n}+
 \vec{\beta}_{m}$} is defined in $\mathbb{Z}^{N} $ as usual.

We define
$ D^{\vec{\alpha}_{n}}:=0$
when there exists at least one negative
$ \alpha_{k}^{n}
 \notin\mathbb{Z}_{\geq0}$
for
$
 k
 \in
\{1,\dots,N\}
$.
Similarly,
\smash{$
 D^{\vec{\beta}_{n}^{\ast}}:=0
$}
if \smash{$\vec{\beta}_{n}^{\ast}$} has negative components.
For
$
 f
 \in
 C^{\infty}
 (U)
$,
we can construct the left $\ast$-multiplication operator $L_{f} $ with respect to $f$
such that
$ L_{f} g := f \ast g $. The following theorem states that $L_{f}$ is expressed by a formal power series with respect to the differential operators in~$\mathcal{S}$.

\begin{Theorem}[\cite{Kara2,Kara1}]\label{thm_Kara_star_opr}
Let $M$ be an $N$-dimensional K\"ahler manifold
with the K\"ahler form~$\omega$.
 There exists a unique star product $\ast$ on $M$ such that for any contractible open neighborhood $U$ on
 $M$ with a K\"ahler potential $\Phi$ of the form $\omega$ and any
 $
 f
 \in
 C^{\infty}
 (U)
$, the left $\ast$-multiplication operator $L_{f}$ satisfies the following three conditions:
\begin{enumerate}\itemsep=0pt
\item[$(1)$] $[ L_{f}, R_{\partial_{\overline{l}}\Phi}] = 0$,
where
$
 R_{\partial_{\overline{l}}\Phi}
 =
 \partial_{\overline{l}}\Phi
 +
 \hbar
 \partial_{\overline{l}}
$
is the right $\ast$-multiplication operator with respect to $\partial_{\overline{l}}\Phi$.
\item[$(2)$]
$
 L_{f}1
 =
 f
 \ast
 1
 =
 f
$.
\item[$(3)$]
For any $
 g,h
 \in
 C^{\infty}
 (U)
$,
the left $\ast$-multiplication operator is associative, i.e.,
\begin{align*}
 L_{f} ( L_{g} h ) = f \ast ( g \ast h ) = ( f \ast g ) \ast h = L_{ L_{f}g } h .
\end{align*}
\end{enumerate}
\end{Theorem}

More generally, Karabegov constructed a unique separation of variables type star product~$\ast_{\omega_{\hbar}}$ for the formal deformation $\omega_{\hbar}=\omega+\hbar\omega_{1}+h^{2}\omega_{2}+\cdots$ for the K\"ahler form $\omega$. $\omega_{\hbar}$ is called a Karabegov form of $\ast_{\omega_{\hbar}}$, denoted by $kf(\ast_{\omega_{\hbar}})$. The formal K\"ahler potential $\Phi_{\hbar}$ is given by~${\Phi_{\hbar}=\Phi+\hbar\Phi_{1}+\hbar^{2}\Phi_{2}+\cdots}$ using each K\"ahler potential $\Phi, \Phi_{1}, \Phi_{2},\dots$. Conversely, given a~star product $\ast$ for a K\"ahler manifold $M$ and its K\"ahler form $\omega$, there exists a unique Karabegov form $kf(\ast)=\omega_{\hbar}=\omega+\hbar\omega_{1}+h^{2}\omega_{2}+\cdots$ such that $\ast=\ast_{\omega_{\hbar}}$ \cite{Kara5}. In particular, the star product for a Karabegov form with no higher-order form terms, i.e., $\omega_{\hbar}=\omega$, is called the standard product. We note that the aim in our study is to give an explicit standard product for the K\"ahler form $\omega$, which contains no higher-order formal terms.

The final goal of this study is to construct a deformation quantization with separation of variables for the complex Grassmannian $G_{p,p+q}(\mathbb{C})$. Noncommutative complex Grassmannians have been studied from several viewpoints. From fuzzy physics, noncommutative Grassmannians as fuzzy manifolds have been obtained by Dolan--Jahn \cite{DJ} and Halima--Wurzbacher \cite{HW}. From the viewpoint of deformation quantization, Karabegov has studied deformation quantization with separation of variables as an example for pseudo-K\"ahler manifold~\cite{Kara4}.
In \cite{Kara4}, however, in the sense that each term of the expansion of the star product is written out explicitly, no explicit expression is given for a star product with separation of variables for $G_{p,p+q}(\mathbb{C})$.
On the other hand, it is known that there is a previous study that has obtained an explicit star product on~$ G_{p,p+q}(\mathbb{C})$~\cite{Schir} using a phase space reduction proposed by Bordemann et al.~\cite{BBEW}. As another approach to the above background, there are known construction methods that explicitly give a~star product with separation of variables on locally symmetric K\"ahler manifolds. These methods were proposed by one of the authors of this paper (A.~Sako), Suzuki, and Umetsu \cite{SSU1,SSU2} and Hara and Sako~\cite{HS1,HS2}.
From these methods, explicit star products with separation of variables on any locally symmetric K\"ahler manifolds were obtained for the $1$- and $2$-dimensional cases~\cite{HS1,HS2,OS_Varna,OS_jrnl}.
Using these construction methods, we can determine a star product with separation of variables by solving the recurrence relations. In previous works \cite{HS1,HS2}, the recurrence relations for $G_{2,4}(\mathbb{C})$ were obtained via the construction methods. In this work, we construct the explicit star product with separation of variables by solving the recurrence relations given by \cite{HS1,HS2}.\looseness=-1

We now introduce the complex Grassmannian. The complex Grassmannian $G_{p,p+q}(\mathbb{C})$ is defined by
\begin{align*}
 G_{p,p+q}(\mathbb{C})
 :=
\bigl\{ V \subset \mathbb{C}^{p+q} \mid V \mbox{ is $p$-dimensional complex vector subspace}
\bigr\}.
\end{align*}
In particular if $p=1$, $G_{1,1+q}(\mathbb{C})$ is complex projective space $\mathbb{C}P^{q}$. The simplest complex Grassmannian which is not $\mathbb{C}P^{N}$ is known to be $G_{2,4}(\mathbb{C})$, i.e., $p=q=2$.
We introduce the local coordinates of $G_{p,p+q}(\mathbb{C})$ using the following procedure \cite{KN}. Let $\bigl(z^{1}, z^{2} ,\dots ,z^{p+q}\bigr)$ be a natural local coordinate of $\mathbb{C}^{p+q}$, and $S$ be the $p$-dimensional subspace of $\mathbb{C}^{p+q}$ such that~${z^{\alpha_{1}},\dots, z^{\alpha_{p}}}$ are linearly independent. Here $z^{\alpha_{1}}|_{S},\dots, z^{\alpha_{p}}|_{S}$ are the local coordinates restricted to $S$, and~${\alpha_{1},\dots, \alpha_{p}}$ are the integers such that $1\leq\alpha_{1}<\cdots<\alpha_{p}\leq p+q$. We take the subset $U_{\alpha}\subset G_{p,p+q}(\mathbb{C})$ defined by
\begin{align*}
U_{\alpha}
:=
\{
S
\subset
\mathbb{C}^{p+q}
\mid
\dim\nolimits_{\mathbb{C}}S
=
p,\,
z^{\alpha_{1}}|_{S},\dots, z^{\alpha_{p}}|_{S}
\ \text{are linearly independent}
\}
\end{align*}
for the set $\alpha:=\{\alpha_{1},\dots, \alpha_{p}\}$. We can take $\binom{p+q}{p}$ patterns of such $U_{\alpha}$, and they give the open coverings of $G_{p,p+q}(\mathbb{C})$. In the following discussion, we consider the case $\alpha=\{1,\dots,p\}$, and denote $U_{\alpha}=U_{\{1,\dots,p\}}$ simply as $U$.

Let us choose a local coordinate of $G_{p,p+q}(\mathbb{C})$ using
$
U
\subset
G_{p,p+q}(\mathbb{C})
$
introduced above. Let~${M^{\ast}(p+q,p;\mathbb{C})}$ be the set of
$
(p+q)
\times
p
$
matrices of rank $p$, and
$
\pi\colon
M^{\ast}(p+q,p;\mathbb{C})
\to
G_{p,p+q}(\mathbb{C})
$
be a smooth projection. ${\rm GL}_{p}(\mathbb{C})$ acts freely and transitively on $M^{\ast}(p+q,p;\mathbb{C})$ from the right. Therefore, $M^{\ast}(p+q,p;\mathbb{C})$ is a principal ${\rm GL}_{p}(\mathbb{C})$-bundle over $G_{p,p+q}(\mathbb{C})$. We now take a set
\begin{align*}
 \pi^{-1}(U)
:={} &
 \left\{
 Y
 =
 \left(
 \begin{array}{c}
 Y_{0}\\
 Y_{1}
 \end{array}
 \right)
 \in
 M^{\ast}(p+q,p;\mathbb{C})
 \,\bigg|\,
 Y_{0}
 \in
 {\rm GL}_{p}
 (
 \mathbb{C}
 ),\,
 Y_{1}
 \in
 M
 (
 q,p
 ;
 \mathbb{C}
 )
 \right\}\\
 \subset{}&
 M^{\ast}(p+q,p;\mathbb{C})
 .
\end{align*}
We consider a map
$
 \phi
\colon
 \pi^{-1}(U)
 \to
 M(q,p;\mathbb{C})
 \cong\mathbb{C}^{qp}
$
defined by
$
 Y
 \mapsto
 \phi
 (Y)
 :=
 Y_{1}
 Y_{0}^{-1}
$
for
$
Y
\in
\pi^{-1}(U)
$.
$\phi$ is a ${\rm GL}_{p}(\mathbb{C})$-invariant, i.e., for any
$
Y
\in
\pi^{-1}(U)
$
and
$
A
\in
{\rm GL}_{p}(\mathbb{C})
$,
$\phi(YA)=\phi(Y)$. By using this $\phi$, we can choose
$
 Z
 =
 \bigl(
 z^{I}
 \bigr)
 =
 \bigl(
 z^{ii'}
 \bigr)
 =
 \phi
 (Y)
$
as a local coordinate of~$U
\subset
G_{p,p+q}(\mathbb{C})
$.
A more detailed discussion of the local coordinates of $G_{p,p+q}(\mathbb{C})$ is given, for example, in~\cite{KN}.
Here $I=ii'$ is ``the capital letter index'', which consists of two indices $i$ and $i'$. In general,
\begin{align*}
 I = ii' = 11', \dots, 1p', \dots, q1', \dots, qp'
 \end{align*}
for
$i\in \{1,\dots,q \}$
and
$i'\in \{1',\dots,p' \}$.
For example,
$
 I
 =
 ii'
 =
 11',
 12',
 13',
 21',
 22',
 23'
$
for~${
q=2}$,
$p=3
$.
If $p=q=2$, then the other index which is not $i$ (or not $i'$) is uniquely determined when $i$
(or $i'$) is fixed. We denote by $\cancel{i}$ (or $\cancel{i}'$) the index such that it is not $i$ (or not $i'$). For example, if $I=12'$, then
$
i\cancel{i}'=11'$,
$\cancel{i}i'=22'
$
and
$
\cancel{I}
=
21'
$.
By using this notation, we can choose~\smash{$
 \bigl(
 z^{I}
 \bigr)
 :
 =
 \bigl(
 z^{I},
 z^{\cancel{I}},
 z^{i\cancel{i}'},
 z^{\cancel{i}i'}
 \bigr)
$}
as the local coordinates of $G_{2,4}(\mathbb{C})$ up to their order. The set of capital letter indices
$
 (
 \{
 I=ii'
\mid
 1
 \leq
 i
 \leq
 q,\,
 1'
 \leq
 i'
 \leq
 p'
 \},
 \leq_{C}
 )
$
can be identified with the set of ordinary indices
$
 (
 \{
 1,\dots,qp
 \},
 \leq
 )
$
as a totally ordered set. By this identification, a capital letter index can also be regarded as an ordinary index
$
 I
 =
 1,
 \dots,
 qp
$.
See Appendix~\ref{appendix_properties_Gpq} for more details.
For $p=q=2$, the set $\mathcal{I}$ of capital letter indices is defined by
$
\mathcal{I}
:=
\{
I,
\cancel{I},
i\cancel{i}',
\cancel{i}i'
\}
$
for a fixed $I$.
Note that it can also be rewritten as
$
\mathcal{I}
=
\{
11',
12',
21',
22'
\}
$.
The K\"ahler potential of $G_{p,p+q}(\mathbb{C})$ is given by
$
 \Phi
 =
 \log
 \det
 B$,
where
$
 B
 =
 \mathrm{Id}_{p}
 +
 Z^{\dagger}
 Z
$.
When we express
$
 B
 =
 (
 b_{\overline{i'}j'}
 )
$,
in terms of its element each entry $b_{\overline{i'}j'}$ is written as
\begin{align}
 b_{\overline{i'}j'}
 =
 \delta_{i'j'}
 +
 \overline{z}^{mi'}
 z^{mj'}
 =
 \delta_{i'j'}
 +
 z^{\overline{mi'}}
 z^{mj'}.
 \label{b_entry}
\end{align}
Since $B$ is regular, its inverse matrix $B^{-1}$ is given by
\smash{$
 B^{-1}
 =
 \bigl(
 b^{i'\overline{j'}}
 \bigr)
 =
 \frac{1}{
 \det
 B
 }
 \bigl(
 \widetilde{B}_{\overline{j'}i'}
 \bigr)$},
where $\widetilde{B}$ is an adjugate matrix of $B$, and
$
 \bigl(
 \widetilde{B}_{\overline{j'}i'}
 \bigr)
$
is the transpose matrix of
$
 \bigl(
 \widetilde{B}_{i'\overline{j'}}
 \bigr)
$.
By using the above $\Phi$, the K\"ahler metric $g_{I\overline{J}}$ can be expressed as
$
 g_{I\overline{J}}
 =
 \partial_{I}
 \partial_{\overline{J}}
 \Phi
$.
Here $\partial_{I}=\partial_{ii'}$ and $\partial_{\overline{J}}=\partial_{\overline{jj'}}$ are defined by
\smash{$
\partial_{I}
:=
\partial/
\partial
z^{I}
=
\partial/
\partial
z^{ii'}
$}
and
\smash{$
\partial_{\overline{J}}
:=
\partial/
\partial
\overline{z}^{J}
=
\partial/
\partial
\overline{z}^{jj'}
$}.
Some useful properties of~${G_{p,p+q}(\mathbb{C})}$ are summarized in Appendix~\ref{appendix_properties_Gpq}.

In order to give an explicit star product with separation of variables on $G_{2,4}(\mathbb{C})$ by using the construction method proposed by \cite{HS1,HS2}, it is necessary to solve the recurrence relations. However, if we attempt to solve the recurrence relations sequentially order by order, we need to solve a system of linear equations for each order. In addition, the recurrence relations contain variables given by the four partitions of order $n$. The number of variables increases combinatorially with increasing $n$. Therefore, there has been no formula for the general term of the recurrence relations. In this work, we derive a method that
determines the general term of the recurrence relations for $G_{2,4}(\mathbb{C})$ without solving the system of linear equations. Furthermore, we~explicitly determine a star product with separation of variables on $G_{2,4}(\mathbb{C})$.
In other words, we obtain a deformation quantization with separation of variables for $G_{2,4}(\mathbb{C})$ based on the solution of the recurrence relations given by~\cite{HS1,HS2}. The main theorem of this paper for an explicit star product with separation of variables on $G_{2,4}(\mathbb{C})$ is as follows.

\begin{Theorem}[main result]\label{thm_star_prd_G2_4}
For $f,g\in C^{\infty}(G_{2,4}(\mathbb{C}))$, a star product with separation of variables on $G_{2,4}(\mathbb{C})$ is given by
\begin{align}
 f \ast g
={}&
 \sum_{n=0}^{\infty}
 \sum_{
 \substack{
 J_{i}
 \in
\{
 J_{i}
\}_{n}\\
 Y_{i}
 \in
\{
 Y_{i}
\}_{n}
 }
 }
 \sum_{
 \substack{
 k_{i}=1\\
 k_{i}
 \in
\{
 k_{i}
\}_{n}
 }
 }^{2}
 \left(
 \prod_{l=1}^{n}
 \frac{
 g_{\overline{k_{l}j_{l}'},Y_{l}}
 \Upsilon_{
 l,
 \{
 J_{i}
 \}_{n},
 \{
 k_{i}
 \}_{n}
 }
 }{\tau_{l}}
 \right)
 \left(
 \prod_{
 S
 \in
 \mathcal{I}
 }
 \prod_{r=1}^{n}
 \theta
 \left(
 \sum_{m=1}^{r-1}
 d_{
 S,
 J_{m},
 k_{m}
 }
 \right)
 \right)
 \nonumber
 \\
 &
 \times
 \bigl(
 D^{
 \sum_{m=1}^{n}
 \vec{e}_{{Y_{m}}}
 }
 f
 \bigr)
 \bigl(
 D^{
 \sum_{
 X
 \in
 \mathcal{I}
 }
 \sum_{m=1}^{n}
 d_{
 X,
 J_{m},
 k_{m}
 }
 \vec{e}_{X}^{\ast}
 }
 g
 \bigr),
 \label{star_prd_G2_4}
\end{align}
where
\begin{gather*}
 \Upsilon_{
 l,
 \{
 J_{i}
 \}_{n},
 \{
 k_{i}
 \}_{n}
 }
 :=
 \frac{
 \tau_{l}
 \delta_{j_{l}k_{l}}
 +
 1
 +
 \sum_{m=1}^{l}
 d_{
 m,
 j_{l}\cancel{j_{l}}',
 J_{m},
 k_{m}
 }
 }{
 l
 (
 \tau_{l}+1
 )
 +
 2
 \bigl\{
 \sum_{m=1}^{l}
 (
 \delta_{IJ_{m}}
 +
 \delta_{\cancel{i}i',J_{m}}
 )
 \bigr\}
 \bigl\{
 \sum_{m=1}^{l}
 (
 \delta_{\cancel{I}J_{m}}
 +
 \delta_{i\cancel{i}',J_{m}}
 )
 \bigr\}
 },
 \\
 d_{
 S,
 J_{m},
 k_{m}
 }
 :=
 \delta_{
 S,
 J_{m}
 }
 +
 \delta_{\cancel{j_{m}}k_{m}}
 (
 \delta_{
 S,
 \cancel{J_{m}}
 }
 -
 \delta_{
 S,
 j_{m}\cancel{j_{m}}'
 }
 )
\end{gather*}
for $l,r=1,\dots,n$,
$
\tau_{l}
 :=
 1
 -
 l
 +
 \frac{1}{\hbar}
$,
\smash{$
 \sum_{
 J_{i}
 \in
 \{
 J_{i}
 \}_{n}
 }
 :=
 \sum_{
 J_{1}
 \in
 \mathcal{I}
 }
 \cdots
 \sum_{
 J_{n}
 \in
 \mathcal{I}
 }
$},
\smash{$
 \sum_{
 Y_{i}
 \in
 \{
 Y_{i}
 \}_{n}
 }
 :=
 \sum_{
 Y_{1}
 \in
 \mathcal{I}
 }
 \cdots
 \sum_{
 Y_{n}
 \in
 \mathcal{I}
 }
$},
\smash{$
 \sum_{
 \substack{
 k_{i}=1
 \\
 k_{i}
 \in
 \{
 k_{i}
 \}_{n}
 }
 }^{2}
 :=
 \sum_{k_{1}=1}^{2}
 \cdots
 \sum_{k_{n}=1}^{2}
$},
and
$\theta\colon \mathbb{R}\to\{0,1\}$
is the step function defined by
\begin{align*}
 \theta(x)
 :=
 \begin{cases}
 1, & x\geq0,
 \\
 0, & x<0.
 \end{cases}
\end{align*}
\end{Theorem}

It may appear that
$\Upsilon_{l,\{J_{i}\}_{n},\{k_{i}\}_{n}}$
depends on the capital letter index $I\in\mathcal{I}$. Note, however, that in fact, it does not depend on the choice of $I\in\mathcal{I}$. Proposition~\ref{prop_app_sym_fct} stated in Appendix~\ref{appendix_symmetry_function} guarantees this $I$-independence. From this fact, the star product \eqref{star_prd_G2_4} is the expression which is independent of $I\in\mathcal{I}$.

This paper is organized mainly into four sections and three appendices. In Section~\ref{G24_review}, we~review the previous works by
\cite{HS1,HS2} related to a deformation quantization with separation of variables for a locally symmetric K\"ahler manifold. In Section~\ref{G24_DQ_sep_var}, we give the solvable recurrence relations which gives a star product with separation of variables for general $G_{p,p+q}(\mathbb{C})$. We focus on the simplest complex Grassmannian $G_{2,4}(\mathbb{C})$ which is not $\mathbb{C}P^{N}$. We give the solution of the recurrence relations for $G_{2,4}(\mathbb{C})$ by using a linear operator on a Fock space. Furthermore, we construct an explicit star product with separation of variables on~$G_{2,4}(\mathbb{C})$ from the solution. In Section~\ref{G24_sol_n012}, we concretely give the explicit lower-order solutions of the recurrence relations for confirmation. In Section~\ref{G24_summary}, we summarize our work. In addition, we present a future work for noncommutative deformation of twistor correspondence using an explicit star product on~$G_{2,4}(\mathbb{C})$.
In Appendix~\ref{appendix_properties_Gpq}, we summarize useful properties of $G_{p,p+q}(\mathbb{C})$. In Appendix~\ref{appendix_symmetry_function}, we denote a property of symmetric functions whose variables are capital letter indices. By using this property, we state that some functions and operators appearing in this paper do not depend on the choice of a capital letter index $I$. Finally, we show that the obtained star product on~$G_{2,4}(\mathbb{C})$ is the expression which is independent of $I$. In Appendix~\ref{appendix_coeff_T}, we show that the explicit form of solutions is actually recovered from a linear operator.\looseness=1

\section[Construction method of deformation quantization for locally symmetric Kahler manifolds]{Construction method of deformation quantization\\ for locally symmetric K\"ahler manifolds}
\label{G24_review}

A general definition of quantization was first proposed by Berezin \cite{Bere1,Bere2}. In particular, Berezin also constructed the quantization of K\"ahler manifolds in the case of phase space by using symbol algebras. Bordemann et al.\ also studied the quantization of K\"ahler manifolds via Toeplitz quantization \cite{BMS}. For the case of compact K\"ahler manifolds, its Berezin--Toeplitz quantization was studied by Karabegov and Schlichenmaier \cite{KS1,Schli1,Schli3,Schli2,Schli4,Schli5}. See the review~\cite{Schli6,Schli7} summarized by Schlichenmaier for their previous works related to Berezin--Toeplitz quantization of compact K\"ahler manifolds. From other perspectives of deformation quantization, deformation quantizations for K\"ahler manifolds were provided by Moreno~\cite{Mor1,Mor2}, Omori--Maeda--Miyazaki--Yoshioka~\cite{OMMY} and Reshetikhin--Takhtajan~\cite{RT}. It is also known that the relations between Berezin quantization and deformation quantization were studied by Cahen--Gutt--Rawnsley \cite{CGR1,CGR2,CGR3,CGR4}. As a different approach to Moreno, Omori--Maeda--Miyazaki--Yoshioka and Reshetikhin--Takhtajan, a deformation quantization with separation of variables for K\"ahler manifolds was proposed by Karabegov \cite{Kara2,Kara1}. The formula for Karabegov's deformation quantization given by interpreting the graph as a bidifferential operator was proposed by Gammelgaard~\cite{Gam}. After that, a deformation quantization with separation of variables for locally symmetric K\"ahler manifolds was studied by one of the authors of this paper (A.~Sako), Suzuki, and Umetsu \cite{SSU1,SSU2} and Hara and Sako \cite{HS1,HS2}, inspired by the construction method by Karabegov. Deformation quantization of K\"ahler manifolds is applied for modern physics. Using Karabegov's deformation quantization, Fock representations on noncommutative K\"ahler manifolds constructed by Sako, Suzuki and Umetsu as a recipe to construct gauge theories on noncommutative K\"ahler manifolds \cite{SSU1,SSU2,SU1,SU2,SU3}. For the case of complex Grassmannians, their Fock representations are studied by authors \cite{OS_Gpq_F}. Gauge theories on noncommutative K\"ahler manifolds were studied by Maeda, Sako, Suzuki and Umetsu \cite{MSSU,Sako}. In particular, for~$\mathbb{C}P^{N}$, Sako, Suzuki and Umetsu analyzed exact solutions corresponding to Bogomol'ny equation \cite{SSU3}. As an application to solitons, scalar noncommutative multi-solitons in noncommutative scalar field theory on K\"ahler manifolds were constructed and analyzed their stability conditions for homogeneous ones by Spradlin--Volovich~\cite{SV}.

In this section, we review the construction method of a deformation quantization with separation of variables for locally symmetric K\"ahler manifolds proposed by Hara and Sako~\cite{HS1,HS2}. A~K\"ahler manifold $M$ is ``locally symmetric" if its Riemann curvature tensor satisfies~$
\nabla_{\partial_{E}}
{R_{ABC}}^{D}\allowbreak
=
0
$
for
$
A,
B,
C,
D,
E
\in
\bigl\{
1,\dots,N,
\overline{1},\dots,\overline{N}
\bigr\}
$.
In this paper, we define the Riemann curvature tensor
$
R^{\nabla}
\colon
\Gamma(TM)
\times
\Gamma(TM)
\times
\Gamma(TM)
\to
\Gamma(TM)
$
on $M$ for vector fields~${
X,Y
\in
\Gamma(TM)
}$~by
\begin{align*}
R^{\nabla}
 (X,Y )
:=
\nabla_{X}
\nabla_{Y}
-
\nabla_{Y}
\nabla_{X}
-
\nabla_{[X,Y]}.
\end{align*}
We denote the component of $R^{\nabla}$ by
$
R^{\nabla}
 (
\partial_{A},
\partial_{B}
 )
\partial_{C}
:=
{R_{ABC}}^{D}
\partial_{D}
$
as a local expression. Here, we note that the notation ${R_{ABC}}^{D}$ used in this paper can be expressed by the relation
$
{R_{ABC}}^{D}
=
{\mathfrak{R}^{D}}_{CAB}
$
using the notation
\begin{align}
R^{\nabla}
 (
\partial_{A},
\partial_{B}
 )
\partial_{C}
:=
{\mathfrak{R}^{D}}_{CAB}
\partial_{D}
\label{Riem_curv_KN}
\end{align}
by Kobayashi--Nomizu \cite{KN}. Note that
${R_{ABC}}^{D}=-{\mathfrak{R}_{ABC}}^{D}$.

Let $M$ be an $N$-dimensional locally symmetric K\"ahler manifold, and $(U,\phi)$ be a chart of $M$. We now assume the following form for a star product with separation of variables on $M$ for any~${ f,g \in C^{\infty} ( U )}$:
\begin{align}
 f\ast g
 =
 L_{f}g
 :=
 \sum_{n=0}^{\infty}
 \sum_{
 \vec{\alpha}_{n},
 \vec{\beta}_{n}^{\ast}
 }
 T_{\vec{\alpha}_{n},\vec{\beta}_{n}^{\ast}}^{n}
 \bigl(
 D^{\vec{\alpha}_{n}}f
 \bigr)
 \bigl(
 D^{\vec{\beta}_{n}^{\ast}}g
 \bigr).
 \label{star_expand}
\end{align}
Here, recalling that
$
 D^{\vec{\alpha}_{n}},
 D^{\vec{\beta}_{n}^{\ast}}
$
are differential operators defined by
\begin{align*}
 &
 D^{\vec{\alpha}_{n}}
 :=
 D^{\alpha_{1}^{n}}\cdots D^{\alpha_{N}^{n}},
 \qquad
 D^{\alpha_{k}^{n}}
 :=
 (
 D^{k}
 )^{\alpha_{k}^{n}},
 \\
 &
 D^{\vec{\beta}_{n}^{\ast}}
 :=
 \overline{
 D^{\vec{\beta}_{n}}
 }
 =
 \overline{D^{\beta_{1}^{n}}}
 \cdots
 \overline{D^{\beta_{N}^{n}}},
 \qquad
 \overline{D^{\beta_{k}^{n}}}
 :=
 \bigl(
 D^{\overline{k}}
 \bigr)^{\beta_{k}^{n}},
 \\
 &
 \vec{\alpha}_{n},
 \vec{\beta}_{n}
 \in
 \left\{
 (\gamma_{1}^{n},\dots,\gamma_{N}^{n} )
 \in
 \mathbb{Z}^{N}
 \,\middle|\,
 \sum_{k=1}^{N}
 \gamma_{k}^{n}
 =
 n
 \right\},
\end{align*}
respectively.
Note that $\hbar$ is included in the coefficient \smash{$T\raisebox{1.5pt}{${}_{\vec{\alpha}_{n},\vec{\beta}_{n}^{\ast}}^{n}$}$} for $n\geq1$. The coefficients
\smash{$
 T\raisebox{1.5pt}{${}_{\vec{\alpha}_{n},\vec{\beta}_{n}^{\ast}}^{n}$}
$}
can be assumed to be covariantly constants by using the fact that $M$ is a locally symmetric space. Here, we define
\smash{$
 T\raisebox{1.5pt}{${}_{\vec{\alpha}_{n},\vec{\beta}_{n}^{\ast}}^{n}$}
 :=
 0
$}
when
$
 \vec{\alpha}_{n}
 \notin
 \mathbb{Z}_{\geq0}^{N}
$
or
\smash{$
 \vec{\beta}_{n}
 \notin
 \mathbb{Z}_{\geq0}^{N}
$}.
For $n=0,1$, it is known that~the coefficient
$
 T_{\vec{\alpha}_{n},\vec{\beta}_{n}^{\ast}}^{n}
$
are explicitly given as follows.

\begin{Proposition}[\cite{HS1,HS2}]\label{prop:covar_const_0and1}
 For a star product with separation of variables $\ast$ on $U,$
 \begin{align}
 T_{\vec{0},\vec{0}^{\ast}}^{0}
 =
 1,
 \qquad
 T_{\vec{e}_{i},\vec{e}_{j}^{\ast}}^{1}
 =
 \hbar
 g_{i\overline{j}},
 \label{init_cond_01}
 \end{align}
 \noindent
 where
 $
 \vec{e}_{i}
 =
 (\delta_{1i},\dots,\delta_{Ni} ).
 $
\end{Proposition}

Proposition~\ref{prop:covar_const_0and1} states that the coefficients for $n=0,1$ are completely determined for any $N$-dimensional locally symmetric K\"ahler manifold.
In particular,
the fact ``\smash{$T_{\vec{0},\vec{0}^{\ast}}^{0}=1$}'' plays the role of the initial condition for the recurrence relations which
gives the star product to be shown later.
The following theorem proposed by 
\cite{HS1,HS2} describes the existence of a star product with separation of variables such that it is given by \eqref{star_expand}.

\begin{Theorem}[\cite{HS1,HS2}] \label{theorem:star_prd}
 Let $M$ be an $N$-dimensional locally symmetric K\"ahler manifold. For~${f,g
 \in
 C^{\infty}(M)}
 $,
 there exists locally a star product with separation of variables $\ast$ such that
 \[
 f
 \ast
 g
 =
 \sum_{n=0}^{\infty}
 \sum_{\vec{\alpha}_{n},\vec{\beta}_{n}^{\ast}}
 T_{\vec{\alpha}_{n},\vec{\beta}_{n}^{\ast}}^{n}
 \bigl(
 D^{\vec{\alpha}_{n}}f
 \bigr)
 \bigl(
 D^{\vec{\beta}_{n}^{\ast}}g
 \bigr),
 \]
 where the coefficient
 \smash{$
 T_{
 \vec{\alpha}_{n},
 \vec{\beta}_{n}^{\ast}
 }^{n}
 $}
 satisfies the following recurrence relations:
 \begin{align}
 \sum_{d=1}^{N}
 \hbar
 g_{\overline{i}d}
 T_{\vec{\alpha}_{n}-\vec{e}_{d},\vec{\beta}_{n}^{\ast}-\vec{e}_{i}^{\ast}}^{n-1}
 ={}&
 \beta_{i}^{n}
 T_{\vec{\alpha}_{n},\vec{\beta}_{n}^{\ast}}^{n}
 +
 \sum_{k=1}^{N}
 \sum_{\rho=1}^{N}
 \hbar
 \binom{\beta_{k}^{n}-\delta_{k\rho}-\delta_{ik}+2}{2}
 R_{\overline{\rho}\overline{i}}^{\overline{k}\overline{k}}
 T_{\vec{\alpha}_{n},\vec{\beta}_{n}^{\ast}-\vec{e}_{\rho}^{\ast}+2\vec{e}_{k}^{\ast}-\vec{e}_{i}^{\ast}}^{n}
 \nonumber
 \\
 &
 +
 \sum_{k=1}^{N-1}
 \sum_{l=1}^{N-k}
 \sum_{\rho=1}^{N}
 \hbar
 (
 \beta_{k}^{n}
 -
 \delta_{k\rho}
 -
 \delta_{ik}
 +
 1
 )
 (
 \beta_{k+l}^{n}
 -
 \delta_{k+l,\rho}
 -
 \delta_{i,k+l}
 +
 1
 ) \nonumber\\
 &\phantom{+}{}\times
 R_{\overline{\rho}\overline{i}}^{\overline{k+l}\overline{k}}
 T_{
 \vec{\alpha}_{n},
 \vec{\beta}_{n}^{\ast}
 -
 \vec{e}_{\rho}^{\ast}
 +
 \vec{e}_{k}^{\ast}
 +
 \vec{e}_{k+l}^{\ast}
 -
 \vec{e}_{i}^{\ast}
 }^{n}.
 \label{HS_recrel}
 \end{align}
 Here
 $
 \binom{a}{b}
 $
 is a binomial coefficient.
\end{Theorem}

Note that these recurrence relations in Theorem~\ref{theorem:star_prd} are equivalent to the \cite[equations~(6.9)]{SSU1}. To obtain an explicit star product with separation of variables on any $N$-dimensional locally symmetric K\"ahler manifold, we have to solve the recurrence relations in Theorem~\ref{theorem:star_prd} with the initial condition
\smash{$
T_{\vec{0},\vec{0}^{\ast}}^{0}
=
1
$}
(in \eqref{init_cond_01} in Proposition~\ref{prop:covar_const_0and1}). However, to determine the general term~\smash{$
 T_{\vec{\alpha}_{n},\vec{\beta}_{n}^{\ast}}^{n}
$}
satisfying this system of recurrence relations is not easy, except for the one-dimensional case.
For 
arbitrary $N$-dimensional locally symmetric K\"ahler manifold, explicit formulae have been obtained by Hara and Sako \cite{HS1,HS2} for $N=1$ and the authors of this paper \cite{OS_Varna,OS_jrnl} for $N=2$.

\section[Deformation quantization with separation of variables for G\_\{2,4\}(C)]{Deformation quantization with separation of variables \\ for $\boldsymbol{G_{2,4}(\mathbb{C})}$}
\label{G24_DQ_sep_var}

In Section~\ref{Gpq_recrel_general}, we give the recurrence relations for general
$ G_{p,p+q}(\mathbb{C})$ by using the recurrence relations \eqref{HS_recrel} given in~\cite{HS1,HS2}. We also focus on the case of $G_{2,4}(\mathbb{C})$ in Sections~\ref{G24_star_prd} and~\ref{G24_recrel_solve}. In Section~\ref{G24_star_prd}, we give another recurrence relations to solve more easily. Finally, we determine a~star product with separation of variables for $G_{2,4}(\mathbb{C})$ based on the general term of the recurrence relations obtained in Section~\ref{G24_recrel_solve}.

\subsection[Recurrence relations for G\_\{p,p+q\}(C)]{Recurrence relations for $\boldsymbol{G_{p,p+q}(\mathbb{C})}$}
\label{Gpq_recrel_general}

For $ G_{p,p+q}(\mathbb{C})$,
we obtain the recurrence relations that give a star product from Theorem~\ref{theorem:star_prd}.

\begin{Theorem}\label{theorem_rec_rel_Grassmann_general}
The recurrence relations which give a star product with separation of variables on
$ G_{p,p+q}(\mathbb{C})$
are given by
\begin{align}
 \hbar
 \sum_{D=1}^{qp}
 g_{\overline{I}D}
 T^{n-1}_{\vec{\alpha}_{n}-\vec{e}_{D},\vec{\beta}_{n}^{\ast}-\vec{e}_{I}^{\ast}}
 ={}&
 \beta_{I}^{n}
 \bigg\{
 1
 +
 \hbar
 -
 \hbar
 \bigg(
 \beta_{I}^{n}
 +
 \sum_{j'\neq i'}
 \beta_{ij'}^{n}
 +
 \sum_{j\neq i}
 \beta_{ji'}^{n}
 \bigg)
 \bigg\}
 T^{n}_{\vec{\alpha}_{n},\vec{\beta}_{n}^{\ast}}\nonumber \\
 &
 -
 \hbar
 \sum_{
 \substack{
 j\neq i
 \\
 j'\neq i'}
 }
 (
 \beta_{ij'}^{n}
 +
 1
 )
 (
 \beta_{ji'}^{n}
 +
 1
 )
 T^{n}_{\vec{\alpha}_{n},\vec{\beta}_{n}^{\ast}-\vec{e}_{J}^{\ast}+\vec{e}_{ij'}^{\ast}+\vec{e}_{ji'}^{\ast}-\vec{e}_{I}^{\ast}}
 ,
\label{rec_rel_Grassmann_general}
\end{align}
where $I=ii'$, $J=jj'$ and $D=dd'$ are capital letter indices.
\end{Theorem}

\begin{proof}It is sufficient to calculate the right-hand side of \eqref{HS_recrel}.
By Proposition~\ref{curv_grassmann}, we have
\smash{$
 {{R_{\overline{M}}}^{\overline{K} \overline{K}}}_{\overline{I}}
 =
 -
 2
 \delta_{
 \overline{mi'},
 \overline{K}
 }
 \delta_{
 \overline{im'},
 \overline{K}
 }
$},
where $\delta_{IJ}=\delta_{ii',jj'}:=\delta_{ij}\delta_{i'j'}$.
Therefore, the second term on the right-hand side of \eqref{HS_recrel} is modified to
\begin{gather}
 \hbar
 \sum_{K=1}^{qp}
 \sum_{M=1}^{qp}
 \dbinom{\beta_{I}^{n}-\delta_{KM}-\delta_{IK}+2}{2}
 {{R_{\overline{M}}}^{\overline{K}\overline{K}}}_{\overline{I}}
 T^{n}_{\vec{\alpha}_{n},\vec{\beta}_{n}^{\ast}-\vec{e}_{M}^{\ast}+2\vec{e}_{K}^{\ast}-\vec{e}_{I}^{\ast}}
 \nonumber
 \\
 \qquad=
 -
 2
 \hbar
 \sum_{K=1}^{qp}
 \sum_{M=1}^{qp}
 \dbinom{\beta_{I}^{n}-\delta_{KM}-\delta_{IK}+2}{2}
 \delta_{
 \overline{mi'},
 \overline{K}
 }
 \delta_{
 \overline{im'},
 \overline{K}
 }
 T^{n}_{\vec{\alpha}_{n},\vec{\beta}_{n}^{\ast}-\vec{e}_{M}^{\ast}+2\vec{e}_{K}^{\ast}-\vec{e}_{I}^{\ast}}.
 \label{rec_rel_term2_1}
\end{gather}
Let us introduce $\psi$ as a one-to-one correspondence between two sets, the set of pairs $(i, i^{\prime})$ with $1 \le i \le q$ and $1 \le i^{\prime} \le p$, and the set of positive integers less than or equal to $pq$ by
\begin{gather*}
 \psi
 \colon\
 \{
 I
 =
 ii'
 :=
 (
 i,i'
 )
\mid
 1
 \leq
 i
 \leq
 q,\,
 1'
 \leq
 i'
 \leq
 p'
 \}
 \\
 \qquad\to
 \{
 1,
 \dots,
 p,
 p+1,
 \dots,
 2p,
 \dots,
 (q-1)p+1,
 \dots,
 qp
 \},
 \\
 I
 =
 ii'
 \mapsto
 \psi(I)
 :=
 p(i-1)
 +
 {i'}_{\mathbb{N}},
\end{gather*}
where ${i'}_{\mathbb{N}}$ is a natural number that
is mapped to $i'$ in a natural way. For example, if $i'=3'$, then~${{i'}_{\mathbb{N}}={3'}_{\mathbb{N}}=3}$. See also Appendix~\ref{appendix_properties_Gpq}.
Therefore, we regard $K$ and $M$ as ordinary ordered indices by such an identification, respectively. In addition, note that summation, addition $K+L$, and Kronecker's delta in \eqref{rec_rel_term2_1} are defined by
\begin{align*}
\sum_{K=1}^{qp}
:=
\sum_{\psi(K)=1}^{qp},
\qquad
\psi( K+L )
:=
\psi(K)+\psi(L),
\qquad
\delta_{KM}
:=
\delta_{\psi(K),\psi(M)},
\end{align*}
respectively. Since
$
 \delta_{
 \overline{mi'},
 \overline{K}
 }
 \delta_{
 \overline{im'},
 \overline{K}
 }
 =
 1
$
if and only if $K=M=I$, \eqref{rec_rel_term2_1} is expressed as
\begin{align}
 \eqref{rec_rel_term2_1}
 =
 -
 \hbar
 \beta_{I}^{n}
 (
 \beta_{I}^{n}
 -
 1
 )
 T^{n}_{\vec{\alpha}_{n},\vec{\beta}_{n}^{\ast}}.
 \label{recrel_gene_2nd}
\end{align}
We next calculate for the third term on the right-hand side of \eqref{HS_recrel}. Since the curvature is~\smash{$
 {{R_{\overline{M}}}^{\overline{K} \overline{K+L}}}_{\overline{I}}
 =
 -
 \delta_{
 \overline{im'},
 \overline{K}
 }
 \delta_{
 \overline{mi'},
 \overline{K+L}
 }
 -
 \delta_{
 \overline{im'},
 \overline{K+L}
 }
 \delta_{
 \overline{mi'},
 \overline{K}
 }
$}
from Proposition~\ref{curv_grassmann}, the third term can be rewritten as
\begin{gather}
 \hbar
 \sum_{K=1}^{qp-1}
 \sum_{L=1}^{qp-K}
 \sum_{M=1}^{qp}
 (
 \beta_{K}^{n}
 -
 \delta_{KM}
 -
 \delta_{IK}
 +
 1
 )
 (
 \beta_{K+L}^{n}
 -
 \delta_{K+L,M}
 -
 \delta_{I,K+L}
 +
 1
 )
 \nonumber
 \\
 \qquad
 \times
 {{R_{\overline{M}}}^{\overline{K} \overline{K+L}}}_{\overline{I}}
 T^{n}_{\vec{\alpha}_{n},\vec{\beta}_{n}^{\ast}-\vec{e}_{M}^{\ast}+\vec{e}_{K}^{\ast}+\vec{e}_{K+L}^{\ast}-\vec{e}_{I}^{\ast}}
 \nonumber
 \\
 \phantom{\qquad
 \times}{}=
 -
 \hbar
 \sum_{K=1}^{qp-1}
 \sum_{L=1}^{qp-K}
 \sum_{M=1}^{qp}
 (
 \beta_{K}^{n}
 -
 \delta_{KM}
 -
 \delta_{IK}
 +
 1
 )
 (
 \beta_{K+L}^{n}
 -
 \delta_{K+L,M}
 -
 \delta_{I,K+L}
 +
 1
 )
 \nonumber
 \\
 \phantom{\qquad
 \times=}{}
 \times
 (
 \delta_{
 \overline{im'},
 \overline{K}
 }
 \delta_{
 \overline{mi'},
 \overline{K+L}
 }
 +
 \delta_{
 \overline{im'},
 \overline{K+L}
 }
 \delta_{
 \overline{mi'},
 \overline{K}
 }
 )
 T^{n}_{\vec{\alpha}_{n},\vec{\beta}_{n}^{\ast}-\vec{e}_{M}^{\ast}+\vec{e}_{K}^{\ast}+\vec{e}_{K+L}^{\ast}-\vec{e}_{I}^{\ast}},
 \label{rhs_rec_rel_3rd_Grassmann}
\end{gather}
where $\sum_{L=1}^{qp-K}$ is defined by
\smash{$
\sum_{L=1}^{qp-K}
:=
\sum_{\psi(L)=1}^{qp-\psi(K)}$}.
To transform \eqref{rhs_rec_rel_3rd_Grassmann} in more detail, we divide the summation $\sum_{M=1}^{pq}$ into the following four cases:
\begin{gather*}
 (\textrm{I}) \
 M=I,
 \qquad
 (\textrm{II}) \
 M=ij',
\qquad
 j'
 \neq
 i',
 \\
 (\textrm{III}) \
 M=ji'
 j
 \neq
 i,
 \qquad
 (\textrm{IV}) \
 M=J=jj',
\qquad
 j
 \neq
 i, \
 j'
 \neq
 i'
 .
\end{gather*}
Based on the above four cases, \eqref{rhs_rec_rel_3rd_Grassmann} can be rewritten as
\begin{align*}
\eqref{rhs_rec_rel_3rd_Grassmann}
=
\gamma_{I}
+
\sum_{
j'
\neq
i'
}
\gamma_{ij'}
+
\sum_{
j
\neq
i
}
\gamma_{ji'}
+
\sum_{
\substack{
j
\neq
i
\\
j'
\neq
i'
}
}
\gamma_{J},
\end{align*}
where
\begin{gather*}
\gamma_{I}
=
 -
 2
 \hbar
 \sum_{K=1}^{qp-1}
 \sum_{L=1}^{qp-K}
 (
 \beta_{K}^{n}
 -
 2
 \delta_{IK}
 +
 1
 )
 (
 \beta_{K+L}^{n}
 -
 2
 \delta_{I,K+L}
 +
 1
 )
 \delta_{
 \overline{ii'},
 \overline{K}
 }
 \delta_{
 \overline{ii'},
 \overline{K+L}
 }\\
 \phantom{\gamma_{I}
=}{}
 \times T^{n}_{\vec{\alpha}_{n},\vec{\beta}_{n}^{\ast}-2\vec{e}_{I}^{\ast}+\vec{e}_{K}^{\ast}+\vec{e}_{K+L}^{\ast}},
\\
\gamma_{ij'}
=
 -
 \hbar
 \sum_{K=1}^{qp-1}
 \sum_{L=1}^{qp-K}
 (
 \beta_{K}^{n}
 -
 \delta_{K,ij'}
 -
 \delta_{IK}
 +
 1
 )
 (
 \beta_{K+L}^{n}
 -
 \delta_{K+L,ij'}
 -
 \delta_{I,K+L}
 +
 1
 )
 \\
\phantom{\gamma_{ij'}
=}{}
 \times
 (
 \delta_{
 \overline{ij'},
 \overline{K}
 }
 \delta_{
 \overline{ii'},
 \overline{K+L}
 }
 +
 \delta_{
 \overline{ij'},
 \overline{K+L}
 }
 \delta_{
 \overline{ii'},
 \overline{K}
 }
 )
 T^{n}_{\vec{\alpha}_{n},\vec{\beta}_{n}^{\ast}-\vec{e}_{ij'}^{\ast}
 +\vec{e}_{K}^{\ast}+\vec{e}_{K+L}^{\ast}-\vec{e}_{I}^{\ast}},
\\
\gamma_{ji'}
=
-
 \hbar
 \sum_{K=1}^{qp-1}
 \sum_{L=1}^{qp-K}
 (
 \beta_{K}^{n}
 -
 \delta_{K,ji'}
 -
 \delta_{IK}
 +
 1
 )
 (
 \beta_{K+L}^{n}
 -
 \delta_{K+L,ji'}
 -
 \delta_{I,K+L}
 +
 1
 )
 \\
\phantom{\gamma_{ji'}
=}{}
 \times
 (
 \delta_{
 \overline{ii'},
 \overline{K}
 }
 \delta_{
 \overline{ji'},
 \overline{K+L}
 }
 +
 \delta_{
 \overline{ii'},
 \overline{K+L}
 }
 \delta_{
 \overline{ji'},
 \overline{K}
 }
 )
 T^{n}_{\vec{\alpha}_{n},\vec{\beta}_{n}^{\ast}-\vec{e}_{ji'}^{\ast}+\vec{e}_{K}^{\ast}
 +\vec{e}_{K+L}^{\ast}-\vec{e}_{I}^{\ast}},
\\
\gamma_{J}
=
-
 \hbar
 \sum_{K=1}^{qp-1}
 \sum_{L=1}^{qp-K}
 (
 \beta_{K}^{n}
 -
 \delta_{KJ}
 -
 \delta_{IK}
 +
 1
 )
 (
 \beta_{K+L}^{n}
 -
 \delta_{K+L,J}
 -
 \delta_{I,K+L}
 +
 1
 ) \\
\phantom{\gamma_{J}
=}{}
 \times
 (
 \delta_{
 \overline{ij'},
 \overline{K}
 }
 \delta_{
 \overline{ji'},
 \overline{K+L}
 }
 +
 \delta_{
 \overline{ij'},
 \overline{K+L}
 }
 \delta_{
 \overline{ji'},
 \overline{K}
 }
 )
 T^{n}_{\vec{\alpha}_{n},\vec{\beta}_{n}^{\ast}-\vec{e}_{J}^{\ast}
 +\vec{e}_{K}^{\ast}+\vec{e}_{K+L}^{\ast}-\vec{e}_{I}^{\ast}},
\end{gather*}
respectively.

$(\textrm{I})$ $M=I$: In this case,
\smash{$
 \delta_{
 \overline{I},
 \overline{K}
 }
 \delta_{
 \overline{I},
 \overline{K+L}
 }
 =
 \delta_{
 \overline{ii'},
 \overline{K}
 }
 \delta_{
 \overline{ii'},
 \overline{K+L}
 }
$}
is equal to 0 since
$
 K
 \neq
 K+L
$
always holds. Hence, $\gamma_{I}=0$ when $M=I$.

 $(\textrm{II})$
 $M=ij'$
 $(
 j'
 \neq
 i'
 )$:
Since $\gamma_{ij'}$ is expressed as a summation over $K=kk'$,
$\gamma_{ij'}$ is written as~\smash{$
 \gamma_{ij'}
 =
 \sum_{K=1}^{qp-1}
 \gamma'_{ij',K}$},
where
\begin{gather}
 \gamma'_{ij',K}
 =
 -
 \hbar
 \sum_{L=1}^{qp-K}
 (
 \beta_{K}^{n}
 -
 \delta_{K,ij'}
 -
 \delta_{IK}
 +
 1
 )
 (
 \beta_{K+L}^{n}
 -
 \delta_{K+L,ij'}
 -
 \delta_{I,K+L}
 +
 1
 )
 \nonumber
 \\
\phantom{ \gamma'_{ij',K}
 =}{}
 \times
 (
 \delta_{
 \overline{ij'},
 \overline{K}
 }
 \delta_{
 \overline{ii'},
 \overline{K+L}
 }
 +
 \delta_{
 \overline{ij'},
 \overline{K+L}
 }
 \delta_{
 \overline{ii'},
 \overline{K}
 }
 )
 T^{n}_{\vec{\alpha}_{n},\vec{\beta}_{n}^{\ast}-\vec{e}_{ij'}^{\ast}+\vec{e}_{K}^{\ast}+\vec{e}_{K+L}^{\ast}-\vec{e}_{I}^{\ast}}.
\label{gamma_ij}
\end{gather}
If
$
 k
 \neq
 i
$,
then $\gamma'_{ij',K}=0$
from
$
 \delta_{
 \overline{ii'}
 \overline{K}
 }
 =
 \delta_{
 \overline{ij'},
 \overline{K}
 }
 =
 0
$.
Thus we consider the case
$
 k
 =
 i
$
as a non-zero contribution.

For
$
 K
 =
 I
 =
 ii'
$,
\smash{$
 \delta_{
 \overline{ij'},
 \overline{I}
 }
 \delta_{
 \overline{I},
 \overline{I+L}
 }
 +
 \delta_{
 \overline{ij'},
 \overline{I+L}
 }
 \delta_{
 \overline{I},
 \overline{I}
 }
 =
 \delta_{
 \overline{ij'},
 \overline{I+L}
 }
$}
from
$
 \delta_{
 \overline{I},
 \overline{I+L}
 }
 =
 0
$,
and
$
 \delta_{
 \overline{ij'},
 \overline{I+L}
 }
 =
 1
$
when
$
L
=
ij'
-
I
$.
Therefore, if $K=I$ with
$
j'
\neq
i'
$,
then
\begin{align*}
 \gamma'_{ij',I}
={}&
 -
 \hbar
 (
 \beta_{I}^{n}
 -
 \delta_{I,ij'}
 -
 \delta_{II}
 +
 1
 )
 (
 \beta_{I+ij'-I}^{n}
 -
 \delta_{I+ij'-I,ij'}
 -
 \delta_{I,I+ij'-I}
 +
 1
 )
 \nonumber
 \\
 &
 \times
 (
 \delta_{
 \overline{ij'}
 \overline{I}
 }
 \delta_{
 \overline{I},
 \overline{I+ij'-I}
 }
 +
 \delta_{
 \overline{ij'},
 \overline{I+ij'-I}
 }
 \delta_{
 \overline{I}
 \overline{I}
 }
 )
 T^{n}_{\vec{\alpha}_{n},\vec{\beta}_{n}^{\ast}-\vec{e}_{ij'}^{\ast}+\vec{e}_{I+ij'-I}^{\ast}}
 \\
={}&
 -
 \hbar
 \beta_{I}^{n}
 \beta_{ij'}^{n}
 T^{n}_{\vec{\alpha}_{n},\vec{\beta}_{n}^{\ast}}
 .
\end{align*}

For the other cases, i.e.,
$
 K
 =
 ij'
 (
 j'
 \neq
 i'
 )
$,
$
 (
 \delta_{
 \overline{ij'},
 \overline{K}
 }
 \delta_{
 \overline{ii'},
 \overline{K+L}
 }
 +
 \delta_{
 \overline{ij'},
 \overline{K+L}
 }
 \delta_{
 \overline{ii'},
 \overline{K}
 }
 )
$
in \eqref{gamma_ij} becomes
\smash{$
 \delta_{
 \overline{ij'},
 \overline{ij'}
 }
 \delta_{
 \overline{I},
 \overline{ij'+L}
 }
 +
 \delta_{
 \overline{ij'},
 \overline{ij'+L}
 }
 \delta_{
 \overline{I},
 \overline{ij'}
 }
 =
 \delta_{
 \overline{I},
 \overline{ij'+L}
 }
$}
since
$
 \delta_{
 \overline{I},
 \overline{ij'}
 }
 =
 0
$,
and
$
 \delta_{
 \overline{I},
 \overline{ij'+L}
 }
 =
 1
$
when $L=I-ij'$. Therefore, if
$
 K
 =
 ij'$
 ($
 j'
 \neq
 i'
 $),
then we have
\begin{align*}
 \gamma'_{ij',ij'}
 &
 =
 -
 \hbar
 (
 \beta_{ij'}^{n}
 -
 \delta_{ij',ij'}
 -
 \delta_{I,ij'}
 +
 1
 )
 (
 \beta_{ij'+I-ij'}^{n}
 -
 \delta_{ij'+I-ij',ij'}
 -
 \delta_{I,ij'+I-ij'}
 +
 1
 )
 \nonumber
 \\
 &
 \qquad
 \times
 (
 \delta_{
 \overline{ij'}
 \overline{ij'}
 }
 \delta_{
 \overline{I},
 \overline{ij'+I-ij'}
 }
 +
 \delta_{
 \overline{ij'},
 \overline{ij'+I-ij'}
 }
 \delta_{
 \overline{I}
 \overline{ij'}
 }
 )
 T^{n}_{\vec{\alpha}_{n},\vec{\beta}_{n}^{\ast}-\vec{e}_{ij'}^{\ast}+\vec{e}_{ij'}^{\ast}+\vec{e}_{ij'+I-ij'}^{\ast}-\vec{e}_{I}^{\ast}}
 \\
 &=
 - \hbar
 \beta_{ij'}^{n}
 \beta_{I}^{n}
 T^{n}_{\vec{\alpha}_{n},\vec{\beta}_{n}^{\ast}}.
\end{align*}
Note that $\psi(L)$ is positive and
$
j'
$
($
\neq
i'
$)
is fixed, $L$ can take either $I-ij'$ or $ij'-I$ depending on $j$. In fact, if $j'>i'$, $L$ cannot take $I-ij'$ because $\psi(I)-\psi(ij')<0$. Conversely, if $j'<i'$, it can be verified that $L$ cannot take $ij'-I$. Therefore,
\begin{align}
 \gamma_{ij'}
 =
 \sum_{K}
 \gamma'_{ij',K}
 =
 -
 \hbar
 \beta_{I}^{n}
 \beta_{ij'}^{n}
 T^{n}_{\vec{\alpha}_{n},\vec{\beta}_{n}^{\ast}}
 \label{rec_rel_gene_3rdterm_2}
\end{align}
for
any
$
j'
(
\neq
i'
)
$. In a similar way, the cases $(\textrm{III})$ and $(\textrm{IV})$ are also calculated, and it can be verified that
\begin{gather}
 \gamma_{ji'}
 =
 -
 \hbar
 \beta_{I}^{n}
 \beta_{ji'}^{n}
 T^{n}_{\vec{\alpha}_{n},\vec{\beta}_{n}^{\ast}}
 \qquad
 \textrm{for}
\quad
 j
 \neq
 i
 ,
 \label{rec_rel_gene_3rdterm_3}
 \\
 \gamma_{J}
 =
 -
 \hbar
 (
 \beta_{ij'}^{n}
 +
 1
 )
 (
 \beta_{ji'}^{n}
 +
 1
 )
 T^{n}_{\vec{\alpha}_{n},\vec{\beta}_{n}^{\ast}-\vec{e}_{I}^{\ast}-\vec{e}_{jj'}^{\ast}+\vec{e}_{ji'}^{\ast}+\vec{e}_{ij'}^{\ast}}
 \qquad
 \textrm{for}
\quad
 j
 \neq
 i,
\
 j'
 \neq
 i'
 .
 \label{rec_rel_gene_3rdterm_4}
\end{gather}
Substituting $\gamma_{I}=0$ and \eqref{rec_rel_gene_3rdterm_2}--\eqref{rec_rel_gene_3rdterm_4} into \eqref{rhs_rec_rel_3rd_Grassmann}, we obtain
\begin{align}
 \eqref{rhs_rec_rel_3rd_Grassmann}
={}&
 \sum_{
 j'
 \neq
 i'
 }
 \eqref{rec_rel_gene_3rdterm_2}
 +
 \sum_{
 j
 \neq
 i
 }
 \eqref{rec_rel_gene_3rdterm_3}
 +
 \sum_{
 j
 \neq
 i
 }
 \sum_{
 j'
 \neq
 i'
 }
 \eqref{rec_rel_gene_3rdterm_4}
 \nonumber
 \\
={}&
 -
 \hbar
 \beta_{I}^{n}
 \bigg(
 \sum_{
 j'
 \neq
 i'
 }
 \beta_{ij'}^{n}
 +
 \sum_{
 j
 \neq
 i
 }
 \beta_{ji'}^{n}
 \bigg)
 T^{n}_{\vec{\alpha}_{n},\vec{\beta}_{n}^{\ast}}\nonumber\\
 &
 -
 \hbar
 \sum_{
 j
 \neq
 i
 }
 \sum_{
 j'
 \neq
 i'
 }
 (
 \beta_{ij'}^{n}
 +
 1
 )
 (
 \beta_{ji'}^{n}
 +
 1
 )
 T^{n}_{\vec{\alpha}_{n},\vec{\beta}_{n}^{\ast}-\vec{e}_{I}^{\ast}-\vec{e}_{jj'}^{\ast}+\vec{e}_{ji'}^{\ast}+\vec{e}_{ij'}^{\ast}}.
 \label{recrel_gene_3rd}
\end{align}

Hence, the right-hand side of \eqref{HS_recrel} for $G_{p,p+q}(\mathbb{C})$ is obtained as follows:
\begin{gather}
 \beta_{I}^{n}
 T^{n}_{\vec{\alpha}_{n},\vec{\beta}_{n}^{\ast}}
 +
 \eqref{recrel_gene_2nd}
 +
 \eqref{recrel_gene_3rd}\nonumber\\
 \qquad
 =
 \beta_{I}^{n}
 T^{n}_{\vec{\alpha}_{n},\vec{\beta}_{n}^{\ast}}
 -
 \hbar
 \beta_{I}^{n}
 (
 \beta_{I}^{n}
 -
 1
 )
 T^{n}_{\vec{\alpha}_{n},\vec{\beta}_{n}^{\ast}}
 -
 \hbar
 \beta_{I}^{n}
 \bigg(
 \sum_{
 j'
 \neq
 i'
 }
 \beta_{ij'}^{n}
 +
 \sum_{
 j
 \neq
 i
 }
 \beta_{ji'}^{n}
 \bigg)
 T^{n}_{\vec{\alpha}_{n},\vec{\beta}_{n}^{\ast}}
 \nonumber
 \\
\phantom{\qquad
 =}{}
 -
 \hbar
 \sum_{
 j
 \neq
 i
 }
 \sum_{
 j'
 \neq
 i'
 }
 (
 \beta_{ij'}^{n}
 +
 1
 )
 (
 \beta_{ji'}^{n}
 +
 1
 )
 T^{n}_{\vec{\alpha}_{n},\vec{\beta}_{n}^{\ast}-\vec{e}_{I}^{\ast}-\vec{e}_{jj'}^{\ast}+\vec{e}_{ji'}^{\ast}
 +\vec{e}_{ij'}^{\ast}}
 \nonumber
 \\
\qquad
 =
 \beta_{I}^{n}
 \bigg\{
 1
 +
 \hbar
 -
 \hbar
 \bigg(
 \beta_{I}^{n}
 +
 \sum_{
 j'
 \neq
 i'
 }
 \beta_{ij'}^{n}
 +
 \sum_{
 j
 \neq
 i
 }
 \beta_{ji'}^{n}
 \bigg)
 \bigg\}
 T^{n}_{\vec{\alpha}_{n},\vec{\beta}_{n}^{\ast}}
 \nonumber
 \\
 \phantom{\qquad
 =}{}
 -
 \hbar
 \sum_{
 j
 \neq
 i
 }
 \sum_{
 j'
 \neq
 i'
 }
 (
 \beta_{ij'}^{n}
 +
 1
 )
 (
 \beta_{ji'}^{n}
 +
 1
 )
 T^{n}_{\vec{\alpha}_{n},\vec{\beta}_{n}^{\ast}-\vec{e}_{I}^{\ast}-\vec{e}_{jj'}^{\ast}+\vec{e}_{ji'}^{\ast}+\vec{e}_{ij'}^{\ast}}.
 \label{recrel_gene_right_mid}
\end{gather}
This means that we obtain \eqref{rec_rel_Grassmann_general}.
\end{proof}

Recalling that a capital letter index can be identified with an ordinary index, \eqref{rec_rel_Grassmann_general} can be rewritten as
\begin{gather}
 \hbar
 \sum_{d=1}^{q}
 \sum_{d'=1'}^{p'}
 g_{\overline{I}D}
 T^{n-1}_{\vec{\alpha}_{n}-\vec{e}_{D},\vec{\beta}_{n}^{\ast}-\vec{e}_{I}^{\ast}}
\nonumber
 \\
 \qquad=
 \beta_{I}^{n}
 \Biggl\{
 1
 +
 \hbar
 -
 \hbar
 \biggl(
 n
 -
 \sum_{
 \substack{
 j
 \neq
 i
 \\
 j'
 \neq
 i'}
 }
 \beta_{jj'}^{n}
 \biggr)
 \Biggr\}
 T^{n}_{\vec{\alpha}_{n},\vec{\beta}_{n}^{\ast}}\nonumber\\
 \phantom{ \qquad=}{}
 -
 \hbar
 \sum_{
 \substack{
 j
 \neq
 i
 \\
 j'
 \neq
 i'}
 }
 (
 \beta_{ij'}^{n}
 +
 1
 )
 (
 \beta_{ji'}^{n}
 +
 1
 )
 T^{n}_{\vec{\alpha}_{n},\vec{\beta}_{n}^{\ast}-\vec{e}_{J}^{\ast}
 +\vec{e}_{ij'}^{\ast}+\vec{e}_{ji'}^{\ast}-\vec{e}_{I}^{\ast}}.
 \label{rec_rel_Grassmann_general_capital_only}
\end{gather}
Here, we use the fact that
$
 \beta_{I}^{n}
 +
 \sum_{
 j'
 \neq
 i'
 }
 \beta_{ij'}^{n}
 +
 \sum_{
 j
 \neq
 i
 }
 \beta_{ji'}^{n}
 +
 \sum_{
 j
 \neq
 i
 }
 \sum_{
 j'
 \neq
 i'
 }
 \beta_{jj'}^{n}
 =
 n$,
since
$
\big|
\vec{\beta}_{n}^{\ast}
\big|
=
\sum_{K}
\beta_{K}^{n}
=
\sum_{k}
\sum_{k'}
\beta_{kk'}^{n}
=
n
$
in the first term on the right-hand side of \eqref{recrel_gene_right_mid}.

\subsection[Star product with separation of variables on G\_\{2,4\}(C)]{Star product with separation of variables on $\boldsymbol{G_{2,4}(\mathbb{C})}$}
\label{G24_star_prd}

If the solution of the recurrence relations for \eqref{rec_rel_Grassmann_general} (or equivalently \eqref{rec_rel_Grassmann_general_capital_only}) is determined, we can explicitly give a star product with separation of variables on
$ G_{p,p+q}(\mathbb{C})$.
However, it is not simple to solve~\eqref{rec_rel_Grassmann_general} (or equivalently~\eqref{rec_rel_Grassmann_general_capital_only}) in general. We need to solve the system of linear equations for solutions of order $n$ to determine \smash{$T^{n}_{\vec{\alpha}_{n},\vec{\beta}_{n}^{\ast}}$}, even if solutions of order $(n-1)$ are given. In addition, the number of variables in this system of linear equations increases combinatorially with increasing $n$. In this subsection, we focus on $G_{2,4}(\mathbb{C})$ and give the general term of the recurrence relations \eqref{rec_rel_Grassmann_general} for $G_{2,4}(\mathbb{C})$ by eliminating these problems. By using the general terms, we give its explicit star product with separation of variables.

For $G_{2,4}(\mathbb{C})$, its recurrence relations can be obtained immediately by considering the case $p=q=2$.
\begin{Corollary}\label{theorem_rec_rel_Grassmann_2_4}
For any $($fixed$)$ capital letter index $I\in\mathcal{I}$, the recurrence relations for
$
 G_{2,4}(\mathbb{C})
$
are given by
\begin{gather}
 \hbar
 \sum_{
 D
 \in
 \mathcal{I}
 }
 g_{\overline{I}D}
 T^{n-1}_{\vec{\alpha}_{n}-\vec{e}_{D},\vec{\beta}_{n}^{\ast}-\vec{e}_{I}^{\ast}}\nonumber\\
 \qquad
 =
 \hbar
 \beta_{I}^{n}
 (
 \tau_{n}
 +
 \beta_{\cancel{I}}^{n}
 )
 T^{n}_{\vec{\alpha}_{n},\vec{\beta}_{n}^{\ast}}
 -
 \hbar
 (
 \beta_{i\cancel{i}'}^{n}
 +
 1
 )
 (
 \beta_{\cancel{i}i'}^{n}
 +
 1
 )
 T^{n}_{\vec{\alpha}_{n},\vec{\beta}_{n}^{\ast}-\vec{e}_{I}^{\ast}+\vec{e}_{i\cancel{i}'}^{\ast}+\vec{e}_{\cancel{i}i'}^{\ast} -\vec{e}_{\cancel{I}}^{\ast}},
 \label{rec_rel_Grassmann_2_4_slash}
\end{gather}
where
$
 \tau_{n}
 :=
 1
 -
 n
 +
 \frac{1}{\hbar}
$,
$
\mathcal{I}
:=
\{
I,
\cancel{I},
i\cancel{i}',
\cancel{i}i'
\}
=
\{
11',
12',
21',
22'
\}
$.
\end{Corollary}

It is not a waste to note the notation of capital letter indices again. In $G_{2,4}(\mathbb{C})$ case, $\cancel{i}$ (or~$\cancel{i}'$) is the other index which is not $i$ (or not $i'$). $\cancel{i}$ (or $\cancel{i}'$) is uniquely determined when $i$ (or $i'$) is fixed. For example, if $I=11'$, then $i\cancel{i}'=12'$, $\cancel{i}i'=21'$, and $\cancel{I}=22'$. $I=ii'$ may take $12'$, $21'$ and $22'$ as well as $11'$.

To determine the explicit star product on $G_{2,4}(\mathbb{C})$, it is necessary to find the general term \smash{$T\raisebox{2pt}{${}^{n}_{\vec{\alpha}_{n},\vec{\beta}_{n}^{\ast}}$}$} as the solution of the recurrence relations \eqref{rec_rel_Grassmann_2_4_slash}. It is not easy to find the solution \smash{$T\raisebox{2pt}{${}^{n}_{\vec{\alpha}_{n},\vec{\beta}_{n}^{\ast}}$}$} of~\eqref{rec_rel_Grassmann_2_4_slash} by direct calculation, because \eqref{rec_rel_Grassmann_2_4_slash} contains another coefficient
\[
T^{n}_{\vec{\alpha}_{n},\vec{\beta}_{n}^{\ast}-\vec{e}_{I}^{\ast}
+\vec{e}_{i\cancel{i}'}^{\ast}+\vec{e}_{\cancel{i}i'}^{\ast}-\vec{e}_{\cancel{I}}^{\ast}}
\]
 on the right-hand side. On the other hand, by introducing another recurrence relation for \smash{$T\raisebox{2pt}{${}^{n}_{\vec{\alpha}_{n},\vec{\beta}_{n}^{\ast}-\vec{e}_{I}^{\ast}+\vec{e}_{i\cancel{i}'}^{\ast}
+\vec{e}_{\cancel{i}i'}^{\ast}-\vec{e}_{\cancel{I}}^{\ast}}$}$}, we can obtain the recurrence relations equivalent to \eqref{rec_rel_Grassmann_2_4_slash}.

\begin{Proposition}\label{Prop_G24_relation}
The recurrence relations \eqref{rec_rel_Grassmann_2_4_slash} are equivalent to the following recurrence relations:
\begin{align}
 \beta_{I}^{n}
 T^{n}_{\vec{\alpha}_{n},\vec{\beta}_{n}^{\ast}}
 &=
 \frac{
 \sum_{
 D
 \in
 \mathcal{I}
 }
 \bigl\{
 (
 \tau_{n}
 +
 \beta_{i\cancel{i}'}^{n}
 +
 1
 )
 g_{\overline{I}D}
 T^{n-1}_{
 \vec{\alpha}_{n}
 -
 \vec{e}_{D},
 \vec{\beta}_{n}^{\ast}
 -
 \vec{e}_{I}^{\ast}
 }
 +
 (
 \beta_{i\cancel{i}'}^{n}
 +
 1
 )
 g_{\overline{\cancel{i}i'},D}
 T^{n-1}_{
 \vec{\alpha}_{n}
 -
 \vec{e}_{D},
 \vec{\beta}_{n}^{\ast}
 -
 \vec{e}_{I}^{\ast}
 -
 \vec{e}_{\cancel{I}}^{\ast}
 +
 \vec{e}_{i\cancel{i}'}^{\ast}
 }
 \bigr\}
 }{
 \tau_{n}
 (
 \tau_{n}
 +
 \beta_{\cancel{I}}^{n}
 +
 \beta_{i\cancel{i}'}^{n}
 +
 1
 )
 }
 \nonumber
 \\
 &=
 \frac{
 \sum_{
 D
 \in
 \mathcal{I}
 }
 \sum_{k=1}^{2}
 (
 \tau_{n}
 \delta_{ik}
 +
 \beta_{i\cancel{i}'}^{n}
 +
 1
 )
 g_{\overline{ki'},D}
 T^{n-1}_{
 \vec{\alpha}_{n}
 -
 \vec{e}_{D},
 \vec{\beta}_{n}^{\ast}
 -
 \vec{e}_{I}^{\ast}
 -
 \delta_{\cancel{i}k}
 \bigl(
 \vec{e}_{\cancel{I}}^{\ast}
 -
 \vec{e}_{i\cancel{i}'}^{\ast}
 \bigr)
 }
 }{
 \tau_{n}
 (
 \tau_{n}
 +
 \beta_{\cancel{I}}^{n}
 +
 \beta_{i\cancel{i}'}^{n}
 +
 1
 )
 }.
\label{G24_n_relation}
\end{align}
\end{Proposition}

\begin{proof}
To derive \eqref{G24_n_relation} from \eqref{rec_rel_Grassmann_2_4_slash}, we introduce another recurrence relation for a fixed index~$\cancel{i}i'$ as $I$ in \eqref{rec_rel_Grassmann_2_4_slash}
\begin{gather}
 \hbar
 \sum_{
 D
 \in
 \mathcal{I}
 }
 g_{\overline{\cancel{i}i'},D}
 T^{n-1}_{\vec{\alpha}_{n}-\vec{e}_{D},\vec{\beta}_{n}^{\ast}-\vec{e}_{\cancel{i}i'}^{\ast}}\nonumber\\
 \qquad
 =
 \hbar
 \beta_{\cancel{i}i'}^{n}
 (
 \tau_{n}
 +
 \beta_{i\cancel{i}'}^{n}
 )
 T^{n}_{\vec{\alpha}_{n},\vec{\beta}_{n}^{\ast}}
 -
 \hbar
 (
 \beta_{\cancel{I}}^{n}
 +
 1
 )
 (
 \beta_{I}^{n}
 +
 1
 )
 T^{n}_{\vec{\alpha}_{n},\vec{\beta}_{n}^{\ast}-\vec{e}_{\cancel{i}i'}^{\ast}+\vec{e}_{\cancel{I}}^{\ast}+\vec{e}_{I}^{\ast}-\vec{e}_{i\cancel{i}'}^{\ast}}.
 \label{rec_rel_Grassmann_2_4_slash_islashi}
\end{gather}
By considering the case of
\smash{$
 T^{n}_{\vec{\alpha}_{n},\vec{\beta}_{n}^{\ast}-\vec{e}_{I}^{\ast}+\vec{e}_{i\cancel{i}'}^{\ast}+\vec{e}_{\cancel{i}i'}^{\ast}-\vec{e}_{\cancel{I}}^{\ast}}
$}
as the coefficient
\smash{$T^{n}_{\vec{\alpha}_{n},\vec{\beta}_{n}^{\ast}}$} in
\eqref{rec_rel_Grassmann_2_4_slash_islashi}, we have
\begin{gather}
 \hbar
 \sum_{
 D
 \in
 \mathcal{I}
 }
 g_{\overline{\cancel{i}i'},D}
 T^{n-1}_{\vec{\alpha}_{n}-\vec{e}_{D},\vec{\beta}_{n}^{\ast}-\vec{e}_{I}^{\ast}-\vec{e}_{\cancel{I}}^{\ast}+\vec{e}_{i\cancel{i}'}^{\ast}}
\nonumber\\
\qquad =
 \hbar
 (
 \beta_{\cancel{i}i'}^{n}
 +
 1
 )
 (
 \tau_{n}
 +
 \beta_{i\cancel{i}'}^{n}
 +
 1
 )
 T^{n}_{\vec{\alpha}_{n},\vec{\beta}_{n}^{\ast}-\vec{e}_{I}^{\ast}+\vec{e}_{i\cancel{i}'}^{\ast}+\vec{e}_{\cancel{i}i'}^{\ast}-\vec{e}_{\cancel{I}}^{\ast}}
 -
 \hbar
 \beta_{\cancel{I}}^{n}
 \beta_{I}^{n}
 T^{n}_{\vec{\alpha}_{n},\vec{\beta}_{n}^{\ast}}.
 \label{rec_rel_Grassmann_2_4_slash_islashi_another}
\end{gather}
We now consider the system of linear equations for \eqref{rec_rel_Grassmann_2_4_slash} and \eqref{rec_rel_Grassmann_2_4_slash_islashi_another}
\begin{gather}
 \hbar
 \left(
 \begin{matrix}
 \sum_{
 D
 \in
 \mathcal{I}
 }
 g_{\overline{I}D}
 T^{n-1}_{\vec{\alpha}_{n}-\vec{e}_{D},\vec{\beta}_{n}^{\ast}-\vec{e}_{I}^{\ast}}
 \\
 \sum_{
 D
 \in
 \mathcal{I}
 }
 g_{\overline{\cancel{i}i'},D}
 T^{n-1}_{\vec{\alpha}_{n}-\vec{e}_{D},\vec{\beta}_{n}^{\ast}-\vec{e}_{I}^{\ast}-\vec{e}_{\cancel{I}}^{\ast}+\vec{e}_{i\cancel{i}'}^{\ast}}
 \end{matrix}
 \right)
 \nonumber
 \\
 \qquad=
 \hbar
 \left(
 \begin{matrix}
 \tau_{n}
 +
 \beta_{\cancel{I}}^{n}
 &
 -
 (
 \beta_{i\cancel{i}'}^{n}
 +
 1
 )
 \\
 -\beta_{\cancel{I}}^{n}
 &
 \tau_{n}
 +
 \beta_{i\cancel{i}'}^{n}
 +
 1
 \end{matrix}
 \right)
 \left(
 \begin{matrix}
 \beta_{I}^{n}
 T^{n}_{\vec{\alpha}_{n},\vec{\beta}_{n}^{\ast}}
 \\
 (
 \beta_{\cancel{i}i'}^{n}
 +
 1
 )
 T^{n}_{\vec{\alpha}_{n},\vec{\beta}_{n}^{\ast}-\vec{e}_{I}^{\ast}+\vec{e}_{i\cancel{i}'}^{\ast}+\vec{e}_{\cancel{i}i'}^{\ast}-\vec{e}_{\cancel{I}}^{\ast}}
 \end{matrix}
 \right).
 \label{recrel_mtx_24}
\end{gather}
Calculating the determinant of the coefficient matrix, we have
\begin{align*}
 \det
 \left(
 \begin{matrix}
 \tau_{n}
 +
 \beta_{\cancel{I}}^{n}
 &
 -
 (
 \beta_{i\cancel{i}'}^{n}
 +
 1
 )
 \\
 -\beta_{\cancel{I}}^{n}
 &
 \tau_{n}
 +
 \beta_{i\cancel{i}'}^{n}
 +
 1
 \end{matrix}
 \right)
 =
 \tau_{n}
 (
 \tau_{n}
 +
 \beta_{\cancel{I}}^{n}
 +
 \beta_{i\cancel{i}'}^{n}
 +
 1
 ).
\end{align*}
Since
\begin{align*}
\tau_{n}^{-1}
 (
\tau_{n}
+
\beta_{\cancel{I}}^{n}
+
\beta_{i\cancel{i}'}^{n}
+
1
 )^{-1}
={}&
\frac{\hbar}{
1
-
\hbar (
n
-
1
 )
}
\cdot
\frac{\hbar}{
1
-
\hbar (
n
-
\beta_{\cancel{I}}^{n}
-
\beta_{i\cancel{i}'}^{n}
-
1
 )
}
\\
={}&
\hbar^{2}
\sum_{k,m=0}^{\infty}
 (
n
-
1
 )^{k}
 (
n
-
\beta_{\cancel{I}}^{n}
-
\beta_{i\cancel{i}'}^{n}
-
1
 )^{m}
\hbar^{k+m}
,
\end{align*}
then
$
\tau_{n}^{-1}
 (
\tau_{n}
+
\beta_{\cancel{I}}^{n}
+
\beta_{i\cancel{i}'}^{n}
+
1
 )^{-1}
$
exists.
Thus, there is the inverse matrix
\begin{gather}
 \left(
 \begin{matrix}
 \tau_{n}
 +
 \beta_{\cancel{I}}^{n}
 &
 -
 (
 \beta_{i\cancel{i}'}^{n}
 +
 1
 )
 \\
 -\beta_{\cancel{I}}^{n}
 &
 \tau_{n}
 +
 \beta_{i\cancel{i}'}^{n}
 +
 1
 \end{matrix}
 \right)^{-1}
 =
 \tau_{n}^{-1}
 (
 \tau_{n}
 +
 \beta_{\cancel{I}}^{n}
 +
 \beta_{i\cancel{i}'}^{n}
 +
 1
 )^{-1}
 \left(
 \begin{matrix}
 \tau_{n}
 +
 \beta_{i\cancel{i}'}^{n}
 +
 1
 &
 \beta_{i\cancel{i}'}^{n}
 +
 1
 \\
 \beta_{\cancel{I}}^{n}
 &
 \tau_{n}
 +
 \beta_{\cancel{I}}^{n}
 \end{matrix}
 \right).\!\!
 \label{recrel_mtx_24_coeff_inv}
\end{gather}
Multiplying the both sides of \eqref{recrel_mtx_24} by \eqref{recrel_mtx_24_coeff_inv} from the left, we get the solution of \eqref{recrel_mtx_24}
\begin{gather}
 \left(
 \begin{matrix}
 \beta_{I}^{n}
 T^{n}_{\vec{\alpha}_{n},\vec{\beta}_{n}^{\ast}}
 \\
 (
 \beta_{\cancel{i}i'}^{n}
 +
 1
 )
 T^{n}_{\vec{\alpha}_{n},\vec{\beta}_{n}^{\ast}-\vec{e}_{I}^{\ast}+\vec{e}_{i\cancel{i}'}^{\ast}+\vec{e}_{\cancel{i}i'}^{\ast}-\vec{e}_{\cancel{I}}^{\ast}}
 \end{matrix}
 \right)
=
 \tau_{n}^{-1}
 (
 \tau_{n}
 +
 \beta_{\cancel{I}}^{n}
 +
 \beta_{i\cancel{i}'}^{n}
 +
 1
 )^{-1}
 \nonumber
 \\
\phantom{\qquad=}{}
 \times
 \left(
 \begin{matrix}
 \sum_{
 D
 \in
 \mathcal{I}
 }
 \bigl\{
 (
 \tau_{n}
 +
 \beta_{i\cancel{i}'}^{n}
 +
 1
 )
 g_{\overline{I}D}
 T^{n-1}_{\vec{\alpha}_{n}-\vec{e}_{D},\vec{\beta}_{n}^{\ast}-\vec{e}_{I}^{\ast}}\\
 \qquad
 +
 (
 \beta_{i\cancel{i}'}^{n}
 +
 1
 )
 g_{\overline{\cancel{i}i'},D}
 T^{n-1}_{\vec{\alpha}_{n}-\vec{e}_{D},\vec{\beta}_{n}^{\ast}-\vec{e}_{I}^{\ast}-\vec{e}_{\cancel{I}}^{\ast}+\vec{e}_{i\cancel{i}'}^{\ast}}
 \bigr\}
 \vspace{1mm}\\
 \sum_{
 D
 \in
 \mathcal{I}
 }
 \bigl\{
 \beta_{\cancel{I}}^{n}
 g_{\overline{I}D}
 T^{n-1}_{\vec{\alpha}_{n}-\vec{e}_{D},\vec{\beta}_{n}^{\ast}-\vec{e}_{I}^{\ast}}
 +
 (
 \tau_{n}
 +
 \beta_{\cancel{I}}^{n}
 )
 g_{\overline{\cancel{i}i'},D}
 T^{n-1}_{\vec{\alpha}_{n}-\vec{e}_{D},\vec{\beta}_{n}^{\ast}-\vec{e}_{I}^{\ast}-\vec{e}_{\cancel{I}}^{\ast}+\vec{e}_{i\cancel{i}'}^{\ast}}
 \bigr\}
 \end{matrix}
 \right).
 \label{recrel_mtx_24_sol}
\end{gather}
By focusing on the first component of both sides of \eqref{recrel_mtx_24_sol}, we obtain \eqref{G24_n_relation}.

Conversely, we show that \eqref{rec_rel_Grassmann_2_4_slash} is derived from \eqref{G24_n_relation}. We introduce another recurrence relation for $\cancel{i}i'$ as $I$ in \eqref{G24_n_relation}
\begin{gather}
 \beta_{\cancel{i}i'}^{n}
 T^{n}_{\vec{\alpha}_{n},\vec{\beta}_{n}^{\ast}}\nonumber
 \\
 \qquad=
 \frac{
 \sum_{
 D
 \in
 \mathcal{I}
 }\!
 \bigl\{
 (
 \tau_{n}\!
 +\!
 \beta_{\cancel{I}}^{n}\!
 +\!
 1
 )
 g_{\overline{\cancel{i}i'},D}
 T^{n-1}_{
 \vec{\alpha}_{n}
 -
 \vec{e}_{D},
 \vec{\beta}_{n}^{\ast}
 -
 \vec{e}_{\cancel{i}i'}^{\ast}
 }\!
 +\!
 (
 \beta_{\cancel{I}}^{n}\!
 +\!
 1
 )
 g_{\overline{I}D}
 T^{n-1}_{
 \vec{\alpha}_{n}
 -
 \vec{e}_{D},
 \vec{\beta}_{n}^{\ast}
 -
 \vec{e}_{\cancel{i}i'}^{\ast}
 -
 \vec{e}_{i\cancel{i}'}^{\ast}
 +
 \vec{e}_{\cancel{I}}^{\ast}
 }
 \bigr\}
 }{
 \tau_{n}
 (
 \tau_{n}
 +
 \beta_{\cancel{I}}^{n}
 +
 \beta_{i\cancel{i}'}^{n}
 +
 1
 )
 }
 .
\label{G24_n_relation_another}
\end{gather}
Considering the case where
\smash{$
\vec{\beta}_{n}^{\ast}
$}
is
\smash{$
\vec{\beta}_{n}^{\ast}
-
\vec{e}_{I}^{\ast}
+
\vec{e}_{i\cancel{i}'}^{\ast}
+
\vec{e}_{\cancel{i}i'}^{\ast}
-
\vec{e}_{\cancel{I}}^{\ast}
$}
as
\smash{$
 T^{n}_{
 \vec{\alpha}_{n},
 \vec{\beta}_{n}^{\ast}
 }
$}
in \eqref{G24_n_relation_another}, it satisfies the recurrence relation
\begin{gather}
 (
 \beta_{\cancel{i}i'}^{n}
 +
 1
 )
 T^{n}_{
 \vec{\alpha}_{n},
 \vec{\beta}_{n}^{\ast}
 -
 \vec{e}_{I}^{\ast}
 +
 \vec{e}_{i\cancel{i}'}^{\ast}
 +
 \vec{e}_{\cancel{i}i'}^{\ast}
 -
 \vec{e}_{\cancel{I}}^{\ast}
 }\nonumber\\
 \qquad
 =
 \frac{
 \sum_{
 D
 \in
 \mathcal{I}
 }
 \bigl\{
 (
 \tau_{n}
 +
 \beta_{\cancel{I}}^{n}
 )
 g_{\overline{\cancel{i}i'},D}
 T^{n-1}_{
 \vec{\alpha}_{n}
 -
 \vec{e}_{D},
 \vec{\beta}_{n}^{\ast}
 -
 \vec{e}_{I}^{\ast}
 +
 \vec{e}_{i\cancel{i}'}^{\ast}
 -
 \vec{e}_{\cancel{I}}^{\ast}
 }
 +
 \beta_{\cancel{I}}^{n}
 g_{\overline{I}D}
 T^{n-1}_{
 \vec{\alpha}_{n}
 -
 \vec{e}_{D},
 \vec{\beta}_{n}^{\ast}
 -
 \vec{e}_{I}^{\ast}
 }
 \bigr\}
 }{
 \tau_{n}
 (
 \tau_{n}
 +
 \beta_{\cancel{I}}^{n}
 +
 \beta_{i\cancel{i}'}^{n}
 +
 1
 )
 }
 .
\label{G24_n_relation_another_2}
\end{gather}
Substituting \eqref{G24_n_relation} and \eqref{G24_n_relation_another_2} into the right hand side of \eqref{rec_rel_Grassmann_2_4_slash}, we have
\begin{gather*}
 \hbar
 \beta_{I}^{n}
 (
 \tau_{n}
 +
 \beta_{\cancel{I}}^{n}
 )
 T^{n}_{\vec{\alpha}_{n},\vec{\beta}_{n}^{\ast}}
 -
 \hbar
 (
 \beta_{i\cancel{i}'}^{n}
 +
 1
 )
 (
 \beta_{\cancel{i}i'}^{n}
 +
 1
 )
 T^{n}_{\vec{\alpha}_{n},\vec{\beta}_{n}^{\ast}-\vec{e}_{I}^{\ast}
 +\vec{e}_{i\cancel{i}'}^{\ast}+\vec{e}_{\cancel{i}i'}^{\ast}-\vec{e}_{\cancel{I}}^{\ast}}
 \\
 \qquad=
 \hbar
 \tau_{n}^{-1}
 (
 \tau_{n}
 +
 \beta_{\cancel{I}}^{n}
 +
 \beta_{i\cancel{i}'}^{n}
 +
 1
 )^{-1} \sum_{
 D
 \in
 \mathcal{I}
 }
 \big[
 (
 \tau_{n}
 +
 \beta_{\cancel{I}}^{n}
 )
 \bigl\{
 (
 \tau_{n}
 +
 \beta_{i\cancel{i}'}^{n}
 +
 1
 )
 g_{\overline{I}D}
 T^{n-1}_{
 \vec{\alpha}_{n}
 -
 \vec{e}_{D},
 \vec{\beta}_{n}^{\ast}
 -
 \vec{e}_{I}^{\ast}
 }
 \\
\phantom{\qquad=}{}
 +
 (
 \beta_{i\cancel{i}'}^{n}
 +
 1
 )
 g_{\overline{\cancel{i}i'},D}
 T^{n-1}_{
 \vec{\alpha}_{n}
 -
 \vec{e}_{D},
 \vec{\beta}_{n}^{\ast}
 -
 \vec{e}_{I}^{\ast}
 -
 \vec{e}_{\cancel{I}}^{\ast}
 +
 \vec{e}_{i\cancel{i}'}^{\ast}
 }
 \bigr\}
 \\
\phantom{\qquad=}{}
 -
 (
 \beta_{i\cancel{i}'}^{n}
 +
 1
 )
 \bigl\{
 (
 \tau_{n}
 +
 \beta_{\cancel{I}}^{n}
 )
 g_{\overline{\cancel{i}i'}D}
 T^{n-1}_{
 \vec{\alpha}_{n}
 -
 \vec{e}_{D},
 \vec{\beta}_{n}^{\ast}
 -
 \vec{e}_{I}^{\ast}
 +
 \vec{e}_{i\cancel{i}'}^{\ast}
 -
 \vec{e}_{\cancel{I}}^{\ast}
 }
 +
 \beta_{\cancel{I}}^{n}
 g_{\overline{I}D}
 T^{n-1}_{
 \vec{\alpha}_{n}
 -
 \vec{e}_{D},
 \vec{\beta}_{n}^{\ast}
 -
 \vec{e}_{I}^{\ast}
 }
 \bigr\}
 \big]
 \\
\phantom{\qquad}{}=
 \hbar
 \sum_{
 D
 \in
 \mathcal{I}
 }
 g_{\overline{I}D}
 T^{n-1}_{
 \vec{\alpha}_{n}
 -
 \vec{e}_{D},
 \vec{\beta}_{n}^{\ast}
 -
 \vec{e}_{I}^{\ast}
 }.
\end{gather*}
This is the left-hand side of \eqref{rec_rel_Grassmann_2_4_slash}.
\end{proof}

By using Proposition~\ref{Prop_G24_relation}, we obtain the following.

\begin{Proposition}\label{Cor_G24_relation_sum}
For $n\geq1$, the coefficient \smash{$T^{n}_{\vec{\alpha}_{n},\vec{\beta}_{n}^{\ast}}$} of \eqref{star_expand}
for $G_{2,4}(\mathbb{C})$ is given by using the coefficients of order $(n-1)$ as follows:
\begin{align}
 T^{n}_{\vec{\alpha}_{n},\vec{\beta}_{n}^{\ast}}
 ={}&
 \frac{
 \sum_{
 J,D
 \in
 \mathcal{I}
 }\!
 \bigl\{
 (
 \tau_{n}\!
 +\!
 \beta_{j\cancel{j}'}^{n}\!
 +\!
 1
 )
 g_{\overline{J}D}
 T^{n-1}_{
 \vec{\alpha}_{n}
 -
 \vec{e}_{D},
 \vec{\beta}_{n}^{\ast}
 -
 \vec{e}_{J}^{\ast}
 }\!
 +\!
 (
 \beta_{j\cancel{j}'}^{n}\!
 +\!
 1
 )
 g_{\overline{\cancel{j}j'},D}
 T^{n-1}_{
 \vec{\alpha}_{n}
 -
 \vec{e}_{D},
 \vec{\beta}_{n}^{\ast}
 -
 \vec{e}_{J}^{\ast}
 -
 \vec{e}_{\cancel{J}}^{\ast}
 +
 \vec{e}_{j\cancel{j}'}^{\ast}
 }
 \bigr\}
 }{
 \tau_{n}
 \bigl\{
 n
 (
 \tau_{n}+1
 )
 +
 2
 (
 \beta_{I}^{n}
 +
 \beta_{\cancel{i}i'}^{n}
 )
 (
 \beta_{\cancel{I}}^{n}
 +
 \beta_{i\cancel{i}'}^{n}
 )
 \bigr\}
 },
 \nonumber
 \\
 ={}&
 \frac{
 \sum_{
 J,D
 \in
 \mathcal{I}
 }
 \sum_{k=1}^{2}
 \bigl\{
 (
 \tau_{n}
 \delta_{jk}
 +
 \beta_{j\cancel{j}'}^{n}
 +
 1
 )
 g_{\overline{kj'},D}
 T^{n-1}_{
 \vec{\alpha}_{n}
 -
 \vec{e}_{D},
 \vec{\beta}_{n}^{\ast}
 -
 \vec{e}_{J}^{\ast}
 -
 \delta_{\cancel{j}k}
 (
 \vec{e}_{\cancel{J}}^{\ast}
 -
 \vec{e}_{j\cancel{j}'}^{\ast}
 )
 }
 \bigr\}
 }{
 \tau_{n}
 \bigl\{
 n
 (
 \tau_{n}+1
 )
 +
 2
 (
 \beta_{I}^{n}
 +
 \beta_{\cancel{i}i'}^{n}
 )
 (
 \beta_{\cancel{I}}^{n}
 +
 \beta_{i\cancel{i}'}^{n}
 )
 \bigr\}
 }.
\label{G24_n_relation_sum}
\end{align}
That is, \smash{$T^{n}_{\vec{\alpha}_{n},\vec{\beta}_{n}^{\ast}}$} given by \eqref{G24_n_relation_sum} gives a star product with separation of variables on $G_{2,4}(\mathbb{C})$.
\end{Proposition}

\begin{proof}
First, we derive \eqref{G24_n_relation_sum} from \eqref{G24_n_relation} in Proposition~\ref{Prop_G24_relation}. Multiplying both sides of \eqref{G24_n_relation} by
\smash{$
 \tau_{n}
 +
 \beta_{\cancel{I}}^{n}
 +
 \beta_{i\cancel{i}'}^{n}
 +
 1
$},
we have
\begin{gather}
 \beta_{I}^{n}
 (
 \tau_{n}
 +
 \beta_{\cancel{I}}^{n}
 +
 \beta_{i\cancel{i}'}^{n}
 +
 1
 )
 T^{n}_{\vec{\alpha}_{n},\vec{\beta}_{n}^{\ast}}
 \nonumber\\
 \qquad=
 \tau_{n}^{-1}
 \sum_{
 D
 \in
 \mathcal{I}
 }
 \sum_{k=1}^{2}
 (
 \tau_{n}
 \delta_{ik}
 +
 \beta_{i\cancel{i}'}^{n}
 +
 1
 )
 g_{\overline{ki'},D}
 T^{n-1}_{
 \vec{\alpha}_{n}
 -
 \vec{e}_{D},
 \vec{\beta}_{n}^{\ast}
 -
 \vec{e}_{I}^{\ast}
 -
 \delta_{\cancel{i}k}
 (
 \vec{e}_{\cancel{I}}^{\ast}
 -
 \vec{e}_{i\cancel{i}'}^{\ast}
 )
 }
 .
 \label{G24_n_relation_modified}
\end{gather}
Summing for $I$ on both sides of \eqref{G24_n_relation_modified}, we obtain
\begin{gather}
 \sum_{
 I
 \in
 \mathcal{I}
 }
 \beta_{I}^{n}
 (
 \tau_{n}
 +
 \beta_{\cancel{I}}^{n}
 +
 \beta_{i\cancel{i}'}^{n}
 +
 1
 )
 T^{n}_{\vec{\alpha}_{n},\vec{\beta}_{n}^{\ast}}\nonumber\\
 \qquad=
 \tau_{n}^{-1}
 \sum_{
 I,
 D
 \in
 \mathcal{I}
 }
 \sum_{k=1}^{2}
 (
 \tau_{n}
 \delta_{ik}
 +
 \beta_{i\cancel{i}'}^{n}
 +
 1
 )
 g_{\overline{ki'},D}
 T^{n-1}_{
 \vec{\alpha}_{n}
 -
 \vec{e}_{D},
 \vec{\beta}_{n}^{\ast}
 -
 \vec{e}_{I}^{\ast}
 -
 \delta_{\cancel{i}k}
 (
 \vec{e}_{\cancel{I}}^{\ast}
 -
 \vec{e}_{i\cancel{i}'}^{\ast} ) } . \label{G24_n_relation_modified_sum}
\end{gather}
Here, we use the following identities:
\begin{align}
 \sum_{
 J
 \in
 \mathcal{I}
 }
 \beta_{J}^{n}
 (
 \tau_{n}
 +
 \beta_{\cancel{J}}^{n}
 +
 \beta_{j\cancel{j}'}^{n}
 +
 1
 )
& =
 n
 (
 \tau_{n}+1
 )
 +
 2
 (
 \beta_{11'}^{n}
 +
 \beta_{21'}^{n}
 )
 (
 \beta_{22'}^{n}
 +
 \beta_{12'}^{n}
 )
 \nonumber
 \\
 &=
 n
 (
 \tau_{n}+1
 )
 +
 2
 (
 \beta_{I}^{n}
 +
 \beta_{\cancel{i}i'}^{n}
 )
 (
 \beta_{\cancel{I}}^{n}
 +
 \beta_{i\cancel{i}'}^{n}
 ).
 \label{sum_identity_beta}
\end{align}
Note that
\[
 (
 \beta_{I}^{n}
 +
 \beta_{\cancel{i}i'}^{n}
 )
 (
 \beta_{\cancel{I}}^{n}
 +
 \beta_{i\cancel{i}'}^{n}
 )
 =
 (
 \beta_{11'}^{n}
 +
 \beta_{21'}^{n}
 )
 (
 \beta_{22'}^{n}
 +
 \beta_{12'}^{n}
 )
\]
does not depend on how we choose fixed~$I$
from
$
\mathcal{I}
=
\{
11',
12',
21',
22'
\}
$.
See Appendix~\ref{appendix_symmetry_function} for more details. Substituting \eqref{sum_identity_beta} into~\eqref{G24_n_relation_modified_sum}, \eqref{G24_n_relation_modified_sum} can be rewritten as \eqref{G24_n_relation_sum}.

According to the construction method proposed by \cite{HS1,HS2} (or equivalently \cite{SSU1,SSU2}), there is the solution
\smash{$T\raisebox{2pt}{${}^{n}_{\vec{\alpha}_{n},\vec{\beta}_{n}^{\ast}}$}$}
of \eqref{rec_rel_Grassmann_2_4_slash}.
Furthermore, the solution \smash{$T\raisebox{2pt}{${}^{n}_{\vec{\alpha}_{n},\vec{\beta}_{n}^{\ast}}$}$} of \eqref{rec_rel_Grassmann_2_4_slash} is determined by~\eqref{G24_n_relation}, which are equivalent to \eqref{rec_rel_Grassmann_2_4_slash} in Proposition~\ref{Prop_G24_relation}.
The equivalence of \eqref{G24_n_relation_sum} with \eqref{G24_n_relation} in Proposition~\ref{Prop_G24_relation} is not clear for the moment. In other words, it is not clear whether the solution of \eqref{G24_n_relation_sum} is \smash{$T\raisebox{2pt}{${}^{n}_{\vec{\alpha}_{n},\vec{\beta}_{n}^{\ast}}$}$} of the star product with separation of variables.
However, \eqref{G24_n_relation_sum} are derived from \eqref{rec_rel_Grassmann_2_4_slash} in Proposition~\ref{Prop_G24_relation}, so the general term \smash{$T\raisebox{2pt}{${}^{n}_{\vec{\alpha}_{n},\vec{\beta}_{n}^{\ast}}$}$} satisfying~\eqref{rec_rel_Grassmann_2_4_slash} satisfies~\eqref{G24_n_relation_sum}.
On the other hand, \eqref{G24_n_relation_sum} states that if the
coefficients of order $(n-1)$ are given, then the coefficient~\smash{$T\raisebox{2pt}{${}^{n}_{\vec{\alpha}_{n},\vec{\beta}_{n}^{\ast}}$}$} can be completely determined.
In fact, since \smash{$T\raisebox{2pt}{${}_{\vec{0},\vec{0}^{\ast}}^{0}$}$} is given by \eqref{init_cond_01} in Proposi\-tion~\ref{prop:covar_const_0and1}, \smash{$T\raisebox{2pt}{${}^{n}_{\vec{\alpha}_{n},\vec{\beta}_{n}^{\ast}}$}$}~can be determined by solving \eqref{G24_n_relation_sum} sequentially from the $n=1$ case.
By applying \eqref{G24_n_relation_sum} sequentially, the solution of \eqref{G24_n_relation_sum}, i.e., \smash{$T\raisebox{2pt}{${}^{n}_{\vec{\alpha}_{n},\vec{\beta}_{n}^{\ast}}$}$}, is explicitly and uniquely determined.
Hence, \smash{$T\raisebox{2pt}{${}^{n}_{\vec{\alpha}_{n},\vec{\beta}_{n}^{\ast}}$}$} obtained as the solution of \eqref{G24_n_relation_sum} is also the solution of \eqref{rec_rel_Grassmann_2_4_slash}, i.e., it gives a star product with separation of variables.
\end{proof}

\subsection[Solution of the recurrence relations for G\_\{2,4\}(C)]{Solution of the recurrence relations for $\boldsymbol{G_{2,4}(\mathbb{C})}$}
\label{G24_recrel_solve}

We give the explicit expression of \smash{$T\raisebox{2pt}{${}^{n}_{\vec{\alpha}_{n},\vec{\beta}_{n}^{\ast}}$}$} as the solution of \eqref{rec_rel_Grassmann_2_4_slash} in this subsection.
\smash{$T\raisebox{2pt}{${}^{n}_{\vec{\alpha}_{n},\vec{\beta}_{n}^{\ast}}$}$}
corresponds
well to a Fock space representation,
since
\smash{$
T\raisebox{2pt}{${}^{n}_{\vec{\alpha}_{n},\vec{\beta}_{n}^{\ast}}$}
=
0
$}
when \smash{$\vec{\alpha}_{n}$} and \smash{$\vec{\beta}_{n}^{\ast}$} contain negative components, and the possible \smash{$\vec{\alpha}_{n}$} and \smash{$\vec{\beta}_{n}^{\ast}$} increase infinitely with increasing order $n$. To carry out the calculations more simply, let us introduce a linear operator $T_{n}$
realizing \smash{$T\raisebox{2pt}{${}^{n}_{\vec{\alpha}_{n},\vec{\beta}_{n}^{\ast}}$}$} as the Fock space representation.

First, we also denote $\vec{\alpha}_{n}$ by
\begin{align*}
\vec{\alpha}_{n}
:=
\alpha_{I}^{n}
\vec{e}_{I}
+
\alpha_{\cancel{I}}^{n}
\vec{e}_{\cancel{I}}
+
\alpha_{i\cancel{i}'}^{n}
\vec{e}_{i\cancel{i}'}
+
\alpha_{\cancel{i}i'}^{n}
\vec{e}_{\cancel{i}i'}
.
\end{align*}
To define a linear operator $T_{n}$, let us construct the Fock space. We introduce the vector space~$V$ over $\mathbb{C}$ by using the basis
$
\big|
\vec{m}
\big\rangle
:=
\bigl|
\sum_{
J
\in
\mathcal{I}
}
m_{J}
\vec{e}_{J}
\big\rangle
$
\begin{align*}
V
:=
\bigg\{
\sum_{
\vec{m}
}
c_{
\vec{m}
}
|
\vec{m}
\rangle
\,\bigg|\,
c_{
\vec{m}
}
\in
\mathbb{C}
\bigg\},
\end{align*}
where
$
m_{J}
\in
\mathbb{Z}_{\geq0}
$
and
$
\mathcal{I}
:=
\{
I,
\cancel{I},
i\cancel{i}',
\cancel{i}i'
\}
$.
$
|
\vec{m}
\rangle
$
is called ``ket'' (of $V$). In particular, the ket whose each $m_{J}$ is 0 is called ``vacuum'' and denoted by
\smash{$
|
 \vec{0}
\rangle
$}.
Here, we introduce the creation and annihilation operators. For any capital letter index $I$, the creation and annihilation operators~$a_{I}$ and \smash{$a_{I}^{\dagger}$} are defined as
linear operators on $V$ satisfying the following relations:
\begin{align*}
 a_{I}
 |
 \vec{m}
 \rangle
 =
 \sqrt{m_{I}}
|
 \vec{m}
 -
 \vec{e}_{I}
\rangle,
 \qquad
 a_{I}
 |\vec{0}\rangle
 =
 0,
 \qquad
 a_{I}^{\dagger}
|
 \vec{m}
\rangle
 =
 \sqrt{m_{I}+1}
|
 \vec{m}
 +
 \vec{e}_{I}
\rangle
 .
\end{align*}
By using the above $a_{I}$ and \smash{$a_{I}^{\dagger}$}, a number operator is defined by
\smash{$
N_{I}
:=
a_{I}^{\dagger}
a_{I}
$}.
$N_{I}$ plays the role of multiplying
by the $I$-th component
of
$
|
 \vec{m}
\rangle
$
,
i.e.,
$
 N_{I}
|
 \vec{m}
\rangle
 =
 m_{I}
|
 \vec{m}
\rangle$.
By definition of $a_{I}$, \smash{$a_{I}^{\dagger}$} and~$N_{I}$, the following relations:
\begin{gather*}
\begin{split}
& [
 a_{I},
 a_{J}
 ]
 =
 0,
\qquad
 \big[
 a_{I}^{\dagger},
 a_{J}^{\dagger}
 \big]
 =
 0,
\qquad
 \big[
 a_{I},
 a_{J}^{\dagger}
 \big]
 =
 \delta_{IJ},
 \\
& [
 a_{I},
 N_{J}
 ]
 =
 \delta_{IJ}a_{J}
 =
 \delta_{IJ}a_{I}
 ,
\qquad
 \big[
 a_{I}^{\dagger},
 N_{J}
 \big]
 =
 -\delta_{IJ}a_{J}^{\dagger}
 =
 -\delta_{IJ}a_{I}^{\dagger}
 ,
\qquad
 [
 N_{I},
 N_{J}
 ]
 =
 0,
\end{split}
\end{gather*}
holds between them for any ket
$
|
\vec{m}
\rangle
$
and any capital letter indices $I$, $J$. For a ket
$
|
\vec{m}
\rangle
$
such that~${m_{I}>0}$, the operator
\smash{$
\frac{1}{
 \sqrt{N_{I}}
}
$}
can be defined by
\[
 \frac{1}{
 \sqrt{N_{I}}
 }
|
 \vec{m}
\rangle
 :=
 \frac{1}{
 \sqrt{m_{I}}
 }
|
 \vec{m}
\rangle
.
\]
We next consider the dual vector space $V^{\ast}$ of $V$. We denote the dual basis of $V^{\ast}$ by
$
\langle
\vec{n}
|
$.
$
\langle
\vec{n}
|
$
satisfies
$
\langle
\vec{n}
|
\vec{m}
\rangle
:=
\delta_{
\vec{n},
\vec{m}
}
$.
$
\langle
\vec{n}
|
$
is called ``bra''.
By taking the Hermitian conjugate of
$
\langle
\vec{n}
|
$,
a~bra~$
\langle
\vec{n}
|
$
becomes a~ket~$
|
\vec{n}
\rangle
$,
i.e.,
\smash{$
|
\vec{n}
\rangle
=
(
\langle
\vec{n}
|
)^{\dagger}
$}.
Let us now consider an operator
\smash{$
a_{I}
\frac{1}{
\sqrt{N_{I}}
}
$}.
Note that since
\begin{align*}
\langle
\vec{m}
|
a_{I}
\frac{1}{
\sqrt{N_{I}}
}
=
 \left(
\frac{1}{
\sqrt{N_{I}}
}
a_{I}^{\dagger}
|
\vec{m}
\rangle
 \right)^{\dagger}
\end{align*}
for any bra
$
\langle
\vec{m}
|
$,
\smash{$
a_{I}
\frac{1}{
\sqrt{N_{I}}
}
$}
can always be defined for any ket
$
|
\vec{n}
\rangle
$
if there exists a bra
$
\langle
\vec{m}
|
$
on the left. For such an operator
$
a_{I}
\frac{1}{
\sqrt{N_{I}}
}
$,
the following commutation relations hold
\begin{gather*}
 \left[
 a_{I}
 \frac{1}{
 \sqrt{N_{I}}
 },
 a_{J}
 \frac{1}{
 \sqrt{N_{J}}
 }
 \right]
|
 \vec{m}
\rangle
 =
 0,
 \qquad
 m_{I},
 m_{J}
 >
 \delta_{IJ}
 ,\nonumber
 \\
 \left[
 a_{I}
 \frac{1}{
 \sqrt{N_{I}}
 },
 a_{J}^{\dagger}
 \frac{1}{
 \sqrt{N_{J}+1}
 }
 \right]
|
 \vec{m}
\rangle
 =
 0,
 \qquad
 m_{I}>0,
 \qquad
 \left[
 a_{I}^{\dagger}
 \frac{1}{
 \sqrt{N_{I}+1}
 },
 a_{J}^{\dagger}
 \frac{1}{
 \sqrt{N_{J}+1}
 }
 \right]
|
 \vec{m}
\rangle
 =
 0.
\end{gather*}
Moreover, let \smash{$f\bigl(\vec{N}\bigr)$} be an arbitrary operator that depends on \smash{$\vec{N}$}, where \smash{$\vec{N}$} is defined by
\begin{align*}
\vec{N}
:=
N_{I}
\vec{e}_{I}
+
N_{\cancel{I}}
\vec{e}_{\cancel{I}}
+
N_{i\cancel{i}'}
\vec{e}_{i\cancel{i}'}
+
N_{\cancel{i}i'}
\vec{e}_{\cancel{i}i'}
.
\end{align*}
Then, \smash{$f\bigl(\vec{N}\bigr)$} holds the following relations:
\begin{gather}
f\bigl(\vec{N}\bigr)
a_{I}
\frac{1}{
\sqrt{N_{I}}
}
|
\vec{m}
\rangle
=
a_{I}
\frac{1}{
\sqrt{N_{I}}
}
f\bigl(
\vec{N}
-
\vec{e}_{I}
\bigr)
|
\vec{m}
\rangle,
\qquad
m_{I}>0,
\label{relation_num_opr_fct1}
\\
f\bigl(\vec{N}\bigr)
a_{I}^{\dagger}
\frac{1}{
\sqrt{N_{I}+1}
}
|
\vec{m}
\rangle
=
a_{I}^{\dagger}
\frac{1}{
\sqrt{N_{I}+1}
}
f
\bigl(
\vec{N}
+
\vec{e}_{I}
\bigr)
|
\vec{m}
\rangle.\label{relation_num_opr_fct2}
\end{gather}
Note that for any
$
\vec{m}
=
m_{I}
\vec{e}_{I}
+
m_{\cancel{I}}
\vec{e}_{\cancel{I}}
+
m_{i\cancel{i}'}
\vec{e}_{i\cancel{i}'}
+
m_{\cancel{i}i'}
\vec{e}_{\cancel{i}i'}
$,
$
a_{J}
|
\vec{m}
\rangle
$
can also be expressed as
\begin{align}
a_{J}
|
\vec{m}
\rangle
=
\sqrt{m_{J}}
|
\vec{m}
-
\vec{e}_{J}
\rangle
=
 \bigg(
\prod_{
L
\in
\mathcal{I}
}
\theta
(
m_{L}
-
\delta_{LJ}
)
 \bigg)
\sqrt{m_{J}}
|
\vec{m}
-
\vec{e}_{J}
\rangle
,\qquad
J\in\mathcal{I},\label{braket_ann_step_fct}
\end{align}
where
$
\theta
\colon
\mathbb{R}
\to
\{
0,1
\}
$
is the step function such that
\begin{align*}
\theta(x)
:=
\begin{cases}
1,
&
x\geq0,
\\
0,
&
x<0.
\end{cases}
\end{align*}
The above expression using the step function $\theta$ is used in calculations in Appendix~\ref{appendix_coeff_T}.

We are now ready to define a linear operator $T_{n}$. We define a linear operator $T_{n}$ on $V$ by
\begin{align}
T_{n}
:=
\sum_{\vec{\alpha}_{n},\vec{\beta}_{n}^{\ast}}
T^{n}_{\vec{\alpha}_{n},\vec{\beta}_{n}^{\ast}}
|
\vec{\alpha}_{n}^{\ast}
\rangle
\langle
\vec{\beta}_{n}|.\label{braket_coeff}
\end{align}
\smash{$T\raisebox{2pt}{${}^{n}_{\vec{\alpha}_{n},\vec{\beta}_{n}^{\ast}}$}$} is a matrix representation of $T_{n}$ using a Fock space $V$.
$
T_{0}
=
|
\vec{0}
\rangle
\langle
\vec{0}
|
$
is given from \eqref{init_cond_01} in Proposition~\ref{prop:covar_const_0and1}.
By Proposition~\ref{Cor_G24_relation_sum}, $T_{n}$ can be expressed by using a linear operator $T_{n-1}$ of order $(n-1)$ as follows:
\begin{align}
 T_{n}
 ={}&
 \sum_{
 J,D
 \in
 \mathcal{I}
 }
 a_{D}^{\dagger}
 \frac{1}{
 \sqrt{
 N_{D}
 +
 1
 }
 }
 T_{n-1}\nonumber\\
 &\times
 \left\{
 a_{J}
 \frac{1}{\sqrt{N_{J}}}
 (
 \tau_{n}
 +
 N_{j\cancel{j}'}
 +
 1
 )
 g_{\overline{J}D}
 +
 a_{J}
 a_{\cancel{J}}
 a_{j\cancel{j}'}^{\dagger}
 \frac{N_{j\cancel{j}'}+1}{
 \sqrt{
 N_{J}
 N_{\cancel{J}}
 (
 N_{j\cancel{j}'}+1
 )
 }
 }
 g_{\overline{\cancel{j}j'},D}
 \right\}
 \nonumber
 \\
 &
 \times
 \tau_{n}^{-1}
 \left\{
 n
 (
 \tau_{n}
 +
 1
 )
 +
 2
 (
 N_{I}
 +
 N_{\cancel{i}i'}
 )
 (
 N_{\cancel{I}}
 +
 N_{i\cancel{i}'}
 )
 \right\}^{-1}. \label{G24_coeff_opr_full}
\end{align}
By applying \eqref{G24_coeff_opr_full} $(n-1)$ times,
the following theorem is obtained.

\begin{Theorem}\label{thm_left_coeff_opr}
A linear operator $T_{n}$ is explicitly given by
\[
 T_{n}
 =
 \sum_{
 \substack{
 J_{i}
 \in
 \{J_{i}\}_{n}\\
 D_{i}
 \in
 \{D_{i}\}_{n}
 }
 }
 \sum_{
 \substack{
 k_{i}=1
 \\
 k_{i}
 \in
 \{
 k_{i}
 \}_{n}
 }
 }^{2}
 a_{D_{n}}^{\dagger}
 \frac{1}{
 \sqrt{
 N_{D_{n}}
 +
 1
 }
 }
 \cdots
 a_{D_{1}}^{\dagger}
 \frac{1}{
 \sqrt{
 N_{D_{1}}
 +
 1
 }
 }
 T_{0}
 \mathcal{B}_{
 1,
 J_{1},
 k_{1}
 }
 \cdots
 \mathcal{B}_{
 n,
 J_{n},
 k_{n}
 },
\]
where
$
T_{0}
=
|
\vec{0}
\rangle
\langle
\vec{0}
|
$,
\begin{align*}
 \mathcal{B}_{
 l,
 J_{l},
 k_{l}
 }
 :={}&
 a_{J_{l}}
 \frac{1}{\sqrt{N_{J_{l}}}}
 \left(
 a_{\cancel{J_{l}}}
 \frac{1}{
 \sqrt{
 N_{\cancel{J_{l}}}
 }
 }
 a_{j_{l}\cancel{j_{l}}'}^{\dagger}
 \frac{1}{
 \sqrt{
 N_{j_{l}\cancel{j_{l}}'}+1
 }
 }
 \right)^{\delta_{\cancel{j_{l}}k_{l}}}
 (
 \tau_{l}
 \delta_{j_{l}k_{l}}
 +
 N_{j_{l}\cancel{j_{l}}'}
 +
 1
 )
 \frac{
 g_{\overline{k_{l}{j_{l}}'},D_{l}}
 }{
 \tau_{l}
 }
 \nonumber
 \\
 &
 \times
 \left\{
 l
 (
 \tau_{l}
 +
 1
 )
 +
 2
 (
 N_{I}
 +
 N_{\cancel{i}i'}
 )
 (
 N_{\cancel{I}}
 +
 N_{i\cancel{i}'}
 )
 \right\}^{-1}
\end{align*}
for $l=1,\dots,n$, and summations
\smash{$
 \sum_{
 \substack{
 J_{i}
 \in
 \{J_{i}\}_{n}\\
 D_{i}
 \in
 \{D_{i}\}_{n}
 }
 }
$}
and
\smash{$
 \sum_{
 \substack{
 k_{i}=1
 \\
 k_{i}
 \in
 \{
 k_{i}
 \}_{n}
 }
 }^{2}
$}
are defined as
\begin{align*}
 \sum_{
 \substack{
 J_{i}
 \in
 \{J_{i}\}_{n}\\
 D_{i}
 \in
 \{D_{i}\}_{n}
 }
 }
 :=
 \sum_{
 J_{1}
 \in
 \mathcal{I}
 }
 \cdots
 \sum_{
 J_{n}
 \in
 \mathcal{I}
 }
 \sum_{
 D_{1}
 \in
 \mathcal{I}
 }
 \cdots
 \sum_{
 D_{n}
 \in
 \mathcal{I}
 },
 \qquad
 \sum_{
 \substack{
 k_{i}=1
 \\
 k_{i}
 \in
 \{
 k_{i}
 \}_{n}
 }
 }^{2}
 :=
 \sum_{
 k_{1}
 =1
 }^{2}
 \cdots
 \sum_{
 k_{n}
 =
 1
 }^{2},
\end{align*}
respectively.
\end{Theorem}

Note that
\begin{align*}
a_{J}
a_{\cancel{J}}
a_{j\cancel{j}'}^{\dagger}
\frac{1}{
\sqrt{
N_{J}
N_{\cancel{J}}
 (
N_{j\cancel{j}'}+1
 )
}
}
=
a_{J}
\frac{1}{\sqrt{N_{J}}}
a_{\cancel{J}}
\frac{1}{\sqrt{N_{\cancel{J}}}}
a_{j\cancel{j}'}^{\dagger}
\frac{1}{\sqrt{N_{j\cancel{j}'}+1}}
\end{align*}
for any
$
J
\in
\mathcal{I}
$,
since
$
j\cancel{j}'
\neq
J
(=jj')
$,
$
\cancel{J}
\neq
J
(=jj')
$,
and
$
j\cancel{j}'
\neq
\cancel{J}
$.
By using the relations \eqref{relation_num_opr_fct1} and \eqref{relation_num_opr_fct2}, we obtain
\begin{gather*}
 (
 \tau_{l}
 \delta_{j_{l}k_{l}}
 +
 N_{j_{l}\cancel{j_{l}}'}
 +
 1
 )
 a_{J_{l+1}}
 \frac{1}{\sqrt{N_{J_{l+1}}}}
 \left(
 a_{\cancel{J_{l+1}}}
 \frac{1}{
 \sqrt{
 N_{\cancel{J_{l+1}}}
 }
 }
 a_{j_{l+1}\cancel{j_{l+1}}'}^{\dagger}
 \frac{1}{
 \sqrt{
 N_{j_{l+1}\cancel{j_{l+1}}'}+1
 }
 }
 \right)^{\delta_{\cancel{j_{l+1}}k_{l+1}}}
 \\
 \qquad=
 a_{J_{l+1}}
 \frac{1}{\sqrt{N_{J_{l+1}}}}
 \left(
 a_{\cancel{J_{l+1}}}
 \frac{1}{
 \sqrt{
 N_{\cancel{J_{l+1}}}
 }
 }
 a_{j_{l+1}\cancel{j_{l+1}}'}^{\dagger}
 \frac{1}{
 \sqrt{
 N_{j_{l+1}\cancel{j_{l+1}}'}+1
 }
 }
 \right)^{\delta_{\cancel{j_{l+1}}k_{l+1}}}
 \\
\phantom{ \qquad=}{} \times
 [
 \tau_{l}
 \delta_{j_{l}k_{l}}
 +
 N_{j_{l}\cancel{j_{l}}'}
 -
 \{
 \delta_{
 j_{l}\cancel{j_{l}}',
 J_{l+1}
 }
 +
 \delta_{\cancel{j_{l+1}}k_{l+1}}
 (
 \delta_{
 j_{l}\cancel{j_{l}}',
 \cancel{J_{l+1}}
 }
 -
 \delta_{
 j_{l}\cancel{j_{l}},
 j_{l+1}\cancel{j_{l+1}}'
 }
 )
 \}
 +
 1
 ],
\end{gather*}
and
\begin{gather*}
 \{
 l
 (
 \tau_{l}
 +
 1
 )
 +
 2
 (
 N_{I}
 +
 N_{\cancel{i}i'}
 )
 (
 N_{\cancel{I}}
 +
 N_{i\cancel{i}'}
 )
 \}^{-1}
 a_{J_{l+1}}\\
 \qquad\times
 \frac{1}{\sqrt{N_{J_{l+1}}}}
 \left(
 a_{\cancel{J_{l+1}}}
 \frac{1}{
 \sqrt{
 N_{\cancel{J_{l+1}}}
 }
 }
 a_{j_{l+1}\cancel{j_{l+1}}'}^{\dagger}
 \frac{1}{
 \sqrt{
 N_{j_{l+1}\cancel{j_{l+1}}'}+1
 }
 }
 \right)^{\delta_{\cancel{j_{l+1}}k_{l+1}}}
 \\
 \phantom{ \qquad\times}{}
 =
 a_{J_{l+1}}
 \frac{1}{\sqrt{N_{J_{l+1}}}}
 \left(
 a_{\cancel{J_{l+1}}}
 \frac{1}{
 \sqrt{
 N_{\cancel{J_{l+1}}}
 }
 }
 a_{j_{l+1}\cancel{j_{l+1}}'}^{\dagger}
 \frac{1}{
 \sqrt{
 N_{j_{l+1}\cancel{j_{l+1}}'}+1
 }
 }
 \right)^{\delta_{\cancel{j_{l+1}}k_{l+1}}}
 \\
 \phantom{ \qquad\times=}{}
 \times
 [
 l
 (
 \tau_{l}
 +
 1
 )
 +
 2
 \{
 N_{I}
 +
 N_{\cancel{i}i'}
 -
 (
 \delta_{IJ_{l+1}}
 +
 \delta_{\cancel{i}i',J_{l+1}}
 )
 \}\\
 \phantom{ \qquad=\times
 \times}{}\times
 \{
 N_{\cancel{I}}
 +
 N_{i\cancel{i}'}
 -
 (
 \delta_{\cancel{I}J_{l+1}}
 +
 \delta_{i\cancel{i}',J_{l+1}}
 )
 \}
 ]^{-1}
 .
\end{gather*}
Here we use
\begin{gather*}
 \delta_{J_{l+1}I}
 +
 \delta_{\cancel{j_{l+1}}k_{l+1}}
 (
 \delta_{\cancel{J_{l+1}},I}
 -
 \delta_{j_{l+1}\cancel{j_{l+1}}',I}
 )
 +
 \delta_{J_{l+1},\cancel{i}i'}
 +
 \delta_{\cancel{j_{l+1}}k_{l+1}}
 (
 \delta_{\cancel{J_{l+1}},\cancel{i}i'}
 -
 \delta_{j_{l+1}\cancel{j_{l+1}}',\cancel{i}i'}
 )\\
 \qquad
 =
 \delta_{J_{l+1}I}
 +
 \delta_{J_{l+1},\cancel{i}i'},
 \\
 \delta_{J_{l+1}\cancel{I}}
 +
 \delta_{\cancel{j_{l+1}}k_{l+1}}
 (
 \delta_{\cancel{J_{l+1}},\cancel{I}}
 -
 \delta_{j_{l+1}\cancel{j_{l+1}}',\cancel{I}}
 )
 +
 \delta_{J_{l+1},i\cancel{i}'}
 +
 \delta_{\cancel{j_{l+1}}k_{l+1}}
 (
 \delta_{\cancel{J_{l+1}},i\cancel{i}'}
 -
 \delta_{j_{l+1}\cancel{j_{l+1}}',i\cancel{i}'}
 )\\
 \qquad
 =
 \delta_{J_{l+1}\cancel{I}}
 +
 \delta_{J_{l+1},i\cancel{i}'}
\end{gather*}
implied from the following identities
\begin{align*}
 \delta_{\cancel{J_{l+1}},I}
 +
 \delta_{\cancel{J_{l+1}},\cancel{i}i'}
 -
 \delta_{j_{l+1}\cancel{j_{l+1}}',I}
 -
 \delta_{j_{l+1}\cancel{j_{l+1}}',\cancel{i}i'}
 &=
 0,
 \\
 \delta_{\cancel{J_{l+1}},\cancel{I}}
 +
 \delta_{\cancel{J_{l+1}},i\cancel{i}'}
 -
 \delta_{j_{l+1}\cancel{j_{l+1}}',\cancel{I}}
 -
 \delta_{j_{l+1}\cancel{j_{l+1}}',i\cancel{i}'}
 &=
 0.
\end{align*}
\noindent
Thus,
$
\mathcal{B}_{
1,
J_{1},
k_{1}
}
\cdots
\mathcal{B}_{
n,
J_{n},
k_{n}
}
$
can be rewritten as
\begin{align*}
\mathcal{B}_{
1,
J_{1},
k_{1}
}
\cdots
\mathcal{B}_{
n,
J_{n},
k_{n}
}
={}&
\frac{
 g_{\overline{k_{1}{j_{1}}'},D_{1}}
 \cdots
 g_{\overline{k_{n}{j_{n}}'},D_{n}}
}{
 \tau_{1}
 \cdots
 \tau_{n}
}
\mathcal{A}_{
J_{1},
k_{1}
}
\cdots
\mathcal{A}_{
J_{n},
k_{n}
}
\mathcal{C}_{
1,
\{
 J_{i}
\}_{n},
\{
 k_{i}
\}_{n}
}\\
&
\cdots
\mathcal{C}_{
n,
\{
 J_{i}
\}_{n},
\{
 k_{i}
\}_{n}
}
\mathcal{F}_{
1,
\{
 J_{i}
\}_{n},
\{
 k_{i}
\}_{n}
}
\cdots
\mathcal{F}_{
n,
\{
 J_{i}
\}_{n},
\{
 k_{i}
\}_{n}
}
.
\end{align*}
Here
$
\mathcal{A}_{
J_{l},
k_{l}
}
$,
\smash{$
\mathcal{C}_{
l,
\{
 J_{i}
\}_{n},
\{
 k_{i}
\}_{n}
}
$}
and
\smash{$
\mathcal{F}_{
l,
\{
 J_{i}
\}_{n},
\{
 k_{i}
\}_{n}
}
$}
are defined as
\begin{gather*}
\mathcal{A}_{
J_{l},
k_{l}
}
:=
a_{J_{l}}
\frac{1}{\sqrt{N_{J_{l}}}}
 \left(
 a_{\cancel{J_{l}}}
 \frac{1}{
 \sqrt{
 N_{\cancel{J_{l}}}
 }
 }
 a_{j_{l}\cancel{j_{l}}'}^{\dagger}
 \frac{1}{
 \sqrt{
 N_{j_{l}\cancel{j_{l}}'}+1
 }
 }
 \right)^{\delta_{\cancel{j_{l}}k_{l}}}
,
\\
\mathcal{C}_{
l,
\{
 J_{i}
\}_{n},
\{
 k_{i}
\}_{n}
}
:=
 (
 \tau_{l}
 \delta_{j_{l}k_{l}}
 +
 N_{j_{l}\cancel{j_{l}}'}
 -
 \Lambda_{
 l,
 j_{l}\cancel{j_{l}}',
 \{
 J_{i}
 \}_{n},
 \{
 k_{i}
 \}_{n}
 }
 +
 1
 ),
\\
\mathcal{F}_{
l,
\{
 J_{i}
\}_{n},
\{
 k_{i}
\}_{n}
}
:=
\{
 l
 (
 \tau_{l}
 +
 1
 )
 +
 2
 (
 N_{I}
 +
 N_{\cancel{i}i'}
 -
 \Delta_{
 I,
 \cancel{i}i',
 l,
 \{
 J_{i}
 \}_{n}
 }
 )
 (
 N_{\cancel{I}}
 +
 N_{i\cancel{i}'}
 -
 \Delta_{
 \cancel{I},
 i\cancel{i}',
 l,
 \{
 J_{i}
 \}_{n}
 }
 )
\}^{-1},
\end{gather*}
where
\begin{gather}
\Lambda_{
 l,
 S,
 \{
 J_{i}
 \}_{n},
 \{
 k_{i}
 \}_{n}
}
:=
\sum_{m=1}^{n}
d_{
 S,
 J_{m},
 k_{m}
}
-
\sum_{m=1}^{l}
 d_{
 S,
 J_{m},
 k_{m}
}
=
\sum_{m=l+1}^{n}
d_{
 S,
 J_{m},
 k_{m}
}
,
\label{notation_Lambda}
\\
d_{
 S,
 J_{m},
 k_{m}
}
:=
\delta_{
 S,
 J_{m}
}
+
\delta_{\cancel{j_{m}}k_{m}}
 (
 \delta_{
 S,
 \cancel{J_{m}}
 }
 -
 \delta_{
 S,
 j_{m}\cancel{j_{m}}'
 }
 )
,
\label{notation_d}
\\
\Delta_{
I,
\cancel{i}i',
l,
\{
 J_{i}
\}_{n}
}
:=
\sum_{m=1}^{n}
 (
 \delta_{IJ_{m}}
 +
 \delta_{\cancel{i}i',J_{m}}
 )
-
\sum_{m=1}^{l}
 (
 \delta_{IJ_{m}}
 +
 \delta_{\cancel{i}i',J_{m}}
 )
=
\sum_{m=l+1}^{n}
 (
 \delta_{IJ_{m}}
 +
 \delta_{\cancel{i}i',J_{m}}
 )
,
\label{notation_delta_Iisi}
\\
\Delta_{
\cancel{I},
i\cancel{i}',
l,
\{
 J_{i}
\}_{n}
}
:=
\sum_{m=1}^{n}
 (
 \delta_{\cancel{I}J_{m}}
 +
 \delta_{i\cancel{i}',J_{m}}
 )
-
\sum_{m=1}^{l}
 (
 \delta_{\cancel{I}J_{m}}
 +
 \delta_{i\cancel{i}',J_{m}}
 )
=
\sum_{m=l+1}^{n}
 (
 \delta_{\cancel{I}J_{m}}
 +
 \delta_{i\cancel{i}',J_{m}}
 )
\label{notation_delta_Isiis}
\end{gather}
for $l=1,\dots,n$,
$
S
\in
\mathcal{I}
$,
$
\{
 J_{i}
\}_{n}
:=
\{
J_{1},
\dots,
J_{n}
\}
$
and
$
\{
 k_{i}
\}_{n}
:=
\{
k_{1},
\dots,
k_{n}
\}
$.
Note that
$
\mathcal{F}_{
l,
\{
 J_{i}
\}_{n},
\{
 k_{i}
\}_{n}
}
$
does not depend on the fixed index $I$ because of
Proposition~\ref{cor_app_sym_opr}.
Hence, we obtain the explicit expression for $T_{n}$ as the solution of
Proposition~\ref{Cor_G24_relation_sum}.

\begin{Theorem}\label{thm_G24_opr}
A linear operator $T_{n}$ is explicitly given by
\begin{align}
T_{n}
={}&
\sum_{
 \substack{
 J_{i}
 \in
 \{
 J_{i}
 \}_{n}
 \\
 D_{i}
 \in
 \{
 D_{i}
 \}_{n}
 }
}
\sum_{
 \substack{
 k_{i}=1
 \\
 k_{i}
 \in
 \{
 k_{i}
 \}_{n}
 }
}^{2}
a_{D_{n}}^{\dagger}
\frac{1}{
 \sqrt{
 N_{D_{n}}
 +
 1
 }
}
\cdots
a_{D_{1}}^{\dagger}
\frac{1}{
 \sqrt{
 N_{D_{1}}
 +
 1
 }
}
T_{0}
\nonumber
\\
&
\times
\mathcal{A}_{
J_{1},
k_{1}
}
\cdots
\mathcal{A}_{
J_{n},
k_{n}
}
\prod_{l=1}^{n}
 \left(
\frac{
 g_{\overline{k_{l}{j_{l}}'},D_{l}}
}{
 \tau_{l}
}
\mathcal{C}_{
l,
\{
 J_{i}
\}_{n},
\{
 k_{i}
\}_{n}
}
\mathcal{F}_{
l,
\{
 J_{i}
\}_{n},
\{
 k_{i}
\}_{n}
}
 \right),
\label{star_coeff_G2_4_Opr}
\end{align}
where
\smash{$
T_{0}
=
|
\vec{0}
\rangle
\langle
\vec{0}
|
$}.
\end{Theorem}

The coefficient
\smash{$T^{n}_{\vec{\alpha}_{n},\vec{\beta}_{n}^{\ast}}$} is recovered from \eqref{braket_coeff} as
\begin{align}
 T^{n}_{\vec{\alpha}_{n},\vec{\beta}_{n}^{\ast}}
 =
 \langle
 \vec{\alpha}_{n}
 |
 T_{n}
 |
 \vec{\beta}_{n}^{\ast}
 \rangle
 .
 \label{opr_T_norm}
\end{align}
By Theorem~\ref{thm_G24_opr}, we can explicitly give \smash{$T^{n}_{\vec{\alpha}_{n},\vec{\beta}_{n}^{\ast}}$}.

\begin{Theorem}\label{thm_star_coeff_G2_4}
The coefficient \smash{$T^{n}_{\vec{\alpha}_{n},\vec{\beta}_{n}^{\ast}}$} which determines a star product with separation of variables on $G_{2,4}(\mathbb{C})$ is given by
\begin{align}
 T^{n}_{\vec{\alpha}_{n},\vec{\beta}_{n}^{\ast}}
 ={}&
 \sum_{
 \substack{
 J_{i}
 \in
 \{
 J_{i}
 \}_{n}
 \\
 D_{i}
 \in
 \{
 D_{i}
 \}_{n}
 }
 }
 \sum_{
 \substack{
 k_{i}=1
 \\
 k_{i}
 \in
 \{
 k_{i}
 \}_{n}
 }
 }^{2}
 \delta_{
 \vec{\alpha}_{n},
 \sum_{m=1}^{n}
 \vec{e}_{D_{m}}
 }
 \delta_{
 \vec{\beta}_{n}^{\ast},
 \sum_{
 X
 \in
 \mathcal{I}
 }
 \sum_{m=1}^{n}
 d_{
 X,
 J_{m},
 k_{m}
 }
 \vec{e}_{X}^{\ast}
 }
 \nonumber
 \\
 &
 \times
 \left(
 \prod_{
 S
 \in
 \mathcal{I}
 }
 \prod_{r=1}^{n}
 \theta
 \left(
 \beta_{S}^{n}
 -
 \sum_{m=r}^{n}
 d_{S,J_{m},k_{m}}
 \right)
 \right)
 \left(
 \prod_{l=1}^{n}
 \frac{g_{\overline{k_{l}j_{l}'},D_{l}}}{\tau_{l}}
 \right)
 \nonumber
 \\
 &
 \times
 \left\{
 \prod_{l=1}^{n}
 \frac{
 \tau_{l}
 \delta_{j_{l}k_{l}}
 +
 \beta_{j_{l}\cancel{j_{l}}'}^{n}
 +
 1
 -
 \Lambda_{
 l,
 j_{l}\cancel{j_{l}}',
 \{
 J_{i}
 \}_{n},
 \{
 k_{i}
 \}_{n}
 }
 }{
 l
 (
 \tau_{l}+1
 )
 +
 2
 (
 \beta_{I}^{n}
 +
 \beta_{\cancel{i}i'}^{n}
 -
 \Delta_{
 I,
 \cancel{i}i',
 l,
 \{
 J_{i}
 \}_{n}
 }
 )
 (
 \beta_{\cancel{I}}^{n}
 +
 \beta_{i\cancel{i}'}^{n}
 -
 \Delta_{
 \cancel{I},
 i\cancel{i}',
 l,
 \{
 J_{i}
 \}_{n}
 }
 )
 }
 \right\}
 .
 \label{star_coeff_G2_4}
\end{align}
\end{Theorem}

See Appendix~\ref{appendix_coeff_T} for detailed calculations to obtain Theorem~\ref{thm_star_coeff_G2_4}. Note that \eqref{star_coeff_G2_4} does not depend on the capital letter index $I$.

Now that we are ready, let us begin the proof of the main theorem, Theorem~\ref{thm_star_prd_G2_4}.
Substituting \eqref{star_coeff_G2_4} into
\eqref{star_expand}, we can obtain the explicit star product with separation of variables on~$G_{2,4}(\mathbb{C})$. Recalling that the right-hand side of \eqref{star_coeff_G2_4} includes Kronecker's delta
\[
 \delta_{
 \vec{\alpha}_{n},
 \sum_{m=1}^{n}
 \vec{e}_{D_{m}}
 }
\qquad \text{and}\qquad
 \delta_{
 \vec{\beta}_{n}^{\ast},
 \sum_{
 P
 \in
 \mathcal{I}
 }
 \sum_{m=1}^{n}
 d_{
 X,
 J_{m},
 k_{m}
 }
 \vec{e}_{X}^{\ast} },
\]
then
$
\vec{\alpha}_{n},
$
and
\smash{$
\vec{\beta}_{n}^{\ast}
$}
that contribute to
\smash{$
\sum_{
\vec{\alpha}_{n},
\vec{\beta}_{n}^{\ast}
}
$}
in 
\eqref{star_expand} as non-zero are as follows:
\begin{align*}
 \vec{\alpha}_{n}
 =
 \sum_{m=1}^{n}
 \vec{e}_{D_{m}},
 \qquad
 \vec{\beta}_{n}^{\ast}
 =
 \sum_{
 X
 \in
 \mathcal{I}
 }
 \sum_{m=1}^{n}
 d_{
 X,
 J_{m},
 k_{m}
 }
 \vec{e}_{X}^{\ast}.
\end{align*}
In more detail, components of \smash{$\vec{\beta}_{n}^{\ast}$} are also rewritten as follows:
\begin{gather}
 \beta_{I}^{n}
 =
 \sum_{m=1}^{n}
 d_{
 I,
 J_{m},
 k_{m}
 },
 \qquad
 \beta_{\cancel{I}}^{n}
 =
 \sum_{m=1}^{n}
 d_{
 \cancel{I},
 J_{m},
 k_{m}
 },
 \qquad
 \beta_{i\cancel{i}'}^{n}
 =
 \sum_{m=1}^{n}
 d_{
 i\cancel{i}',
 J_{m},
 k_{m}
 },\nonumber
 \\
 \beta_{\cancel{i}i'}^{n}
 =
 \sum_{m=1}^{n}
 d_{
 \cancel{i}i',
 J_{m},
 k_{m}
 },
 \label{nontri_I}
\end{gather}
respectively.
From \eqref{notation_Lambda} and \eqref{nontri_I}, then we have
\begin{align*}
 \beta_{j_{l}\cancel{j_{l}}'}^{n}
 -
 \Lambda_{
 l,
 j_{l}\cancel{j_{l}}',
 \{
 J_{i}
 \}_{n},
 \{
 k_{i}
 \}_{n}
 }
 =&
 \sum_{m=1}^{l}
 d_{
 j_{l}\cancel{j_{l}}',
 J_{m},
 k_{m}
 }
 .
\end{align*}
From \eqref{notation_d}, \eqref{notation_delta_Iisi}, \eqref{nontri_I}, and using
$
\delta_{I\cancel{J_{m}}}
-
\delta_{I,j_{m}\cancel{j_{m}}'}
+
\delta_{\cancel{i}i',\cancel{J_{m}}}
-
\delta_{\cancel{i}i',j_{m}\cancel{j_{m}}'}
=
0
$,
\begin{gather*}
 (
 \beta_{I}^{n}
 +
 \beta_{\cancel{i}i'}^{n}
 -
 \Delta_{
 I,
 \cancel{i}i',
 l,
 \{
 J_{i}
 \}_{n}
 }
 )
 (
 \beta_{\cancel{I}}^{n}
 +
 \beta_{i\cancel{i}'}^{n}
 -
 \Delta_{
 \cancel{I},
 i\cancel{i}',
 l,
 \{
 J_{i}
 \}_{n}
 }
 )\\
 \qquad
 =
 \left\{
 \sum_{m=1}^{l}
 (
 \delta_{IJ_{m}}
 +
 \delta_{\cancel{i}i',J_{m}}
 )
 \right\}
 \left\{
 \sum_{m=1}^{l}
 (
 \delta_{\cancel{I}J_{m}}
 +
 \delta_{i\cancel{i}',J_{m}}
 )
 \right\}
\end{gather*}
for $l=1,\dots,n$.
In this case, the step function
$
 \theta
 \bigl(
 \beta_{S}^{n}
 -
 \sum_{m=r}^{n}
 d_{S,J_{m},k_{m}}
 \bigr)
$
in \eqref{star_coeff_G2_4} can be rewritten~as
\begin{align*}
 \theta
 \left(
 \beta_{S}^{n}
 -
 \sum_{m=r}^{n}
 d_{S,J_{m},k_{m}}
 \right)
 =
 \theta
 \left(
 \sum_{m=1}^{r-1}
 d_{
 S,
 J_{m},
 k_{m}
 }
 \right)
\end{align*}
from \eqref{notation_d} and \eqref{nontri_I}.
The derivations are also rewritten as
\begin{align*}
 D^{\vec{\alpha}_{n}}
 =
 D^{
 \sum_{m=1}^{n}
 \vec{e}_{D_{m}}
 },
 \qquad
 D^{\vec{\beta}_{n}^{\ast}}
 =
 D^{
 \sum_{
 X
 \in
 \mathcal{I}
 }
 \sum_{m=1}^{n}
 d_{
 X,
 J_{m},
 k_{m}
 }
 \vec{e}_{X}^{\ast}
 }
\end{align*}
from \eqref{nontri_I}. Hence, we obtain Theorem~\ref{thm_star_prd_G2_4} (our main theorem).

\section[Solution T\_n for n=0,1,2]{Solution $\boldsymbol{T_{n}}$ for $\boldsymbol{n=0,1,2}$}
\label{G24_sol_n012}

In this section, we give the linear operators
$T_{1}$, $T_{2}$
that are $n=1,2$ cases of~\eqref{star_coeff_G2_4_Opr}. By using the obtained linear operators, we calculate
\[
T^{1}_{\vec{\alpha}_{1},\vec{\beta}_{1}^{\ast}}
=
\langle
\vec{\alpha}_{1}
|
T_{1}
|
\vec{\beta}_{1}^{\ast}
\rangle
\qquad \text{and} \qquad
T^{2}_{\vec{\alpha}_{2},\vec{\beta}_{2}^{\ast}}
=
\langle
\vec{\alpha}_{2}
|
T_{2}
|
\vec{\beta}_{2}^{\ast}
\rangle
\]
from Theorem~\ref{thm_G24_opr}. We also calculate \smash{$T\raisebox{2pt}{${}^{n}_{\vec{\alpha}_{n},\vec{\beta}_{n}^{\ast}}$}$} for $n=1,2$ straightforwardly from \eqref{star_coeff_G2_4}. We then compare the obtained
\smash{$
\langle
\vec{\alpha}_{n}
|
T_{n}
|
\vec{\beta}_{n}^{\ast}
\rangle
$}
and
\smash{$T\raisebox{2pt}{${}^{n}_{\vec{\alpha}_{n},\vec{\beta}_{n}^{\ast}}$}$} from~\eqref{star_coeff_G2_4}
for $n=1,2$ and check that they are~equal to each other. Finally, we show that the obtained \smash{$T\raisebox{2pt}{${}^{2}_{\vec{\alpha}_{2},\vec{\beta}_{2}^{\ast}}$}$} actually satisfies the recurrence relations~\eqref{rec_rel_Grassmann_2_4_slash} in
Proposition~\ref{theorem_rec_rel_Grassmann_2_4}.

$T_{0}$ is immediately given as
\begin{align}
 T_{0}
 =
 |
 \vec{0}
 \rangle
\langle
 \vec{0}
|
 \label{opr_T_0_braket}
\end{align}
by Proposition~\ref{prop:covar_const_0and1}. We first check that $T_{1}$ recovers the coefficient for $n=1$. Let
\smash{$
\langle
\vec{e}_{M}
|
$}
and
\smash{$
|
\vec{e}_{L}^{\ast}
\rangle
$}
be some $n=1$ bra and ket, respectively. By Theorem~\ref{thm_G24_opr}, $T_{1}$ is given by
\begin{align}
 T_{1}
 ={}&
 \sum_{
 J,D
 \in
 \mathcal{I}
 }
 a_{D}^{\dagger}
 \frac{1}{
 \sqrt{
 N_{D}
 +
 1
 }
 }
|
 \vec{0}
\rangle
\langle
 \vec{0}
|
 \nonumber
 \\
 &
 \times
 \left\{
 a_{J}
 \frac{1}{\sqrt{N_{J}}}
 (
 \tau_{1}
 +
 N_{j\cancel{j}'}
 +
 1
 )
 g_{\overline{J}D}
 +
 a_{J}
 \frac{1}{
 \sqrt{
 N_{J}
 }
 }
 a_{\cancel{J}}
 \frac{1}{
 \sqrt{
 N_{\cancel{J}}
 }
 }
 a_{j\cancel{j}'}^{\dagger}
 \frac{1}{
 \sqrt{
 N_{j\cancel{j}'}+1
 }
 }
 (
 N_{j\cancel{j}'}+1
 )
 g_{\overline{\cancel{j}j'},D}
 \right\}
 \nonumber
 \\
 &
 \times
 \tau_{1}^{-1}
 \{
 1
 \cdot
 (
 \tau_{1}
 +
 1
 )
 +
 2
 (
 N_{I}
 +
 N_{\cancel{i}i'}
 )
 (
 N_{\cancel{I}}
 +
 N_{i\cancel{i}'}
 )
 \}^{-1}
 ,
 \nonumber
 \\
 ={}&
 \sum_{
 J,D
 \in
 \mathcal{I}
 }
 \big|
 \vec{e}_{D}^{\ast}
 \big\rangle
 \left\{
 \big\langle
 \vec{e}_{J}
 \big|
 (
 \tau_{1}
 +
 N_{j\cancel{j}'}
 +
 1
 )
 g_{\overline{J}D}
 +
 \big\langle
 \vec{e}_{J}
 +
 \vec{e}_{\cancel{J}}
 \big|
 a_{j\cancel{j}'}^{\dagger}
 \frac{1}{
 \sqrt{
 N_{j\cancel{j}'}+1
 }
 }
 (
 N_{j\cancel{j}'}+1
 )
 g_{\overline{\cancel{j}j'},D}
 \right\}
 \nonumber
 \\
 &
 \times
 \hbar
 \{
 (
 \tau_{1}
 +
 1
 )
 +
 2
 (
 N_{I}
 +
 N_{\cancel{i}i'}
 )
 (
 N_{\cancel{I}}
 +
 N_{i\cancel{i}'}
 )
 \}^{-1}
 ,
 \label{opr_another_exp_n1}
\end{align}
where we use \eqref{opr_T_0_braket} and
$
\tau_{1}^{-1}
=
\hbar
$.
Since
$
j\cancel{j}'
\neq
J$
(${=jj'}$),
$
j\cancel{j}'
\neq
\cancel{J}$
(${=
\cancel{j}\cancel{j}'}$), then
\smash{$
\langle
\vec{e}_{J}
+
\vec{e}_{\cancel{J}}
|
a_{j\cancel{j}'}^{\dagger}=0$},
and we obtain
\begin{align}
\eqref{opr_another_exp_n1}
=
\sum_{
J,D
\in
\mathcal{I}
}
\hbar
g_{\overline{J}D}
|
\vec{e}_{D}^{\ast}
\rangle
\langle
\vec{e}_{J}
|.\label{opr_T1_exp}
\end{align}
Here we use
$ ( N_{I} + N_{\cancel{i}i'} )
 ( N_{\cancel{I}} +
 N_{i\cancel{i}'}
 )
|
 \vec{e}_{L}^{\ast}
\rangle
 =
 0
$.
Hence,
$
\langle
\vec{e}_{M}
|
T_{1}
|
\vec{e}_{L}^{\ast}
\rangle
$
can be calculated as
\[
\langle
\vec{e}_{M}
|
T_{1}
|
\vec{e}_{L}^{\ast}
\rangle
=
\hbar
g_{M\overline{L}}
=
T^{1}_{\vec{e}_{M},\vec{e}_{L}^{\ast}}.
\]
This result coincides with \smash{$T^{n}_{\vec{\alpha}_{n},\vec{\beta}_{n}^{\ast}}$} for $n=1$ in~\eqref{init_cond_01} in Proposition~\ref{prop:covar_const_0and1}. Next, we consider the case $n=2$. For $n=2$, a linear operator $T_{2}$ is given by
\begin{align}
 T_{2}
 ={}&
 \sum_{
 D_{1},D_{2}
 \in
 \mathcal{I}
 }
 \sum_{
 J_{1},J_{2}
 \in
 \mathcal{I}
 }
 \hbar
 g_{\overline{J_{1}}D_{1}}
|
 \vec{e}_{D_{1}}^{\ast}
 +
 \vec{e}_{D_{2}}^{\ast}
\rangle
\langle
 \vec{e}_{J_{1}}
|
 \nonumber
 \\
 &
 \times
 \left\{
 a_{J_{2}}
 \frac{1}{\sqrt{N_{J_{2}}}}
 (
 \tau_{2}
 +
 N_{j_{2}\cancel{j_{2}}'}
 +
 1
 )
 g_{\overline{J_{2}}D}\right.\nonumber\\
 &\phantom{\times}{}\left.
 {}+
 a_{J_{2}}
 \frac{1}{
 \sqrt{
 N_{J_{2}}
 }
 }
 a_{\cancel{J_{2}}}
 \frac{1}{
 \sqrt{
 N_{\cancel{J_{2}}}
 }
 }
 a_{j_{2}\cancel{j_{2}}'}^{\dagger}
 \frac{1}{
 \sqrt{
 N_{j_{2}\cancel{j_{2}}'}+1
 }
 }
 (
 N_{j_{2}\cancel{j_{2}}'}+1
 )
 g_{\overline{\cancel{j_{2}}j_{2}'},D}
 \right\}
 \nonumber
 \\
 &
 \times
 \frac{1}{2}
 \tau_{2}^{-1}
 \{
 \tau_{1}
 +
 1
 +
 (
 N_{I}
 +
 N_{\cancel{i}i'}
 )
 (
 N_{\cancel{I}}
 +
 N_{i\cancel{i}'}
 )
 \}^{-1} , \label{opr_T_n2}
\end{align}
from \eqref{G24_coeff_opr_full} and \eqref{opr_T1_exp}.
Here we use
\begin{align*}
 a_{D_{1}}^{\dagger}
 \frac{1}{
 \sqrt{N_{D_{1}}+1}
 }
 a_{D_{2}}^{\dagger}
 \frac{1}{
 \sqrt{N_{D_{2}}+1}
 }
|
 \vec{0}
\rangle
 =
|
 \vec{e}_{D_{1}}^{\ast}
 +
 \vec{e}_{D_{2}}^{\ast}
\rangle .
\end{align*}
We now check that $T_{2}$ recovers \smash{$T\raisebox{2pt}{${}^{2}_{\vec{\alpha}_{2},\vec{\beta}_{2}^{\ast}}$}$}, that is, the solution for $n=2$. There are 10 possib\-le~$\vec{\beta}_{2}^{\ast}$ patterns. They can be classified into the following four patterns for a fixed
$
P
\in
\mathcal{I}
$:
\begin{align*}
 (\textrm{I}) \
 \vec{\beta}_{2}^{\ast}
 =
 2
 \vec{e}_{P}^{\ast},
 \qquad
 (\textrm{II}) \
 \vec{\beta}_{2}^{\ast}
 =
 \vec{e}_{P}^{\ast}
 +
 \vec{e}_{p\cancel{p}'}^{\ast},
 \qquad
 (\textrm{III}) \
 \vec{\beta}_{2}^{\ast}
 =
 \vec{e}_{P}^{\ast}
 +
 \vec{e}_{\cancel{p}p'}^{\ast},
 \qquad
 (\textrm{IV}) \
 \vec{\beta}_{2}^{\ast}
 =
 \vec{e}_{P}^{\ast}
 +
 \vec{e}_{\cancel{P}}^{\ast} .
\end{align*}
Therefore, we only need to check the above four patterns. For the case
\smash{$
\vec{\beta}_{2}^{\ast}
=
2
\vec{e}_{P}^{\ast}
$},
\smash{$
T_{2}
|
2
\vec{e}_{P}^{\ast}
\rangle
$}
is non-zero only for
$
J_{2}
=
P
$
in \eqref{opr_T_n2},
and
$
 (
\delta_{IP}
+
\delta_{\cancel{i}i',P}
 )
 (
\delta_{\cancel{I}P}
+
\delta_{i\cancel{i}',P}
 )
=
0
$
for a fixed $P\in\mathcal{I}$.
\smash{$
T_{2}
|
2
\vec{e}_{P}^{\ast}
\rangle
$}
can be rewritten as
\begin{align}
T_{2}
|
2
\vec{e}_{P}^{\ast}
\rangle
={}&
\frac{\hbar}{2}
\sum_{
J_{1},D_{1},D_{2}
\in
\mathcal{I}
}
g_{\overline{J_{1}}D_{1}}
|
 \vec{e}_{D_{1}}^{\ast}
 +
 \vec{e}_{D_{2}}^{\ast}
\rangle
\big\langle
\vec{e}_{J_{1}}
|
\vec{e}_{P}^{\ast}
\rangle
\cdot
g_{\overline{P}D_{2}}
\tau_{2}^{-1}
 (
 \tau_{2}
 +
 1
 )
 (
 \tau_{2}
 +
 1
 )^{-1}
\nonumber
\\
={}&
\frac{\hbar}{2}
 \left(
-1
+
\frac{1}{\hbar}
 \right)^{-1}
\sum_{
D_{1},D_{2}
\in
\mathcal{I}
}
g_{\overline{P}D_{1}}
g_{\overline{P}D_{2}}
|
 \vec{e}_{D_{1}}^{\ast}
 +
 \vec{e}_{D_{2}}^{\ast}
\rangle.
\label{opr_T_n2_act_2e_midmid}
\end{align}
Here we use
\smash{$
\langle
\vec{e}_{J_{1}}
|
\vec{e}_{P}^{\ast}
\rangle
=
\delta_{J_{1}P}
$}
and
$
\tau_{2}
=
-1
+
\frac{1}{\hbar}
$.
By using \eqref{opr_T_n2_act_2e_midmid}, we obtain
\begin{align}
T^{2}_{\vec{\alpha}_{2},2\vec{e}_{P}^{\ast}}
=
\langle
\vec{\alpha}_{2}
|
T_{2}
|
2
\vec{e}_{P}^{\ast}
\rangle
=
\frac{\hbar}{2}
 \left(
-1
+
\frac{1}{\hbar}
 \right)^{-1}
\sum_{
D_{1},D_{2}
\in
\mathcal{I}
}
g_{\overline{P}D_{1}}
g_{\overline{P}D_{2}}
\delta_{
\vec{\alpha}_{2},
\vec{e}_{D_{1}}
+
\vec{e}_{D_{2}}
}
.
\label{T2_2e_opr}
\end{align}
We also calculate directly \smash{$T^{2}_{\vec{\alpha}_{2},2\vec{e}_{P}^{\ast}}$} by using \eqref{star_coeff_G2_4} in Theorem~\ref{thm_star_coeff_G2_4} for confirmation. Since
$
\vec{\beta}_{2}^{\ast}
=
2\vec{e}_{P}^{\ast}
$
and $n=2$, the step function in \eqref{star_coeff_G2_4} is
\begin{align*}
\prod_{
S
\in
\mathcal{I}
}
\prod_{r=1}^{2}
\theta
 \left(
 \beta_{S}^{2}
 -
 \sum_{m=r}^{2}
 d_{S,J_{m},k_{m}}
 \right)
=&
\prod_{
S
\in
\mathcal{I}
}
\prod_{r=1}^{2}
\theta
 \left(
 2
 \delta_{SP}
 -
 \sum_{m=r}^{2}
 d_{S,J_{m},k_{m}}
 \right)
.
\end{align*}
Recalling that
\begin{align*}
d_{
 S,
 J_{m},
 k_{m}
}
={}&
\delta_{
 S,
 J_{m}
}
+
\delta_{\cancel{j_{m}}k_{m}}
 (
 \delta_{
 S,
 \cancel{J_{m}}
 }
 -
 \delta_{
 S,
 j_{m}\cancel{j_{m}}'
 }
 )
,
\end{align*}
the step function
$
\prod_{
S
\in
\mathcal{I}
}
\prod_{r=1}^{2}
\theta
\bigl(
 2
 \delta_{SP}
 -
 \sum_{m=r}^{2}
 d_{S,J_{m},k_{m}}
\bigr)
$
is 1 if and only if $J_{1}=J_{2}=P$, $k_{1}=j_{1}$ and $k_{2}=j_{2}$. If $J_{1}=J_{2}=P=pp'$ and $k_{1}=j_{1}=k_{2}=j_{2}=p$, \smash{$T^{2}_{\vec{\alpha}_{2},2\vec{e}_{P}^{\ast}}$} can be expressed as
\begin{align}
 T^{2}_{\vec{\alpha}_{2},2\vec{e}_{P}^{\ast}}
 ={}&
 \sum_{
 D_{1},
 D_{2}
 \in
 \mathcal{I}
 }
 \frac{
 g_{\overline{P}D_{1}}
 g_{\overline{P}D_{2}}
 }{
 \tau_{1}
 \tau_{2}
 }
 \delta_{
 \vec{\alpha}_{2},
 \vec{e}_{D_{1}}
 +
 \vec{e}_{D_{2}}
 }
 \delta_{
 2
 \vec{e}_{P}^{\ast},
 \sum_{
 X
 \in
 \mathcal{I}
 }
 2
 d_{
 X,
 P,
 p
 }
 \vec{e}_{X}^{\ast}
 }
 \nonumber
 \\
 &
 \times
 \frac{
 \tau_{1}
 +
 2
 \delta_{p\cancel{p}',P}
 +
 1
 -
 \Lambda_{
 1,
 j_{1}\cancel{j_{1}}',
 \{
 J_{i}
 \}_{2},
 \{
 j_{i}
 \}_{2}
 }
 \big|_{J_{1}=J_{2}=P}
 }{
 \tau_{1}+1
 +
 2
 (
 2
 \delta_{IP}
 +
 2
 \delta_{\cancel{i}i',P}
 -
 \Delta_{
 I,
 \cancel{i}i',
 1,
 \{
 J_{i}
 \}_{2}
 }
 \big|_{J_{1}=J_{2}=P}
 )
 }
 \nonumber
 \\
 &\times\frac{1}{ (
 2
 \delta_{\cancel{I}P}
 +
 2
 \delta_{i\cancel{i}',P}
 -
 \Delta_{
 \cancel{I},
 i\cancel{i}',
 1,
 \{
 J_{i}
 \}_{2}
 }
 \big|_{J_{1}=J_{2}=P}
 )}\nonumber\\
 &
 \times
 \frac{
 \tau_{2}
 +
 2
 \delta_{p\cancel{p}',P}
 +
 1
 -
 \Lambda_{
 2,
 j_{2}\cancel{j_{2}}',
 \{
 J_{i}
 \}_{2},
 \{
 j_{i}
 \}_{2}
 }
 \big|_{J_{1}=J_{2}=P}
 }{
 2
 (
 \tau_{2}+1
 )
 +
 2
 (
 2
 \delta_{IP}
 +
 2
 \delta_{\cancel{i}i',P}
 -
 \Delta_{
 I,
 \cancel{i}i',
 2,
 \{
 J_{i}
 \}_{2}
 }
 \big|_{J_{1}=J_{2}=P}
 )
 }
 \nonumber\\
 &\times\frac{1}{ (
 2
 \delta_{\cancel{I}P}
 +
 2
 \delta_{i\cancel{i}',P}
 -
 \Delta_{
 \cancel{I},
 i\cancel{i}',
 2,
 \{
 J_{i}
 \}_{2}
 }
 \big|_{J_{1}=J_{2}=P}
 )}\label{star_coeff_G24_n2_2P}
\end{align}
for a fixed $P=pp'$, where
$
\beta_{M}^{2}
=
2\delta_{MP}
$
for
$
M
\in
\mathcal{I}
$.
Note that
\begin{gather}
d_{
 X,
 P,
 p
}
=
\delta_{
 X,
 P
},
\label{d_calc_2}
\\
\Lambda_{
 1,
 j_{1}\cancel{j_{1}}',
 \{
 J_{i}
 \}_{2},
 \{
 j_{i}
 \}_{2}
}
|_{J_{1}=J_{2}=P}
=
\delta_{p\cancel{p}',P}
=
0
,
\qquad
\Lambda_{
 2,
 j_{2}\cancel{j_{2}}',
 \{
 J_{i}
 \}_{2},
 \{
 j_{i}
 \}_{2}
}
|_{J_{1}=J_{2}=P}
=
0,
\label{lambda_calc_1and2}
\\
\Delta_{
I,
\cancel{i}i',
1,
\{
 J_{i}
\}_{2}
}
|_{J_{1}=J_{2}=P}
=
\delta_{IP}
+
\delta_{\cancel{i}i',P},
\qquad
\Delta_{
\cancel{I},
i\cancel{i}',
1,
\{
 J_{i}
\}_{2}
}
|_{J_{1}=J_{2}=P}
=
\delta_{\cancel{I}P}
+
\delta_{i\cancel{i}',P},
\label{Delta_calc_1}
\\
\Delta_{
I,
\cancel{i}i',
2,
\{
 J_{i}
\}_{2}
}
|_{J_{1}=J_{2}=P}
=
0,
\qquad
\Delta_{
\cancel{I},
i\cancel{i}',
2,
\{
 J_{i}
\}_{2}
}
|_{J_{1}=J_{2}=P}
=
0
\label{Delta_calc_2}
\end{gather}
from \eqref{notation_Lambda}--\eqref{notation_delta_Isiis}. Since
\smash{$
\sum_{
X
\in
\mathcal{I}
}
\delta_{
X,
P
}
\vec{e}_{X}^{\ast}
=
\vec{e}_{P}^{\ast}
$},
then Kronecker's delta is calculated as
\begin{align}
 \delta_{
 2
 \vec{e}_{P}^{\ast},
 2
 \sum_{
 X
 \in
 \mathcal{I}
 }
 d_{
 X,
 P,
 p
 }
 \vec{e}_{X}^{\ast}
 }
 =
 \delta_{
 2
 \vec{e}_{P}^{\ast},
 2
 \vec{e}_{P}^{\ast}
 }
 =
 1
 \label{K_delta_calc_2}
\end{align}
by using \eqref{d_calc_2}. The denominators and numerators in \eqref{star_coeff_G24_n2_2P} can be calculated by using \eqref{lambda_calc_1and2}--\eqref{Delta_calc_2}. The numerators in \eqref{star_coeff_G24_n2_2P} are expressed as
\begin{gather*}
\tau_{1}
+
2
\delta_{p\cancel{p}',P}
+
1
-
\Lambda_{
1,
j_{1}\cancel{j_{1}}',
\{
J_{i}
\}_{2},
\{
j_{i}
\}_{2}
}
|_{J_{1}=J_{2}=P}
=
\tau_{1}
+
1,
\\
\tau_{2}
+
2
\delta_{p\cancel{p}',P}
+
1
-
\Lambda_{
2,
j_{2}\cancel{j_{2}}',
\{
J_{i}
\}_{2},
\{
j_{i}
\}_{2}
}
|_{J_{1}=J_{2}=P}
=
\tau_{2}
+
1
\end{gather*}
from \eqref{lambda_calc_1and2}, respectively. Here, we use
$
\delta_{
P,
p\cancel{p}'
}
=
0
$
since
$
P
\neq
p\cancel{p}'
$
for any
$
P
=
pp'
$.
By using \eqref{Delta_calc_1} and \eqref{Delta_calc_2}, for $l=1,2$, the denominators in \eqref{star_coeff_G24_n2_2P} are also calculated as
\begin{gather*}
\tau_{1}
+
1
+
2
\{
 2
 (
 \delta_{IP}
 +
 \delta_{\cancel{i}i',P}
 )
 -
 \Delta_{
 I,
 \cancel{i}i',
 1,
 \{
 J_{i}
\}_{2}
}
\}
\{
 2
 (
 \delta_{\cancel{I}P}
 +
 \delta_{i\cancel{i}',P}
 )
 -
 \Delta_{
 \cancel{I},
 i\cancel{i}',
 1,
 \{
 J_{i}
 \}_{2}
 }
\}
\\
\qquad=
\tau_{1}
+
1
+
2
 (
 \delta_{IP}
 +
 \delta_{\cancel{i}i',P}
 )
 (
 \delta_{\cancel{I}P}
 +
 \delta_{i\cancel{i}',P}
 )
=
\tau_{1}
+
1,
\end{gather*}
and
\begin{gather}
2
 (
\tau_{2}
+
1
 )
+
2
\{
 2
 (
 \delta_{IP}
 +
 \delta_{\cancel{i}i',P}
 )
 -
 \Delta_{
 I,
 \cancel{i}i',
 2,
 \{
 J_{i}
\}_{2}
}
\}
\{
 2
 (
 \delta_{\cancel{I}P}
 +
 \delta_{i\cancel{i}',P}
 )
 -
 \Delta_{
 \cancel{I},
 i\cancel{i}',
 2,
 \{
 J_{i}
 \}_{2}
 }
\}
\nonumber
\\
\qquad=
2
 (
\tau_{2}
+
1
 )
+
8
 (
 \delta_{IP}
 +
 \delta_{\cancel{i}i',P}
 )
 (
 \delta_{\cancel{I}P}
 +
 \delta_{i\cancel{i}',P}
 )
=
2
 (
\tau_{2}
+
1
 ).
\label{denominator_2}
\end{gather}
Here, we use
$
 (
 \delta_{IP}
 +
 \delta_{\cancel{i}i',P}
 )
 (
 \delta_{\cancel{I}P}
 +
 \delta_{i\cancel{i}',P}
 )
=
0
$
for any indices
$
P,
I
\in
\mathcal{I}
$.
Substituting \eqref{K_delta_calc_2}--\eqref{denominator_2} into \eqref{star_coeff_G24_n2_2P}, we eventually obtain
\begin{align}
T^{2}_{\vec{\alpha}_{2},2\vec{e}_{P}^{\ast}}
={}&
\sum_{
D_{1},D_{2}
\in
\mathcal{I}
}
\frac{1}{2}
g_{\overline{P}D_{1}}
g_{\overline{P}D_{2}}
\tau_{1}^{-1}
\tau_{2}^{-1}
 (
\tau_{1}
+
1
 )
 (
\tau_{2}
+
1
 )
 (
\tau_{1}
+
1
 )^{-1}
 (
\tau_{2}
+
1
 )^{-1}
\delta_{
\vec{\alpha}_{2},
\vec{e}_{D_{1}}
+
\vec{e}_{D_{2}}
}
\nonumber
\\
={}&
\frac{\hbar}{2}
 \left(
-1
+
\frac{1}{\hbar}
 \right)^{-1}
\sum_{
D_{1},D_{2}
\in
\mathcal{I}
}
g_{\overline{P}D_{1}}
g_{\overline{P}D_{2}}
\delta_{
\vec{\alpha}_{2},
\vec{e}_{D_{1}}
+
\vec{e}_{D_{2}}
}
.
\label{T2_2e_sf}
\end{align}
\eqref{T2_2e_sf} coincides with \eqref{T2_2e_opr}.

Similarly, we can calculate for the cases
\smash{$
\vec{\beta}_{2}^{\ast}
=
\vec{e}_{P}^{\ast}
+
\vec{e}_{p\cancel{p}'}^{\ast}
$},
\smash{$
\vec{\beta}_{2}^{\ast}
=
\vec{e}_{P}^{\ast}
+
\vec{e}_{\cancel{p}p'}^{\ast}
$}
and
\smash{$
\vec{\beta}_{2}^{\ast}
=
\vec{e}_{P}^{\ast}
+
\vec{e}_{\cancel{P}}^{\ast}
$}
as follows:
\begin{gather}
T^{2}_{\vec{\alpha}_{2},
\vec{e}_{P}^{\ast}
+
\vec{e}_{p\cancel{p}'}^{\ast}
}
=
\langle
\vec{\alpha}_{2}
|
T_{2}
|
\vec{e}_{P}^{\ast}
+
\vec{e}_{p\cancel{p}'}^{\ast}
\rangle
=
\hbar
 \left(
-1
+
\frac{1}{\hbar}
 \right)^{-1}
\sum_{
D_{1},D_{2}
\in
\mathcal{I}
}
g_{\overline{P}D_{1}}
g_{\overline{p\cancel{p}'},D_{2}}
\delta_{
\vec{\alpha}_{2},
\vec{e}_{D_{1}}
+
\vec{e}_{D_{2}}
}.
\label{T2_case2_opr}
\\
T^{2}_{\vec{\alpha}_{2},
\vec{e}_{P}^{\ast}
+
\vec{e}_{\cancel{p}p'}^{\ast}
}
=
\langle
\vec{\alpha}_{2}
|
T_{2}
|
\vec{e}_{P}^{\ast}
+
\vec{e}_{\cancel{p}p'}^{\ast}
\rangle
=
\hbar
 \left(
-1
+
\frac{1}{\hbar}
 \right)^{-1}
\sum_{
D_{1},D_{2}
\in
\mathcal{I}
}
g_{\overline{P}D_{1}}
g_{\overline{\cancel{p}p'},D_{2}}
\delta_{
\vec{\alpha}_{2},
\vec{e}_{D_{1}}
+
\vec{e}_{D_{2}}
},
\label{T2_case3_opr}
\\
T^{2}_{\vec{\alpha}_{2},
\vec{e}_{P}^{\ast}
+
\vec{e}_{\cancel{P}}^{\ast}
}
=
\langle
\vec{\alpha}_{2}
|
T_{2}
|
\vec{e}_{P}^{\ast}
+
\vec{e}_{\cancel{P}}^{\ast}
\rangle=
 \left(
-1
+
\frac{1}{\hbar}
 \right)^{-1}
 \left(
1
+
\frac{1}{\hbar}
 \right)^{-1}
\nonumber
\\
\phantom{T^{2}_{\vec{\alpha}_{2},
\vec{e}_{P}^{\ast}
+
\vec{e}_{\cancel{P}}^{\ast}
}
=}{}
\times
\sum_{
D_{1},D_{2}
\in
\mathcal{I}
}
 (
g_{\overline{P}D_{1}}
g_{\overline{\cancel{P}}D_{2}}
+
\hbar
g_{\overline{p\cancel{p}'},D_{1}}
g_{\overline{\cancel{p}p'},D_{2}}
 )
\delta_{
\vec{\alpha}_{2},
\vec{e}_{D_{1}}
+
\vec{e}_{D_{2}}
}
.
\label{T2_case4_opr}
\end{gather}
By long but straightforward calculations, we can also check that \eqref{T2_case2_opr}, \eqref{T2_case3_opr} and \eqref{T2_case4_opr} coincide with
\smash{$
T^{2}_{\vec{\alpha}_{2},
\vec{e}_{P}^{\ast}
+
\vec{e}_{p\cancel{p}'}^{\ast}
}
$},
$
T^{2}_{\vec{\alpha}_{2},
\vec{e}_{P}^{\ast}
+
\vec{e}_{\cancel{p}p'}^{\ast}
}
$
and
$
T^{2}_{\vec{\alpha}_{2},
\vec{e}_{P}^{\ast}
+
\vec{e}_{\cancel{P}}^{\ast}
}
$,
calculated from \eqref{star_coeff_G2_4}, respectively.

Lastly, we check that
\smash{$T^{2}_{\vec{\alpha}_{2},\vec{\beta}_{2}^{\ast}}$} satisfies the recurrence relations \eqref{rec_rel_Grassmann_2_4_slash} for $n=2$ in
Proposition~\ref{theorem_rec_rel_Grassmann_2_4}. In this discussion, we give
detailed calculations
only for the case
\smash{$
\vec{\beta}_{2}^{\ast}
=
2
\vec{e}_{P}^{\ast}
$}
to avoid the repetition of boring discussions. In this case, the right-hand side of \eqref{rec_rel_Grassmann_2_4_slash} is
\begin{align}
2
\hbar
\delta_{PI}
 (
\tau_{2}
+
2
\delta_{P\cancel{I}}
 )
T^{2}_{\vec{\alpha}_{2},
2
\vec{e}_{P}^{\ast}
}
-
\hbar
 (
2
\delta_{P,i\cancel{i}'}
+
1
 )
 (
2
\delta_{P,\cancel{i}i'}
+
1
 )
T^{2}_{\vec{\alpha}_{2},
2
\vec{e}_{P}^{\ast}
-
\vec{e}_{I}^{\ast}
+
\vec{e}_{i\cancel{i}'}^{\ast}
+
\vec{e}_{\cancel{i}i'}^{\ast}
-
\vec{e}_{\cancel{I}}^{\ast}
}
.
\label{rhs_recrel_2e_mid}
\end{align}
Here,
we have
$
\delta_{PI}
\delta_{P\cancel{I}}
=
0
$
since
$
I
\neq
\cancel{I}
$,
and
\[
T^{2}_{\vec{\alpha}_{2},
2
\vec{e}_{P}^{\ast}
-
\vec{e}_{I}^{\ast}
+
\vec{e}_{i\cancel{i}'}^{\ast}
+
\vec{e}_{\cancel{i}i'}^{\ast}
-
\vec{e}_{\cancel{I}}^{\ast}
}
=
0
\]
since
\smash{$
2
\vec{e}_{P}^{\ast}
-
\vec{e}_{I}^{\ast}
+
\vec{e}_{i\cancel{i}'}^{\ast}
+
\vec{e}_{\cancel{i}i'}^{\ast}
-
\vec{e}_{\cancel{I}}^{\ast}
$}
has negative components
in the \smash{$\vec{e}_{I}^{\ast}$}-direction or \smash{$\vec{e}_{\cancel{I}}^{\ast}$}-direction.
Therefore, after substituting \eqref{T2_2e_sf}, the right-hand side is rewritten as
\begin{align}
\eqref{rhs_recrel_2e_mid}
={}&
2
\hbar
\delta_{PI}
 \left(
-1
+
\frac{1}{\hbar}
 \right)
\frac{\hbar}{2}
 \left(
-1
+
\frac{1}{\hbar}
 \right)^{-1}
\sum_{
D_{1},D_{2}
\in
\mathcal{I}
}
g_{\overline{P}D_{1}}
g_{\overline{P}D_{2}}
\delta_{
\vec{\alpha}_{2},
\vec{e}_{D_{1}}
+
\vec{e}_{D_{2}}
}
\nonumber
\\
={}&
\hbar^{2}
\sum_{
D_{1},D_{2}
\in
\mathcal{I}
}
g_{\overline{P}D_{1}}
g_{\overline{P}D_{2}}
\delta_{
\vec{\alpha}_{2},
\vec{e}_{D_{1}}
+
\vec{e}_{D_{2}}
}
\delta_{PI}
.
\label{rhs_recrel_2e}
\end{align}
We next calculate the left-hand side of \eqref{rec_rel_Grassmann_2_4_slash} for $n=2$. Note that the following equation holds:
\begin{align*}
T^{1}_{
\vec{\alpha}_{2}
-
\vec{e}_{D_{1}}
,
\vec{\beta}_{2}^{\ast}
-
\vec{e}_{I}^{\ast}
}
=
\hbar
\sum_{
J,D_{2}
\in
\mathcal{I}
}
g_{\overline{J}D_{2}}
\delta_{
\vec{\alpha}_{2},
\vec{e}_{D_{1}}
+
\vec{e}_{D_{2}}
}
\delta_{
\vec{\beta}_{2}^{\ast},
\vec{e}_{I}^{\ast}
+
\vec{e}_{J}^{\ast}
}
\end{align*}
from \eqref{init_cond_01} for $n=1$ in Proposition~\ref{prop:covar_const_0and1}. Then, the left-hand side of \eqref{rec_rel_Grassmann_2_4_slash} can be calculated as
\begin{align}
\hbar
\sum_{
D_{1}
\in
\mathcal{I}
}
g_{\overline{I}
D_{1}
}
T^{1}_{
\vec{\alpha}_{2}
-
\vec{e}_{D_{1}}
,
2
\vec{e}_{P}^{\ast}
-
\vec{e}_{I}^{\ast}
}
={}&
\hbar
\sum_{
D_{1}
\in
\mathcal{I}
}
g_{\overline{I}D_{1}}
\cdot
\hbar
\sum_{
J,D_{2}
\in
\mathcal{I}
}
g_{\overline{J}D_{2}}
\delta_{
\vec{\alpha}_{2},
\vec{e}_{D_{1}}
+
\vec{e}_{D_{2}}
}
\delta_{
2
\vec{e}_{P}^{\ast},
\vec{e}_{I}^{\ast}
+
\vec{e}_{J}^{\ast}
}
\nonumber
\\
={}&
\hbar^{2}
\sum_{
D_{1},D_{2}
\in
\mathcal{I}
}
g_{\overline{I}D_{1}}
g_{\overline{I}D_{2}}
\delta_{
\vec{\alpha}_{2},
\vec{e}_{D_{1}}
+
\vec{e}_{D_{2}}
}
\delta_{
2
\vec{e}_{P}^{\ast},
2
\vec{e}_{I}^{\ast}
}
.
\label{lhs_recrel_2e_mid}
\end{align}
Since
\smash{$
\delta_{
2
\vec{e}_{P}^{\ast},
2
\vec{e}_{I}^{\ast}
}
=
\delta_{PI}
$}
and then
$
g_{\overline{I}D_{1}}
g_{\overline{I}D_{2}}
\delta_{PI}
=
g_{\overline{P}D_{1}}
g_{\overline{P}D_{2}}
\delta_{PI}
$,
we obtain
\begin{align}
\eqref{lhs_recrel_2e_mid}
=&
\hbar^{2}
\sum_{
D_{1},D_{2}
\in
\mathcal{I}
}
g_{\overline{P}D_{1}}
g_{\overline{P}D_{2}}
\delta_{
\vec{\alpha}_{2},
\vec{e}_{D_{1}}
+
\vec{e}_{D_{2}}
}
\delta_{PI}.
\label{lhs_recrel_2e}
\end{align}
Therefore, the right-hand side \eqref{rhs_recrel_2e} coincides with the left-hand side \eqref{lhs_recrel_2e}. Hence, it is checked that the solution
\smash{$
T^{2}_{\vec{\alpha}_{2},
2
\vec{e}_{P}^{\ast}
}
$}
satisfies the recurrence relations \eqref{rec_rel_Grassmann_2_4_slash} in
Proposition~\ref{theorem_rec_rel_Grassmann_2_4} for~${n=2}$. In the same way, we can also verify by straightforward calculations that the other solutions~\smash{$
T^{2}_{\vec{\alpha}_{2},
\vec{e}_{P}^{\ast}
+
\vec{e}_{\cancel{p}p'}^{\ast}
}
$},
\smash{$
T^{2}_{\vec{\alpha}_{2},
\vec{e}_{P}^{\ast}
+
\vec{e}_{\cancel{p}p'}^{\ast}
}
$}
and
\smash{$
T^{2}_{\vec{\alpha}_{2},
\vec{e}_{P}^{\ast}
+
\vec{e}_{\cancel{P}}^{\ast}
}
$}
satisfy the recurrence relations \eqref{rec_rel_Grassmann_2_4_slash} in
Proposition~\ref{theorem_rec_rel_Grassmann_2_4}.

\section{Summary}\label{G24_summary}

We gave the recurrence relations \eqref{rec_rel_Grassmann_general} in Theorem~\ref{theorem_rec_rel_Grassmann_general} whose solutions are coefficients of a star product with separation of variables on $G_{p,p+q}(\mathbb{C})$, based on the construction method proposed by \cite{HS1,HS2}. In particular, we focused on the case $p=q=2$ to give the explicit star product with separation of variables on $G_{2,4}(\mathbb{C})$. Considering the case $p=q=2$, we can obtain the recurrence relations~\eqref{rec_rel_Grassmann_2_4_slash} for~${G_{2,4}(\mathbb{C})}$ from Theorem~\ref{theorem_rec_rel_Grassmann_general}. To give the solution of~\eqref{rec_rel_Grassmann_2_4_slash} for~$G_{2,4}(\mathbb{C})$, it was necessary to solve a system of linear equations. The process of solving this system of linear equations prevented obtaining the explicit general term expression. To resolve this problem, we derived \eqref{G24_n_relation} in Proposition~\ref{Prop_G24_relation} that is equivalent to \eqref{rec_rel_Grassmann_2_4_slash}, in which the general term~\smash{$T\raisebox{2pt}{${}^{n}_{\vec{\alpha}_{n},\vec{\beta}_{n}^{\ast}}$}$} of order $n$ are expressed by coefficients of order~${(n-1)}$.
By using \eqref{G24_n_relation}, we derived the formula~\eqref{G24_n_relation_sum} in Proposition~\ref{Cor_G24_relation_sum}, which uniquely determines \smash{$T\raisebox{2pt}{${}^{n}_{\vec{\alpha}_{n},\vec{\beta}_{n}^{\ast}}$}$} from coefficients of order~${(n-1)}$ without solving a system of linear equations.
To solve
\eqref{G24_n_relation_sum}, we introduced a linear operator $T_{n}$ on a Fock space $V$ that recovers \smash{$T\raisebox{2pt}{${}^{n}_{\vec{\alpha}_{n},\vec{\beta}_{n}^{\ast}}$}$}
as a matrix representation. Furthermore, by substituting~\smash{$T\raisebox{2pt}{${}^{n}_{\vec{\alpha}_{n},\vec{\beta}_{n}^{\ast}}$}$} into the equation \eqref{star_expand} given by~\cite{HS1,HS2}, we succeeded in~constructing the explicit star product with separation of variables on $G_{2,4}(\mathbb{C})$. For confirmation, we checked that for $n=1,2$,
\smash{$
T\raisebox{2pt}{${}^{1}_{\vec{\alpha}_{1},\vec{\beta}_{1}^{\ast}}$}$},
\smash{$
T\raisebox{2pt}{${}^{2}_{\vec{\alpha}_{2},\vec{\beta}_{2}^{\ast}}$}
$}
recovered from
$
T_{1}$,
$
T_{2}
$
coincides with
\smash{$
T\raisebox{2pt}{${}^{1}_{\vec{\alpha}_{1},\vec{\beta}_{1}^{\ast}}$}$},
\smash{$
T\raisebox{2pt}{${}^{2}_{\vec{\alpha}_{2},\vec{\beta}_{2}^{\ast}}$}
$}
obtained straightforwardly from \eqref{star_coeff_G2_4} in Theorem~\ref{thm_star_coeff_G2_4}, respectively. We also verified that the obtained~\smash{$T\raisebox{2pt}{${}^{2}_{\vec{\alpha}_{2},\vec{\beta}_{2}^{\ast}}$}$} actually satisfies the recurrence relations \eqref{rec_rel_Grassmann_2_4_slash}.

Here we discuss an outlook for our work toward Penrose's twistor theory.
Penrose's twistor theory can generally be formulated via twister correspondences \cite{Pen1,Pen2,Pen3,Ward1,WaWe}.
The twistor (Klein) correspondence \smash{$G_{2,4}(\mathbb{C})\xleftarrow{\pi_{2}}F_{1,2,4}(\mathbb{C})\xrightarrow{\pi_{1}}\mathbb{C}P^{3}$} is defined by two fibrations
$
\pi_{1}
\colon
F_{1,2,4}(\mathbb{C})\allowbreak
\to
\mathbb{C}P^{3}
$
and
$
\pi_{2}
\colon
F_{1,2,4}(\mathbb{C})
\to
G_{2,4}(\mathbb{C})
$
such that $\pi_{1}(V_{1},V_{2}):=V_{1}$, $\pi_{2}(V_{1},V_{2}):=V_{2}$. Here the complex flag manifold $F_{1,2,4}(\mathbb{C})$ is defined by
\begin{align*}
F_{1,2,4}(\mathbb{C})
:=
\bigl\{
(V_{1},V_{2})
\mid
V_{1}
\subset
V_{2}
\subset
\mathbb{C}^{4},\,
\dim\nolimits_{\mathbb{C}}
V_{k}
=
k,\,
k=1,2
\bigr\}.
\end{align*}
As an application of the star product on $G_{2,4}(\mathbb{C})$ obtained in our work, we expect to realize a~noncommutative deformation of the twistor correspondence
\begin{align*}
(C^{\infty}(G_{2,4}(\mathbb{C}))[\![\hbar]\!],\ast_{G_{2,4}(\mathbb{C})})
\xleftarrow{\pi_{2,\ast}}
(C^{\infty}(F_{1,2,4}(\mathbb{C}))[\![\hbar]\!],\ast_{F_{1,2,4}(\mathbb{C})})
\xrightarrow{\pi_{1,\ast}}
\bigl(C^{\infty}\bigl(\mathbb{C}P^{3}\bigr)[\![\hbar]\!],\ast_{\mathbb{C}P^{3}}\bigr)
\end{align*}
defined by two fibrations
\begin{gather*}
\pi_{1,\ast}\colon\
(C^{\infty}(F_{1,2,4}(\mathbb{C}))[\![\hbar]\!],\ast_{F_{1,2,4}(\mathbb{C})})
\to
\bigl(C^{\infty}\bigl(\mathbb{C}P^{3}\bigr)[\![\hbar]\!],\ast_{\mathbb{C}P^{3}}\bigr),
\\
\pi_{2,\ast}\colon\
(C^{\infty}(F_{1,2,4}(\mathbb{C}))[\![\hbar]\!],\ast_{F_{1,2,4}(\mathbb{C})})
\to
(C^{\infty}(G_{2,4}(\mathbb{C}))[\![\hbar]\!],\ast_{G_{2,4}(\mathbb{C})}).
\end{gather*}
Toward this goal, it must be useful to construct an explicit star product with separation of variables on $F_{1,2,4}(\mathbb{C})$.
We expect that the key to define $\pi_{1,\ast}$ and $\pi_{2,\ast}$ is the relations between an explicit star product on $F_{1,2,4}(\mathbb{C})$ and one on $\mathbb{C}P^{3}$ or $G_{2,4}(\mathbb{C})$.
By constructing a noncommutative deformation of twistor correspondence, the development of noncommutative twistor theory and the finding of new relations in noncommutative integrable systems are expected.

\appendix

\section[Properties of G\_\{p,p+q\}(C)]{Properties of $\boldsymbol{G_{p,p+q}(\mathbb{C})}$}
\label{appendix_properties_Gpq}

We denote some useful properties with respect to $G_{p,p+q}(\mathbb{C})$ in this appendix. The set of capital letter indices $\{I=ii'\mid 1\leq i\leq q,\, 1'\leq i'\leq p' \}$ introduced in Section~\ref{G24_introduction} and
$\{1,\dots,qp\}$
are ordered isomorphic. We now define the isomorphism
$ \psi\colon \{ I=ii' \mid {1 \leq i \leq q},\allowbreak 1' \leq i' \leq p'\} \to\{ 1,\dots,qp\} \subset \mathbb{N}$
by
$
 \psi
( I)
 =
 \psi
 (
 ii'
 )
 :=
 p
(
 i-1
)
 +
 {i'}_{\mathbb{N}}$,
where ${i'}_{\mathbb{N}}$ denotes the natural number corresponding to $i'$ in a natural way. For example, if~${i'=1'}$, then ${i'}_{\mathbb{N}}={1'}_{\mathbb{N}}=1$.
Using~this~$\psi$, for any
$
 I
 =
 ii'$,
$
 J=jj'
$,
we define the binary relation $\leq_{C}$ by
\smash{$
 I
 \leq_{C}
 J
 \overset{\mathrm{def}}{\Longleftrightarrow}
 \psi
 (
 I
 )
 \leq
 \psi
 (
 J
 )
$}
on~${
 \{
 I=ii'
 \mid
 1
 \leq
 i
 \leq
 q,\,
 1'
 \leq
 i'
 \leq
 p'
 \}}
$.
Then $\leq_{C}$ is the total order
on $\{I=ii'\mid 1\leq i\leq q,\allowbreak {1'\leq i'\leq p'}\}$. Hence, $\psi$ is an ordered isomorphism, i.e.,
\begin{align*}
 (
 \{
 I=ii'
 |
 1
 \leq
 i
 \leq
 q,
 1'
 \leq
 i'
 \leq
 p'
 \},
 \leq_{C}
 )
 \cong
 (
 \{
 1,\dots,qp
 \},
 \leq
 ).
\end{align*}
Using the above ordered isomorphism $\psi$, we can identify
$
 (
 \{
 I=ii'
\mid
 1
 \leq
 i
 \leq
 q,\,
 1'
 \leq
 i'
 \leq
 p'
 \},\allowbreak
 \leq_{C}
 )
$
and
$
 (
 \{
 1,\dots,qp
 \},
 \leq
 )
$
as ordered sets, and we can regard capital letter indices as ordinary indices~${I=1,\dots,qp}$.

We next discuss the properties of
the K\"ahler potential and the K\"ahler metric.
Since the first derivative of $\Phi$ often appears in this paper, we now introduce the following proposition related to the first derivative of~$\Phi$.

\begin{Proposition}\label{prop_first_deriv_potential}
Let
$
\Phi
:=
\log
\det
B
$
be the K\"ahler potential of $G_{p,p+q}(\mathbb{C})$. Let
$
B$, $
z^{ij'}$,
\smash{$b^{i'\overline{j'}}
$}
be them defined in Section~{\rm\ref{G24_introduction}}. Then, the first derivatives of $\Phi$ are
\begin{align}
 \partial_{J}
 \Phi
 =
 \partial_{jj'}
 (
 \log
 \det
 B
 )
 =
 \overline{z}^{jl'}
 b^{j'\overline{l'}},
 \qquad
 \partial_{\overline{J}}
 \Phi
 =
 \partial_{\overline{jj'}}
 (
 \log
 \det
 B
 )
 =
 z^{jl'}
 b^{l'\overline{j'}},
 \label{partial_1st}
\end{align}
respectively.
\end{Proposition}

\begin{proof}
We shall
show
\smash{$
 \partial_{\overline{J}}
 \Phi
 =
 z^{jl'}
 b^{l'\overline{j'}}
$}
in this proof
.
The derivative of
$
\det
B
$
with respect to~$b_{\overline{k'}l'}$ is~obtained
\smash{$
 \frac{\partial}{\partial b_{\overline{k'}l'}}
 \det
 B
 =
 \widetilde{B}_{\overline{k'}l'}
 =
 \det
 B
 \cdot
 b^{l'\overline{k'}}
$}
by cofactor expansion, where $b^{l'k'}$ is the~entry of~$B^{-1}$.
By the chain rule, the derivative of
$
\det
B
$
with respect to $\overline{z}^{jj'}$ is
\begin{align}
 \partial_{\overline{jj'}}
 (
 \det
 B
 )
 &=
 \det
 B
 \cdot
 \partial_{\overline{jj'}}
 b_{\overline{k'}l'}
 \cdot
 b^{l'\overline{k'}}
=
 \det
 B
 \cdot
 z^{jl'}
 b^{l'\overline{j'}}.
 \label{det_deriv_z}
\end{align}
Here we use the fact that
$
\partial_{\overline{jj'}}
b_{\overline{k'}l'}
=
\delta_{j'
k'
}
z^{jl'}
$
from \eqref{b_entry}.
Using \eqref{det_deriv_z}, the first derivative of $\Phi$ is
\begin{align*}
 \partial_{\overline{J}}
 \Phi
 =
 \partial_{\overline{J}}
 \log
 \det
 B
 =
 (
 \det
 B
 )^{-1}
 \partial_{\overline{J}}
 (
 \det
 B
 )
 =
 (
 \det
 B
 )^{-1}
 \det
 B
 \cdot
 {z}^{jl'}
 b^{l'\overline{j'}}
 =
 z^{jl'}
 b^{l'\overline{j'}}.
\end{align*}
In a similar way,
$
 \partial_{J}
 \Phi
 =
 \overline{z}^{jl'}
 b^{j'\overline{l'}}
$
is also shown.
\end{proof}

The K\"ahler metric
on
$G_{p,p+q}(\mathbb{C})$ can be expressed by using the entries of matrices $B^{-1}$ and~$
 A^{-1}
$. Here we introduce the matrix
$
 A
 \in
 {\rm GL}_{q}
 (\mathbb{C})
$
and its inverse matrix $A^{-1}$ as follows:
\begin{align*}
 A
 =
 (
 a_{i\overline{j}}
 )
 :=
 \mathrm{Id}_{q}
 +
 Z
 Z^{\dagger},
 \qquad
 A^{-1}
 =
 \bigl(
 a^{\overline{i}j}
 \bigr).
\end{align*}

\begin{Proposition}\label{thm_metric_mtx}
Let $g_{I\overline{J}}$ be the K\"ahler metric
on $G_{p,p+q}(\mathbb{C})$. Then, K\"ahler metric
\smash{$
 g_{I\overline{J}}
 =
 g_{ii',\overline{jj'}}
$}
can be written as
\smash{$
 g_{I\overline{J}}
 =
 g_{ii',\overline{jj'}}
 =
 a^{\overline{j}i}
 b^{i'\overline{j'}}$},
where \smash{$a^{\overline{i}j}$} and \smash{$b^{i'\overline{j'}}$} are the entries of
$
 A^{-1}
$
and
$
 B^{-1}
$,
respectively.
\end{Proposition}

\begin{proof}
Since
$
 \partial_{ii'}
 \bigl(
 B^{-1}
 \bigr)
 =
 -
 B^{-1}
 (
 \partial_{ii'}
 B
 )
 B^{-1}
$,
the derivative of \smash{$b^{l'\overline{k'}}$} with respect to $z^{ii'}$ is calculated as
\begin{align}
 \partial_{ii'}
 b^{l'\overline{k'}}
 =
 -
 b^{l'\overline{m'}}
 (
 \partial_{ii'}
 b_{\overline{m'}n'}
 )
 b^{n'\overline{k'}}=
 -
 \overline{z}^{im'}
 b^{l'\overline{m'}}
 b^{i'\overline{k'}}.
 \label{Binv_deriv}
\end{align}
Hence, from \eqref{partial_1st} in Proposition~\ref{prop_first_deriv_potential} and \eqref{Binv_deriv}, we have
\begin{align}
 g_{I\overline{J}}
 =
 \partial_{ii'}
 \partial_{\overline{jj'}}
 (
 \log
 \det
 B
 )
 =
 \partial_{I}
 \bigl(
 z^{jl'}
 b^{l'\overline{j'}}
 \bigr)
 =
 b^{i'\overline{j'}}
 \bigl(
 \delta_{ji}
 -
 z^{jl'}
 b^{l'\overline{m'}}
 z^{\overline{im'}}
 \bigr).
 \label{Gpq_metric_mid}
\end{align}
On the other hand, recalling that
$
 \mathrm{Id}_{p}
 =
 BB^{-1}
 =
 \bigl(
 \mathrm{Id}_{p}
 +
 Z^{\dagger}
 Z
 \bigr)
 B^{-1}
$,
$\mathrm{Id}_{q}$ can be written as
\begin{align}
 \mathrm{Id}_{q}= \mathrm{Id}_{q} + ZZ^{\dagger} - ZZ^{\dagger}
= \mathrm{Id}_{q} + ZZ^{\dagger} - Z \bigl( \mathrm{Id}_{p} + Z^{\dagger} Z \bigr) B^{-1} Z^{\dagger} . \label{A_inv_mid}
\end{align}
Note that
\begin{align*}
 Z\bigl( \mathrm{Id}_{p} + Z^{\dagger} Z\bigr) B^{-1} Z^{\dagger} =
 Z B^{-1} Z^{\dagger} + Z Z^{\dagger} Z B^{-1} Z^{\dagger}
 =\bigl( \mathrm{Id}_{q} + Z Z^{\dagger} \bigr) Z B^{-1} Z^{\dagger}
\end{align*}
and
$ A = \mathrm{Id}_{q} + Z Z^{\dagger}$,
\eqref{A_inv_mid} can also be rewritten as
\begin{align*}
 \eqref{A_inv_mid}= A - A Z B^{-1} Z^{\dagger}= A \bigl( \mathrm{Id}_{q} - Z B^{-1} Z^{\dagger} \bigr).
\end{align*}
This means that
$ \mathrm{Id}_{q} - Z B^{-1} Z^{\dagger}$ is the inverse matrix of $A$, i.e.,
\begin{align}
 a^{\overline{j}i}
 =
 \delta_{ji}
 -
 z^{jl'}
 b^{l'\overline{m'}}
 z^{\overline{im'}}
 .
 \label{A_inv_entry}
\end{align}
Hence, substituting \eqref{A_inv_entry} into \eqref{Gpq_metric_mid}, we obtain
$
 g_{I\overline{J}}
 =
 g_{ii',\overline{jj'}}
 =
 a^{\overline{j}i}
 b^{i'\overline{j'}}
$.
\end{proof}

In addition, the second derivative of $\Phi$ with respect to $z^{I}$ (or $\bar{z}^{I}$) can be expressed as a~product of the first derivative of $\Phi$ with respect to $z^{I}$ (or $\bar{z}^{I}$).
\begin{Proposition}\label{potential_deriv_grassmann_2nd}
For $ G_{p,p+q}(\mathbb{C})$, the second derivative of the K\"ahler potential $\Phi$ satisfies
\begin{align*}
 \partial_{I}
 \partial_{J}
 \Phi = -
 \partial_{ij'}
 \Phi
 \partial_{ji'}
 \Phi,
 \qquad
 \partial_{\overline{I}}
 \partial_{\overline{J}}
 \Phi = -
 \partial_{\overline{ij'}}
 \Phi
 \partial_{\overline{ji'}} \Phi.
\end{align*}
\end{Proposition}

\begin{proof}
By using \eqref{partial_1st} and \eqref{Binv_deriv}, we obtain
\begin{align*}
 \partial_{I}
 \partial_{J}
 \Phi
 =
 \partial_{I}
 \bigl(
 \overline{z}^{jl'}
 b^{j'\overline{l'}}
 \bigr)
 =
 -
 \overline{z}^{in'}
 b^{j'\overline{n'}}
 \overline{z}^{jl'}
 b^{i'\overline{l'}}
 =
 -
 \partial_{ij'}
 \Phi
 \partial_{ji'}
 \Phi.
\end{align*}
In a similar way, we can also prove
$ \partial_{\overline{I}} \partial_{\overline{J}} \Phi = - \partial_{\overline{ij'}} \Phi \partial_{\overline{ji'}} \Phi$.
\end{proof}

In general, to explicitly determine the curvature
\smash{$
 {{R_{\overline{I}}}^{\overline{J}\overline{K}}}_{\overline{L}}
 =
 g^{\overline{J}P}
 g^{\overline{K}Q}
 R_{\overline{I}PQ\overline{L}}
$}
of K\"ahler manifolds, it is necessary to calculate the higher-order derivatives of $\Phi$.
Fortunately, the curvature of
$ G_{p,p+q}(\mathbb{C})$
can be obtained easily by Proposition~\ref{potential_deriv_grassmann_2nd}.
\begin{Proposition}\label{curv_grassmann}
For $G_{p,p+q}(\mathbb{C})$, the curvature \smash{${{R_{\overline{I}}}^{\overline{J}\overline{K}}}_{\overline{L}}$} is given by
\begin{align*}
 {{R_{\overline{I}}}^{\overline{J}\overline{K}}}_{\overline{L}}
 =
 -
 \delta_{
 \overline{il'},
 \overline{J}
 }
 \delta_{
 \overline{li'},
 \overline{K}
 }
 -
 \delta_{
 \overline{li'},
 \overline{J}
 }
 \delta_{
 \overline{il'},
 \overline{K}
 }
 =
 -
 \delta_{ij}
 \delta_{kl}
 \delta_{i'k'}
 \delta_{j'l'}
 -
 \delta_{ik}
 \delta_{jl}
 \delta_{i'j'}
 \delta_{k'l'},
\end{align*}
where $\delta_{IJ}=\delta_{ii',jj'}:=\delta_{ij}\delta_{i'j'}$.
\end{Proposition}
\begin{proof}
Recalling that for any K\"ahler manifold,
\begin{align*}
 R_{\overline{I}P\overline{L}Q}
 =
 R_{P\overline{I}Q\overline{L}}
 =
 -\partial_{P}\partial_{\overline{I}}\partial_{Q}\partial_{\overline{L}}\Phi
 +
 g^{A\overline{B}}
 (
 \partial_{P}\partial_{\overline{B}}\partial_{Q}\Phi
 )
 (
 \partial_{A}\partial_{\overline{I}}\partial_{\overline{L}}\Phi
 )
\end{align*}
(see \cite[Appendix A]{OS_jrnl}),
we get
\begin{align}
 R_{\overline{I}P\overline{L}Q}
& =
 -
 \partial_{P}\partial_{Q}
 (
 \partial_{\overline{I}}\partial_{\overline{L}}\Phi
 )
 +
 g^{A\overline{B}}
 \{
 \partial_{\overline{B}}
 (
 \partial_{P}\partial_{Q}\Phi
 )
 \}
 \{
 \partial_{A}
 (
 \partial_{\overline{I}}\partial_{\overline{L}}\Phi
 )
 \}
 \nonumber
 \\
 &=
 -
 \partial_{P}\partial_{Q}
 (
 -
 \partial_{\overline{il'}}\Phi
 \partial_{\overline{li'}}\Phi
 )
 +
 g^{A\overline{B}}
 \{
 \partial_{\overline{B}}
 (
 -
 \partial_{pq'}\Phi
 \partial_{qp'}\Phi
 )
 \}
 \{
 \partial_{A}
 (
 -
 \partial_{\overline{il'}}\Phi
 \partial_{\overline{li'}}\Phi
 )
 \}
 \label{curv_calc_mid1}
\end{align}
by Proposition~\ref{potential_deriv_grassmann_2nd}.
By calculating the derivative in \eqref{curv_calc_mid1}, it can be rewritten as
\begin{gather}
 \partial_{\overline{il'}}
 (
 \partial_{P}\partial_{Q}\Phi
 )
 \partial_{\overline{li'}}\Phi
 +
 g_{Q,\overline{il'}}
 g_{P,\overline{li'}}
 +
 g_{P,\overline{il'}}
 g_{Q,\overline{li'}}
 +
 \partial_{\overline{il'}}\Phi
 \partial_{\overline{li'}}
 (
 \partial_{P}\partial_{Q}\Phi
 ) +
 g^{A\overline{B}}
 \bigl(
 g_{\overline{B},pq'}
 g_{A,\overline{il'}}
 \partial_{qp'}\Phi
 \partial_{\overline{li'}}\Phi
 \nonumber
 \\
 \qquad
 {}+
 g_{\overline{B},pq'}
 g_{A,\overline{li'}}
 \partial_{qp'}\Phi
 \partial_{\overline{il'}}\Phi
 +
 g_{\overline{B},qp'}
 g_{A,\overline{il'}}
 \partial_{pq'}\Phi
 \partial_{\overline{li'}}\Phi
 +
 g_{\overline{B},qp'}
 g_{A,\overline{li'}}
 \partial_{pq'}\Phi
 \partial_{\overline{il'}}\Phi
 \bigr)
 \nonumber
 \\
 \phantom{\qquad
 +}{}=
 \partial_{\overline{il'}}
 (
 -
 \partial_{pq'}\Phi
 \partial_{qp'}\Phi
 )
 \partial_{\overline{li'}}\Phi
 +
 g_{Q,\overline{il'}}
 g_{P,\overline{li'}}
 +
 g_{P,\overline{il'}}
 g_{Q,\overline{li'}}
 +
 \partial_{\overline{il'}}\Phi
 \partial_{\overline{li'}}
 (
 -
 \partial_{pq'}\Phi
 \partial_{qp'}\Phi
 )
 \nonumber
 \\
 \phantom{\qquad
 +=}{}
 +
 \delta^{A}_{pq'}
 g_{A,\overline{il'}}
 \partial_{qp'}\Phi
 \partial_{\overline{li'}}\Phi
 +
 \delta^{A}_{pq'}
 g_{A,\overline{li'}}
 \partial_{qp'}\Phi
 \partial_{\overline{il'}}\Phi
 +
 \delta^{A}_{qp'}
 g_{A,\overline{il'}}
 \partial_{pq'}\Phi
 \partial_{\overline{li'}}\Phi \nonumber
 \\
 \phantom{\qquad
 +=}{}
 +
 \delta^{A}_{qp'}
 g_{A,\overline{li'}}
 \partial_{pq'}\Phi
 \partial_{\overline{il'}}\Phi.
 \label{curv_calc_mid2}
\end{gather}
Here we use Proposition~\ref{potential_deriv_grassmann_2nd} and
$
 g_{\overline{I}J}
 =
 \partial_{\overline{I}}
 \partial_{J}
 \Phi
$,
then \eqref{curv_calc_mid2} becomes
\begin{gather}
 -
 g_{\overline{il'},pq'}
 \partial_{qp'}\Phi
 \partial_{\overline{li'}}\Phi
 -
 g_{\overline{il'},qp'}
 \partial_{pq'}\Phi
 \partial_{\overline{li'}}\Phi
 +
 g_{Q,\overline{il'}}
 g_{P,\overline{li'}}
 +
 g_{P,\overline{il'}}
 g_{Q,\overline{li'}}
 -
 g_{\overline{li'},pq'}
 \partial_{qp'}\Phi
 \partial_{\overline{il'}}\Phi
 \nonumber
 \\
 \qquad
 {}-
 g_{\overline{li'},qp'}
 \partial_{pq'}\Phi
 \partial_{\overline{il'}}\Phi
 +
 g_{\overline{il'},pq'}
 \partial_{qp'}\Phi
 \partial_{\overline{li'}}\Phi
 +
 g_{\overline{li'},pq'}
 \partial_{qp'}\Phi
 \partial_{\overline{il'}}\Phi
 +
 g_{\overline{il'},qp'}
 \partial_{pq'}\Phi
 \partial_{\overline{li'}}\Phi \nonumber
 \\
 \qquad
 {}+
 g_{\overline{li'},qp'}
 \partial_{pq'}\Phi
 \partial_{\overline{il'}}\Phi
 \nonumber
 \\
 \phantom{
 \qquad+}{}=
 g_{P,\overline{il'}}
 g_{Q,\overline{li'}}
 +
 g_{Q,\overline{il'}}
 g_{P,\overline{li'}}.
 \label{curv_calc_mid3}
\end{gather}
Substituting \eqref{curv_calc_mid3} into
$
 {{R_{\overline{I}}}^{\overline{J} \overline{K}}}_{\overline{L}}
 =
 g^{\overline{J}P}
 g^{\overline{K}Q}
 R_{\overline{I}PQ\overline{L}}
 =
 -
 g^{\overline{J}P}
 g^{\overline{K}Q}
 R_{\overline{I}P\overline{L}Q}
$,
we obtain
\begin{align*}
 {{R_{\overline{I}}}^{\overline{J} \overline{K}}}_{\overline{L}}
 &=
 -
 g^{\overline{J}P}
 g^{\overline{K}Q}
 (
 g_{P,\overline{il'}}
 g_{Q,\overline{li'}}
 +
 g_{Q,\overline{il'}}
 g_{P,\overline{li'}}
 )
=
 -
 \delta_{
 \overline{il'},
 \overline{J}
 }
 \delta_{
 \overline{li'},
 \overline{K}
 }
 -
 \delta_{
 \overline{li'},
 \overline{J}
 }
 \delta_{
 \overline{il'},
 \overline{K}
 }
 \\
 &=
 -
 \delta_{ij}
 \delta_{kl}
 \delta_{i'k'}
 \delta_{j'l'}
 -
 \delta_{ik}
 \delta_{jl}
 \delta_{i'j'}
 \delta_{k'l'}.
\end{align*}
This proof is completed.
\end{proof}

Note that, due to the difference with the notation \eqref{Riem_curv_KN} for the Riemann curvature tensor, \smash{${{R_{\overline{I}}}^{\overline{J} \overline{K}}}_{\overline{L}}$} has opposite sign to \smash{${{\mathfrak{R}_{\overline{I}}}^{\overline{J} \overline{K}}}_{\overline{L}}$} in the notation by Kobayashi--Nomizu, i.e.,
\[
{{R_{\overline{I}}}^{\overline{J} \overline{K}}}_{\overline{L}}
= -{{\mathfrak{R}_{\overline{I}}}^{\overline{J} \overline{K}}}_{\overline{L}}.
\] In addition, this curvature has the symmetry
\[
 {{R_{\overline{I}}}^{\overline{J} \overline{K}}}_{\overline{L}}
 =
 {{R_{\overline{L}}}^{\overline{J} \overline{K}}}_{\overline{I}}
 =
 {{R_{\overline{I}}}^{\overline{K} \overline{J}}}_{\overline{L}}
 =
 {{R_{\overline{L}}}^{\overline{K} \overline{J}}}_{\overline{I}}
\] with respect to capital letter indices $I$, $J$, $K$ and $L$. See \cite{HS1,OS_jrnl} for more details.

\section{A property of symmetric functions of capital letter indices}
\label{appendix_symmetry_function}

We denote a property of symmetric functions whose variables are capital letter indices in this appendix. By using this property, we state that some functions and operators appearing in this paper do not depend on the choice of a capital letter index $I$. Finally, we show from this fact that the star product~\eqref{star_prd_G2_4} in our main theorem (see Theorem~\ref{thm_star_prd_G2_4}) which is independent of $I$.

We now introduce the following proposition related to the functions which depends on capital letter indices.
\begin{Proposition}\label{prop_app_sym_fct}
Let
$
f\colon
\mathcal{I}
\times
\mathcal{I}
\to
\mathbb{R}
$
be a function satisfying
$
f(I,J)
=
f(J,I)
$
for
$
I,J
\in
\mathcal{I}
=
\{
I,
\cancel{I},
i\cancel{i}',
\cancel{i}i'
\}
$.
Then,
$
f(I,\cancel{i}i')
f(\cancel{I},i\cancel{i}')
$
does not depend on
$
I
=
ii'
\in
\mathcal{I}
$.
\end{Proposition}

\begin{proof}
From the assumption of $f$, the following relations hold
$
f(I,\cancel{i}i')
=
f(\cancel{i}i',I)$,
$
f(\cancel{I},i\cancel{i}')
=
f(i\cancel{i}',\cancel{I})$.
Since the product of
$
f(I,\cancel{i}i')
$
and
$
f(\cancel{I},i\cancel{i}')
$
is commutative, we eventually obtain
\begin{align*}
f(I,\cancel{i}i')
f(\cancel{I},i\cancel{i}')
=
f(\cancel{I},i\cancel{i}')
f(I,\cancel{i}i')
=
f(\cancel{i}i',I)
f(i\cancel{i}',\cancel{I})
=
f(i\cancel{i}',\cancel{I})
f(\cancel{i}i',I).
\end{align*}
This shows that
$
f(I,\cancel{i}i')
f(\cancel{I},i\cancel{i}')
$
does not depend on the choice of a capital letter index $I\in\mathcal{I}$.
\end{proof}

This property can be extended to commutative operators on the Fock space $V$ as well as symmetric functions on $\mathbb{R}$.
\begin{Proposition}\label{cor_app_sym_opr}
Let
$
\mathfrak{f}\colon
\mathcal{I}
\times
\mathcal{I}
\to
\mathcal{L}(V)
$
be an operator satisfying
$
\mathfrak{f}(I,J)
=
\mathfrak{f}(J,I)
$
for
$
I,J
\in
\mathcal{I}
$,
where $\mathcal{L}(V)$ is the set of operators on the Fock space $V$.
If $\mathfrak{f}$ is commutative, i.e.,
$
\mathfrak{f}(I,J)
\circ
\mathfrak{f}(K,L)
=
\mathfrak{f}(K,L)
\circ
\mathfrak{f}(I,J)
$
for $I,J,K,L\in\mathcal{I}$, then
$
\mathfrak{f}(I,\cancel{i}i')
\mathfrak{f}(\cancel{I},i\cancel{i}')
$
does not depend on
$
I
=
ii'
\in
\mathcal{I}
$.
\end{Proposition}

By Proposition~\ref{cor_app_sym_opr}, the operators
\begin{align*}
 (
 N_{I}
 +
 N_{\cancel{i}i'}
 )
 (
 N_{\cancel{I}}
 +
 N_{i\cancel{i}'}
 ),
 \qquad
 (
 N_{I}
 +
 N_{\cancel{i}i'}
 -
 \Delta_{
 I,
 \cancel{i}i',
 l,
 \{
 J_{i}
 \}_{n}
 }
 )
 (
 N_{\cancel{I}}
 +
 N_{i\cancel{i}'}
 -
 \Delta_{
 \cancel{I},
 i\cancel{i}',
 l,
 \{
 J_{i}
 \}_{n}
 }
 )
\end{align*}
do not depend on the choice of $I\in\mathcal{I}$, where
\smash{$
\Delta_{
I,
\cancel{i}i',
l,
\{
J_{i}
\}_{n}
}
$},
\smash{$
\Delta_{
\cancel{I},
i\cancel{i}',
l,
\{
J_{i}
\}_{n}
}
$}
are given by \eqref{notation_delta_Iisi}, \eqref{notation_delta_Isiis}. Then, it is shown that
$
\mathcal{F}_{
l,
\{
 J_{i}
\}_{n},
\{
 k_{i}
\}_{n}
}
$
in Section~\ref{G24_recrel_solve} is independent of
$
I
\in
\mathcal{I}
$.
Furthermore, by Proposition~\ref{prop_app_sym_fct}, the following functions do not depend on $I\in\mathcal{I}$:
\begin{gather*}
 (
\beta_{I}^{n}
+
\beta_{\cancel{i}i'}^{n}
 )
 (
\beta_{\cancel{I}}^{n}
+
\beta_{i\cancel{i}'}^{n}
 ),
\qquad
 (
\beta_{I}^{n}
+
\beta_{\cancel{i}i'}^{n}
-
\Delta_{
I,
\cancel{i}i',
l,
\{
J_{i}
\}_{n}
}
 )
 (
\beta_{\cancel{I}}^{n}
+
\beta_{i\cancel{i}'}^{n}
-
\Delta_{
\cancel{I},
i\cancel{i}',
l,
\{
J_{i}
\}_{n}
}
 ),
\\
\left\{
\sum_{m=1}^{l}
 (
\delta_{IJ_{m}}
+
\delta_{\cancel{i}i',J_{m}}
 )
\right\}
\left\{
\sum_{m=1}^{l}
 (
\delta_{\cancel{I}J_{m}}
+
\delta_{i\cancel{i}',J_{m}}
 )
\right\}.
\end{gather*}
From this fact,
$
\Upsilon_{
l,
\{J_{i}\}_{n},
\{k_{i}\}_{n}
}
$
in Theorem~\ref{thm_star_prd_G2_4} (our main theorem) and \eqref{star_coeff_G2_4}
do not depend on the choice of
$
I
\in
\mathcal{I}
$.
Hence, the star product with separation of variables \eqref{star_prd_G2_4} on $G_{2,4}(\mathbb{C})$ in Theorem~\ref{thm_star_prd_G2_4} (our main theorem) does not depend on a capital letter index $I$.

\section[Calculations of T\string^n in Section 3.3]{Calculations of $\boldsymbol{T^{n}_{\protect\vec{\alpha}_{n},\protect\vec{\beta}_{n}^{\ast}}}$ in Section~\ref{G24_recrel_solve}}
\label{appendix_coeff_T}

In this appendix, we denote the calculation of \smash{$T\raisebox{2pt}{${}^{n}_{\vec{\alpha}_{n},\vec{\beta}_{n}^{\ast}}$}$} using
\smash{$
 T\raisebox{2pt}{${}^{n}_{\vec{\alpha}_{n},\vec{\beta}_{n}^{\ast}}$}
 =
\langle
 \vec{\alpha}_{n}
|
 T_{n}
|
 \vec{\beta}_{n}^{\ast}
\rangle
$}
in
Theorem~\ref{thm_G24_opr}. Since
\smash{$
\mathcal{C}_{
l,
\{
J_{i}
\}_{n}
,
\{
k_{i}
\}_{n}
}
$}
and
\smash{$
\mathcal{F}_{
l,
\{
J_{i}
\}_{n}
,
\{
k_{i}
\}_{n}
}
$}
depend on only number operators, then
\smash{$
\mathcal{C}_{
l,
\{
J_{i}
\}_{n}
,
\{
k_{i}
\}_{n}
}
|
\vec{\beta}_{n}^{\ast}
\rangle
$}
and~\smash{$
\mathcal{F}_{
l,
\{
J_{i}
\}_{n}
,
\{
k_{i}
\}_{n}
}
|
\vec{\beta}_{n}^{\ast}
\rangle
$}
are immediately calculated as
\begin{gather}
\mathcal{C}_{
l,
\{
J_{i}
\}_{n}
,
\{
k_{i}
\}_{n}
}
|
\vec{\beta}_{n}^{\ast}
\rangle
=
 (
 \tau_{l}
 \delta_{j_{l}k_{l}}
 +
 \beta_{j_{l}\cancel{j_{l}}'}^{n}
 -
 \Lambda_{
 l,
 j_{l}\cancel{j_{l}}',
 \{
 J_{i}
 \}_{n}
 ,
 \{
 k_{i}
 \}_{n}
 }
 +
 1
 )
|
\vec{\beta}_{n}^{\ast}
\rangle,
\label{act_C}
\\
\mathcal{F}_{
l,
\{
J_{i}
\}_{n}
,
\{
k_{i}
\}_{n}
}
|
\vec{\beta}_{n}^{\ast}
\rangle\nonumber\\
\qquad
=
\{
 l
 (
 \tau_{l}
 +
 1
 )
 +
 2
 (
 \beta_{I}^{n}
 +
 \beta_{\cancel{i}i'}^{n}
 -
 \Delta_{
 I,
 \cancel{i}i',
 l,
 \{
 J_{i}
 \}_{n}
 }
 )
 (
 \beta_{\cancel{I}}^{n}
 +
 \beta_{i\cancel{i}'}^{n}
 -
 \Delta_{
 \cancel{I},
 i\cancel{i}',
 l,
 \{
 J_{i}
 \}_{n}
 }
 )
\}^{-1}
|
\vec{\beta}_{n}^{\ast}
\rangle,
\label{act_F}
\end{gather}
respectively.
Here we use \eqref{notation_Lambda}, \eqref{notation_delta_Iisi} and \eqref{notation_delta_Isiis}.
By using \eqref{act_C} and \eqref{act_F}, we obtain
\begin{gather}
 \mathcal{C}_{
 1,
 \{
 J_{i}
 \}_{n}
 ,
 \{
 k_{i}
 \}_{n}
 }
 \cdots
 \mathcal{C}_{
 n,
 \{
 J_{i}
 \}_{n}
 ,
 \{
 k_{i}
 \}_{n}
 }
 \mathcal{F}_{
 1,
 \{
 J_{i}
 \}_{n}
 ,
 \{
 k_{i}
 \}_{n}
 }
 \cdots
 \mathcal{F}_{
 n,
 \{
 J_{i}
 \}_{n}
 ,
 \{
 k_{i}
 \}_{n}
 }
|
 \vec{\beta}_{n}^{\ast}
\rangle
 \nonumber
 \\
\qquad
 =
 \left\{
 \prod_{l=1}^{n}
 \frac{
 \tau_{l}
 \delta_{j_{l}k_{l}}
 +
 \beta_{j_{l}\cancel{j_{l}}'}^{n}
 +
 1
 -
 \Lambda_{
 l,
 j_{l}\cancel{j_{l}}',
 \{
 J_{i}
 \}_{n}
 ,
 \{
 k_{i}
 \}_{n}
 }
 }{
 l
 (
 \tau_{l}+1
 )
 +
 2
 (
 \beta_{I}^{n}
 +
 \beta_{\cancel{i}i'}^{n}
 -
 \Delta_{
 I,
 \cancel{i}i',
 l,
 \{
 J_{i}
 \}_{n}
 }
 )
 (
 \beta_{\cancel{I}}^{n}
 +
 \beta_{i\cancel{i}'}^{n}
 -
 \Delta_{
 \cancel{I},
 i\cancel{i}',
 l,
 \{
 J_{i}
 \}_{n}
 }
 )
 }
 \right\}
|
 \vec{\beta}_{n}^{\ast}
\rangle . \label{act_CandF}
\end{gather}
Next, we calculate
$
\mathcal{A}_{
J_{1},
k_{1}
}
\cdots
\mathcal{A}_{
J_{n},
k_{n}
}
$
acting on
\smash{$
|
\vec{\beta}_{n}^{\ast}
\rangle
$}.
Recalling \eqref{braket_ann_step_fct}, we have
\begin{align*} 
 \mathcal{A}_{ J_{n}, k_{n} }
|
 \vec{\beta}_{n}^{\ast}
\rangle
 ={}&
 a_{J_{n}}
 \frac{1}{\sqrt{N_{J_{n}}}}
 \left(
 a_{\cancel{J_{n}}}
 \frac{1}{
 \sqrt{
 N_{\cancel{J_{n}}}
 }
 }
 a_{j_{n}\cancel{j_{n}}'}^{\dagger}
 \frac{1}{
 \sqrt{
 N_{j_{n}\cancel{j_{n}}'}+1
 }
 }
 \right)^{\delta_{\cancel{j_{n}}k_{n}}}
|
 \vec{\beta}_{n}^{\ast}
\rangle
 \\
 ={}&
 \bigg(
 \prod_{
 S
 \in
 \mathcal{I}
 }
 \theta
 (
 \beta_{S}^{n}
 -
 \delta_{SJ_{n}}
 -
 \delta_{\cancel{j_{n}}k_{n}}
 (
 \delta_{S\cancel{J_{n}}}
 -
 \delta_{S,j_{n}\cancel{j_{n}}'}
 )
 )
 \bigg)\\
&\times
 \big|
 \vec{\beta}_{n}^{\ast}
 -
 \vec{e}_{J_{n}}^{\ast}
 -
 \delta_{\cancel{j_{n}}k_{n}}
 \bigl(
 \vec{e}_{\cancel{J_{n}}}^{\ast}
 -
 \vec{e}_{j_{n}\cancel{j_{n}}'}^{\ast}
 \bigr)
 \big\rangle
 .
\end{align*}
Doing a similar calculation $(n-1)$ times,
\smash{$
 \mathcal{A}_{
 1,
 J_{1},
 k_{1}
 }
 \cdots
 \mathcal{A}_{
 n,
 J_{n},
 k_{n}
 }
|
 \vec{\beta}_{n}^{\ast}
\rangle
$}
is calculated as
\begin{gather}
 \mathcal{A}_{
 1,
 J_{1},
 k_{1}
 }
 \cdots
 \mathcal{A}_{
 n,
 J_{n},
 k_{n}
 }
|
 \vec{\beta}_{n}^{\ast}
\rangle\nonumber\\
 \qquad
 =
 \left(
 \prod_{
 S
 \in
 \mathcal{I}
 }
 \prod_{r=1}^{n}
 \theta
 \left(
 \beta_{S}^{n}
 -
 \sum_{m=r}^{n}
 d_{
 S,
 J_{m},
 k_{m}
 }
 \right)
 \right)
 \left|\,
 \vec{\beta}_{n}^{\ast}
 -
 \sum_{
 X
 \in
 \mathcal{I}
 }
 \sum_{m=1}^{n}
 d_{
 X,
 J_{m},
 k_{m}
 }
 \vec{e}_{X}^{\ast}
 \right\rangle
 ,
 \label{act_A}
\end{gather}
where we use \eqref{notation_d}, i.e.,
$
d_{
 S,
 J_{m},
 k_{m}
}
=
\delta_{
 S,
 J_{m}
}
+
\delta_{\cancel{j_{m}}k_{m}}
 (
 \delta_{
 S,
 \cancel{J_{m}}
 }
 -
 \delta_{
 S,
 j_{m}\cancel{j_{m}}'
 }
 )
$.
From \eqref{act_CandF} and \eqref{act_A}, we obtain
\begin{gather}
 \frac{
 g_{\overline{k_{1}{j_{1}}'},D_{1}}
 \cdots
 g_{\overline{k_{n}{j_{n}}'},D_{n}}
 }{
 \tau_{1}
 \cdots
 \tau_{n}
 }
 \mathcal{A}_{
 1,
 J_{1},
 k_{1}
 }
 \cdots
 \mathcal{A}_{
 n,
 J_{n},
 k_{n}
 }
 \mathcal{C}_{
 1,
 \{
 J_{i}
 \}_{n}
 ,
 \{
 k_{i}
 \}_{n}
 }
 \cdots
 \mathcal{C}_{
 n,
 \{
 J_{i}
 \}_{n}
 ,
 \{
 k_{i}
 \}_{n}
 }
 \mathcal{F}_{
 1,
 \{
 J_{i}
 \}_{n}
 ,
 \{
 k_{i}
 \}_{n}
 }\nonumber\\
 \qquad
 \cdots
 \mathcal{F}_{
 n,
 \{
 J_{i}
 \}_{n}
 ,
 \{
 k_{i}
 \}_{n}
 }
 |
 \vec{\beta}_{n}^{\ast}
 \rangle
 \nonumber
 \\
 \qquad\phantom{\cdots}{}=
 \left(
 \prod_{
 S
 \in
 \mathcal{I}
 }
 \prod_{r=1}^{n}
 \theta
 \left(
 \beta_{S}^{n}
 -
 \sum_{m=r}^{n}
 d_{
 S,
 J_{m},
 k_{m}
 }
 \right)
 \right)
 \left(
 \prod_{l=1}^{n}
 \frac{g_{\overline{k_{l}j_{l}'},D_{l}}}{\tau_{l}}
 \right)\nonumber
 \\
 \qquad\phantom{\cdots=}{}
 \times
 \left\{
 \prod_{l=1}^{n}
 \frac{
 \tau_{l}
 \delta_{j_{l}k_{l}}
 +
 \beta_{j_{l}\cancel{j_{l}}'}^{n}
 +
 1
 -
 \Lambda_{
 l,
 j_{l}\cancel{j_{l}}',
 \{
 J_{i}
 \}_{n}
 ,
 \{
 k_{i}
 \}_{n}
 }
 }{
 l
 (
 \tau_{l}+1
 )
 +
 2
 (
 \beta_{I}^{n}
 +
 \beta_{\cancel{i}i'}^{n}
 -
 \Delta_{
 I,
 \cancel{i}i',
 l,
 \{
 J_{i}
 \}_{n}
 }
 )
 (
 \beta_{\cancel{I}}^{n}
 +
 \beta_{i\cancel{i}'}^{n}
 -
 \Delta_{
 \cancel{I},
 i\cancel{i}',
 l,
 \{
 J_{i}
 \}_{n}
 }
 )
 }
 \right\}\nonumber\\
 \qquad\phantom{\cdots=}{}\times
 \left|
 \vec{\beta}_{n}^{\ast}
 -
 \sum_{
 X
 \in
 \mathcal{I}
 }
 \sum_{m=1}^{n}
 d_{
 X,
 J_{m},
 k_{m}
 }
 \vec{e}_{X}^{\ast}
 \right\rangle
 ,
 \label{act_ACF}
\end{gather}
where we use \eqref{notation_Lambda}--\eqref{notation_delta_Isiis}.
Recalling that Proposition~\ref{prop:covar_const_0and1} and \eqref{braket_coeff}, $T_{0}$ is immediately obtained as
\smash{$
T_{0}
=
|
\vec{0}
\rangle
\langle
\vec{0}
|
$}.
Thus
\smash{$
\langle
 \vec{\alpha}_{n}
|
 a_{D_{n}}^{\dagger}
 \frac{1}{
 \sqrt{
 N_{D_{n}}
 +
 1
 }
 }
 \cdots
 a_{D_{1}}^{\dagger}
 \frac{1}{
 \sqrt{
 N_{D_{1}}
 +
 1
 }
 }
 T_{0}
$}
is easily calculated as
\begin{align}
\langle
 \vec{\alpha}_{n}
|
 a_{D_{n}}^{\dagger}
 \frac{1}{
 \sqrt{
 N_{D_{n}}
 +
 1
 }
 }
 \cdots
 a_{D_{1}}^{\dagger}
 \frac{1}{
 \sqrt{
 N_{D_{1}}
 +
 1
 }
 }
 T_{0}
 =
\Big\langle
 \vec{\alpha}_{n}
|
 \sum_{m=1}^{n}
 \vec{e}_{D_{m}}
 \Big\rangle
 \langle
 \vec{0}
 |
 =
 \delta_{
 \vec{\alpha}_{n},
 \sum_{m=1}^{n}
 \vec{e}_{D_{m}}
 }
 \langle
 \vec{0}
 | .
 \label{act_a_opr}
\end{align}
Hence, substituting \eqref{act_ACF} and \eqref{act_a_opr} into \eqref{opr_T_norm}, we eventually obtain the explicit expression of~\smash{$T^{n}_{\vec{\alpha}_{n},\vec{\beta}_{n}^{\ast}}$}
\begin{align*}
 T^{n}_{\vec{\alpha}_{n},\vec{\beta}_{n}^{\ast}}
 ={}&
\langle
 \vec{\alpha}_{n}
|
 T_{n}
|
 \vec{\beta}_{n}^{\ast}
\rangle
 =
 \sum_{
 \substack{
 J_{i}
 \in
 \{J_{i}\}_{n}
 \\
 D_{i}
 \in
 \{D_{i}\}_{n}
 }
 }
 \sum_{
 \substack{
 k_{i}
 =
 1
 \\
 k_{i}
 \in
 \{k_{i}\}_{n}
 }
 }^{2}
 \delta_{
 \vec{\alpha}_{n},
 \sum_{m=1}^{n}
 \vec{e}_{D_{m}}
 }
 \delta_{
 \vec{\beta}_{n}^{\ast},
 \sum_{
 X
 \in
 \mathcal{I}
 }
 \sum_{m=1}^{n}
 d_{
 X,
 J_{m},
 k_{m}
 }
 \vec{e}_{X}^{\ast}
 }
 \nonumber
 \\
 &
 \times
 \left(
 \prod_{
 S
 \in
 \mathcal{I}
 }
 \prod_{r=1}^{n}
 \theta
 \left(
 \beta_{S}^{n}
 -
 \sum_{m=r}^{n}
 d_{
 S,
 J_{m},
 k_{m}
 }
 \right)
 \right)
 \left(
 \prod_{l=1}^{n}
 \frac{g_{\overline{k_{l}j_{l}'},D_{l}}}{\tau_{l}}
 \right)
 \\
 &
 \times
 \left\{
 \prod_{l=1}^{n}
 \frac{
 \tau_{l}
 \delta_{j_{l}k_{l}}
 +
 \beta_{j_{l}\cancel{j_{l}}'}^{n}
 +
 1
 -
 \Lambda_{
 l,
 j_{l}\cancel{j_{l}}',
 \{
 J_{i}
 \}_{n},
 \{
 k_{i}
 \}_{n}
 }
 }{
 l
 (
 \tau_{l}+1
 )
 +
 2
 (
 \beta_{I}^{n}
 +
 \beta_{\cancel{i}i'}^{n}
 -
 \Delta_{
 I,
 \cancel{i}i',
 l,
 \{
 J_{i}
 \}_{n}
 }
 )
 (
 \beta_{\cancel{I}}^{n}
 +
 \beta_{i\cancel{i}'}^{n}
 -
 \Delta_{
 \cancel{I},
 i\cancel{i}',
 l,
 \{
 J_{i}
 \}_{n}
 }
 )
 }
 \right\}
 .
\end{align*}

\subsection*{Acknowledgements}

A.S.\ was supported by JSPS KAKENHI Grant Number 21K03258. The authors are grateful to Masashi Hamanaka, Shota Hamanaka, Yasushi Homma, Noriaki Ikeda, Taichiro Kugo, Thomas Raujouan, Hiroshi Tamaru and Masashi Yasumoto for useful advice and comments in ``Mini-School on Differential Geometry and Integrable Systems'', ``Poisson geometry and related topics~23'', ``The 4th International Conference on Surfaces, Analysis, and Numerics in Differential Geometry'', ``Strings and Fields 2024'' and our private discussion at Nagoya University. The authors thank the Yukawa Institute for Theoretical Physics at Kyoto University. Discussions during the YITP workshop YITP-W-24-08 on ``Strings and Fields 2024'' were useful to complete this work. The authors sincerely appreciate the anonymous referees, their thoughtful feedback on errors and typos contributed greatly to improving the quality of this paper.

\pdfbookmark[1]{References}{ref}
\LastPageEnding

\end{document}